%% file: main.tex
\documentclass{article}

\usepackage{color}
\usepackage{commands}
\usepackage{amsthm,amssymb}
\usepackage{algorithm}
\usepackage[noend]{algpseudocode}

\usepackage{enumitem}

\usepackage{thm-restate}
\usepackage{calc}

\algrenewcommand\textproc{}

\usepackage{graphicx}
\usepackage{amsthm}
\usepackage{amssymb}
\usepackage{amsmath}
\usepackage{mathtools}
\usepackage{caption} 
\usepackage[caption=false,subrefformat=parens,labelformat=parens]{subfig} 
\usepackage{listings} 
\usepackage{multicol} 

\usepackage{ebproof}

\theoremstyle{definition}
\newtheorem{definition}{Definition}
\newtheorem{example}{Example}
\newtheorem{theorem}{Theorem}
\newtheorem{lemma}[theorem]{Lemma}
\newtheorem{case}{Case}[lemma]

\newcommand{\myqed}{\hfill\blacksquare}

\usepackage[english]{babel}

\usepackage[a4paper,top=2cm,bottom=2cm,left=2cm,right=2cm,marginparwidth=1.75cm]{geometry}

\usepackage{amsmath}
\usepackage{graphicx}
\usepackage[colorlinks=true, allcolors=blue]{hyperref}

\title{Reversing an Imperative Concurrent Programming Language}
\author{James Hoey, Irek Ulidowski \\ {\small School of Computing and Mathematical Sciences, University of Leicester, Leicester, UK}}
\date{\vspace{-20pt}}

\begin{document}

\lstdefinelanguage{lang}{
    alsodigit = {-},
    alsoletter={\{},
    keywords = {begin,var,proc,end,while,do,if,then,call,else,remove,par},
}

\definecolor{codegray}{rgb}{0.5,0.5,0.5}
\definecolor{backcolour}{rgb}{0.95,0.95,0.92}

\lstdefinestyle{mystyle}{
	language=lang,
	keywordstyle=\color{color:keyword},
    backgroundcolor=\color{backcolour},   
    numberstyle=\small\color{codegray},
    basicstyle=\footnotesize,
    breakatwhitespace=false,         
    breaklines=true,                 
    captionpos=b,                    
    keepspaces=true,                 
    numbers=left,                    
    numbersep=5pt,                  
    showspaces=false,                
    showstringspaces=false,
    showtabs=false,                  
    tabsize=2
}
 
\lstset{style=mystyle}

\newtheorem{prop}[theorem]{Proposition}

\maketitle

\begin{abstract}
We introduce a method of reversing the execution of imperative concurrent programs. Given an irreversible program, we describe the process of producing two versions. The first performs forward execution and saves information necessary for reversal. The second uses this saved information to simulate reversal. We propose using identifiers to overcome challenges of reversing concurrent programs. We prove this reversibility to be correct, showing that the initial program state is restored and that all saved information is used (garbage-free). 
\end{abstract}

\section{Introduction} \label{sec:intro}
\input{Introduction.tex}

\section{Related Work} \label{sec:related-work}
\input{Related-work.tex}

\section{Imperative Concurrent Programming Language and Program State} \label{sec:lang-state}
\input{LanguageState}

\section{Annotation} \label{sec:ann}
\input{Annotation}

\section{Inversion} \label{sec:inv}
\input{Inversion}

\section{Simulation} \label{sec:implem}
\input{Simulation}

\section{Conclusion} \label{sec:conc}
\input{Conclusion.tex}

\bibliographystyle{plain}
\bibliography{bib-fixed}

\clearpage
\appendix

\section{Complete Function Definitions} \label{complete-func}
\input{Full-functions.tex}

\clearpage

\section{Additional Operational Semantics} \label{appen-rules}
\input{Rules-extended.tex}

\section{Statement Property} \label{appen-sp}
\input{sp-proof.tex}

\section{Program Property} \label{appen-pp}
\input{pp-proof.tex}

\end{document}

%% file: Introduction.tex
Traditional computation is fundamentally irreversible, starting in an initial state and reaching a final state. Typically the initial store cannot be restored as the computation cannot be performed in reverse due to the loss of information throughout forward execution. Reversible computation, an interesting and advancing field of research, aims to propose solutions to this challenge and allow computation from a final state back to the corresponding initial state.   

Executions of traditional programming languages are irreversible. Widely used languages such as C or Java execute correctly and efficiently in the traditional forward direction, but do not directly support reversal. Consider a simple program that contains only the destructive assignment \code{X = 5}, where \code{X} is initially \code{2}. Forward execution of this program overwrites the value of the variable, discarding the previous value \code{2}  (typically this lost information is dissipated as heat~\cite{RL1961}). Now imagine we wish to reverse this execution (a possible reason being we optimistically performed this execution under assumptions that are no longer true). To correctly restore to the previous state, the variable must hold its previous value. With no way of re-calculating this old value, reversal of this execution is not possible. Had the assignment been \code{X += 5}, a constructive update that increments the value of a variable, the previous value can easily be re-calculated via a decrement of \code{5}. Therefore constructive assignments are reversible provided the expression evaluates during reversal to the same value.

There are a number of motivations for continued research into reversible computation. Traditional execution loses information as energy, as shown by the Landauer principle that states that erasing 1 bit of information costs at least $kT\ln2$ in energy (where $K$ is the Boltzmann constant and $T$ is the temperature of the heat sink)~\cite{RL1961}. With reversibility requiring no information to be lost, there is the potential for more energy efficient computation. In 2018, Frank corrected some common misconceptions regarding this principle \cite{MF2018}. Reversible computation is also desirable due to the potential applications it supports. This includes to debugging, where reversing an execution experiencing a bug can avoid the difficultly of reproducing a bug in a concurrent, or multi-threaded, program \cite{IL2018,JH2019,JE2012}, and optimistic Parallel Discrete Event Simulation (PDES), where events can be optimistically executed as soon as they are possible and rolled back, or reversed, if this is later found to be incorrect \cite{MS2018,GV2011,DC2016,DC2017}. Additional detail is given in Section~\ref{sec:related-work}. More recent work has discussed the link between reversibility and robotics~\cite{IL2021}, and reversible circuits and logic gates~\cite{ADV2010}. Important background and further information regarding all applications of reversible computation can be found through the European Cooperation in Science and Technology (COST) project IC1405 Reversible Computation - Extending Horizons of Computing~\cite{costbook}.  

Previous research into reversible execution is broadly focused on two key areas; introducing reversible programming languages or semantics, and adding reversibility to irreversible programming languages or semantics. Reversible programming languages \cite{CL1986,TH2017A} ensure that any valid program written is executable in both the forward and reverse direction. This elegant solution to reversibility is described in Section~\ref{sec:related-work}, where several examples of such languages are given. As we will describe later, some of the components of reversible languages (such as post-conditions for reverse control flow) are difficult to automatically generate. This means work is required to transform an irreversible program into an equivalent reversible version. The second category aims to avoid this challenge by adding reversible capabilities directly into an irreversible program. This removes the need for conversion of traditional programs into reversible equivalents, and instead makes each statement reversible by saving lost information needed for reversal. Section~\ref{sec:related-work} details examples of this approach, including checkpointing, which typically requires the program state to be saved regularly and then used to restore to a previous position, and source-to-source translation~\cite{KP2014}, which takes an original program and produces a version that is now capable of saving all information needed for reverse execution. 

One example of the source-to-source translation approach is the Reverse C Compiler (RCC for short)~\cite{KP2014}, a tool capable of taking a (sequential) C program and modifying this with state saving statements. To the best of our knowledge, RCC has not been released publicly. More detail is given in Section~\ref{sec:related-work}. We follow closely to this, but focus on supporting a language containing interleaving parallelism (concurrency). Our approach is also semantics-driven, we do not insert extra statements into the code of the original program, but define operational semantics that perform both the expected forward execution, and the necessary saving of lost information. A second operational semantics is defined for reverse execution, where this saved information is used. A major motivation of this work is to produce an approach proved formally to be correct, showing that reverse execution does indeed produce exactly the initial program state.  

\begin{figure}[t]
\begin{align*}
& \text{Original execution} & (\code{P},\square)  &\rightarrow (\code{skip},\square') \\
& \text{Annotated execution} & (\code{P$^a$},\square,\delta) &\rightarrow (\code{skip},\square', \delta') \\
& \text{Inverted execution} & (\code{P$^r$},\square', \delta') &\rtran{} (\code{skip},\square'',\delta)
\end{align*}
\vspace{-.4cm}
\caption{An original, annotated and inverted execution. The program states $\square$ (the original state) and $\square''$ (the state produced after inverted execution) are equivalent.}
\label{fig:exs}
\end{figure}

In this paper we propose a source to source translator approach, outlined in Figure~\ref{fig:exs}, to reversing imperative concurrent programs. Consider a program \code{P}, a program state $\square$ and the original execution shown in Figure~\ref{fig:exs}. We produce a modified version \code{P$^{a}$}, named the \emph{annotated version}, that performs forward execution and saves \emph{reversal information} that would otherwise have been lost. Let $\delta$ be the initially empty store used to record this additional information. Execution of this annotated program produces the same final state $\square'$ as traditional execution, and some reversal information $\delta'$. We next produce \code{P$^{r}$}, named the \emph{inverted version}. This is a forward-executing program (with inverted statement order) that simulates reversal by undoing each statement, using (and removing) reversal information from $\delta'$ to do so. 
This inverted execution begins in the final program state $\square'$ with the reversal information $\delta'$, and produces a final program state  $\square''$. Later we prove that this reversal is correct meaning that $\square''$ is equivalent to $\square$. All reversal information is used, producing the original $\delta$ and thus showing that our reversal is garbage-free. To ensure we achieve our aim of proving this method to be correct, reversibility is defined using basic structures such as stacks. We note that some of these choices may impact the overall efficiency, with future optimisations offering potential improvements. We also note here that our approach does not currently support floating point numbers. Adding such support is an interesting and challenging area of potential future work. Other examples include adding support for pointers, user defined data types and division. 

Our main contribution here is to introduce a method of reversing the execution of imperative concurrent programs. A common issue with much of the current literature is that such approaches are not formally proved correct. We have addressed this via a detailed proof of correctness, confirming that our approach reverses concurrent program execution as intended. We expand on our work published previously \cite{JHphd} by introducing support for arrays. Our first data structure increases the complexity of the language supported, and allows for more interesting programs to be reversed. This required the definition of operational semantics for array operations, including the saving and using of necessary information lost as these execute. The correctness of our method of reversibility has also been maintained and now includes the new constructs introduced to support arrays.   

%% file: Related-work.tex
Perumalla describes important background research into reversible computation~\cite{KP2014}, including multiple potential paradigms of reversal. Some examples include \emph{Forward-Reverse-Commit}, where programs are reversibly executed until a point at which reversibility is no longer required (and the execution is committed), \emph{Compute-Copy-Uncompute}, where a modified forward execution is performed and the result stored, before the execution is reversed. 

Bennett stated in 1973 that a program execution can be reversed by saving a history of the information lost as it is executed forwards and using this to undo the program execution~\cite{CB1973}. There are two types of reversibility, namely \emph{physical reversibility} and \emph{logical reversibility}. Physical reversibility focuses on producing executions that do not lose information, including reversible hardware such as circuits~\cite{costbook}. We focus on logical reversibility, that is the ability to reproduce a previous execution state from some later position.

Reversible programming languages is one way of implementing this. Unlike traditional languages, any valid program written in a reversible language is executable both forwards and backwards, without the need to save any other information. One example of a reversible imperative programming language is Janus~\cite{CL1986,TY2007}. Only reversible statements are supported in Janus. For example, destructive assignments (those that overwrite a value of a variable) are not valid, with only constructive assignments (an increment or decrement, each of which have an obvious inverse) supported. Conditional statements and while loops each contain a post-condition, a boolean expression that is used to determine reverse control flow. 
Other examples of reversible imperative programming languages include R-WHILE~\cite{RG2016}, R-CORE~\cite{RG2017} and R~\cite{MPF1999}. 

Several object-oriented reversible programming languages have also been proposed, such as Joule proposed by Schultz \cite{UPS2016,UPS2018}. Haulund et al introduced ROOPL~\cite{TH2017A,TH2018}. More recently, a web interface for ROOPL was introduced~\cite{LHS2021} as part of working towards a unified language architecture. 

Reversible functional programming languages include that defined by Yokoyama et al~\cite{TY2011}, and Rfun from Thomsen et al~\cite{MKT2015}. CoreFun is a typed functional language from Jacobsen et al~\cite{PAHJ2018}. Lanese et al presented an approach for reversible execution of Erlang, a functional and concurrent programming language, by defining reversible semantics~\cite{IL2018B}. Alongside working on functional languages, Mogensen has also discussed a reversible version of the Static Single Assignment, named RSSA~\cite{TM2015}, and reversible functional array programming~\cite{TM2021}. 

An alternative implementation of logical reversibility is to add reversibility to irreversible languages. An obvious approach is named \emph{full checkpointing}~\cite{KP2014}, where the entire state is saved after each step. Previous execution positions can be restored simply by reverting to the saved state. This faces multiple challenges including the time and memory required for saving. Reduced checkpointing approaches, such as \emph{periodic} or \emph{incremental}, decrease the frequency or size of each checkpoint, but typically require forward re-execution to reach the majority of states~\cite{KP2014}. Such re-execution takes time (depending on the frequency of checkpoints) and suffers possible concurrency problems. 

The Reverse C Compiler (RCC) is a source-to-source translator for a subset of the programming language C~\cite{KP2014,CC1999}. Given an original program, RCC produces two versions of it. The first version contains inserted `SAVE' commands to record lost information. The assignment statement \code{X = 2} becomes \code{SAVE(X); X = 2}. The second version, containing all original statements but in inverted order, includes `RESTORE' commands that use saved information to correctly reverse each statement. The  assignment statement \code{X = 2} becomes \code{RESTORE(X)} for reversal. Conditional statements save the result of evaluating the boolean condition, whereas while loops use a loop counter. While supporting other statement types, there is no support for parallel (or concurrent) programs.

Backstroke is a framework developed to support automatic generation of reverse code for functions written in C++~\cite{GV2011}. Schordan et al continue this work of applying this to PDES~\cite{MS2018,DC2017MS}, as have others~\cite{YT2006}. Using the Execute-Reverse-Commit paradigm, Backstroke takes an original code fragment, an event E, and produces three versions. Backstroke produces E$_{forward}$  (that performs forward execution and saves necessary information), E$_{reverse}$ (that uses  this saved information to reverse all effects), and E$_{commit}$ (that performs actions that are irreversible or that should not be undone). Optimistic PDES~\cite{RF1999} can then be performed, with any events executed in the wrong order corrected via reversal~\cite{CC1999,DC2016,DC2017}. 

Schordan et al describe an approach using transformation based on recording pairs of the form (address,value) \cite{MS2015,MS2018}. Code is instrumented in such a way that these pairs are recorded, while no information regarding the executed program path (or interleaving) is required. This approach of recording only data on the heap is sufficient for implementing reversibility at a function-level, a granularity that suits its application to PDES. The work in \cite{GV2011} is slightly different to this and instead generates reverse code (a similar notion to that presented in this work). 

Reversible process calculi is the subject of previous research. Danos and Krivine proposed a reversible version of the Calculus of Communicating Systems (CCS), named RCCS~\cite{VD2004}. Processes are extended to contain memories that help identify processes used within actions. Phillips and Ulidowski introduced an alternative approach to reversing CCS, named CCS with Communication Keys (CCSK)~\cite{IP2007}. All actions are annotated with an identifier (communication key), with the two parties of a communication using the same identifier. Both RCCS and CCSK allows actions to be reversed in any order that respects all causal dependencies. This is named \emph{causal-consistent reversibility}, and differs to \emph{backtracking reversibility}~\cite{backtracking}, where actions (or steps) must be reversed in exactly the inverted forward order. Properties of reversible systems are described~\cite{IL2020}, as is out of causal order reversibility~\cite{IP2012}. Medi\'c et al have compared the two previously mentioned reversible versions of CCS~\cite{DM2020,DM2021}.

Debugging, the process of finding and fixing the underlying cause of a software defect, is a possible application of reversible computation~\cite{JE2012,JH2020,AZ2009}. \emph{Cyclic debugging} repeatedly executes a program experiencing a bug with small code changes each time. This suits sequential programs, where re-execution is guaranteed to produce the same behaviour. Two executions of a concurrent program can behave differently, with a bug potentially only occurring under a specific interleaving (named a `Heisenbug'). Whether a fix has worked is therefore difficult to judge. Engblom reviews reversible debuggers~\cite{JE2012}, with an example of a commercially available reversible debugger being UndoDB~\cite{UndoWhitePaper,Undo}.

The previous debuggers follow backtracking order. Giachino et al introduced the causal-consistent debugger CaReDeb for the language $\mu$Oz~\cite{caredeb}. Lanese et al. defined a causal-consistent reversible debugger for Erlang named Cauder~\cite{cauder,IL2018} with support for distributed programs~\cite{GF2021}. 

%% file: LanguageState.tex
We begin by introducing our imperative concurrent programming language, and the necessary program state. 

\subsection{Syntax}
The syntax of our programming language is shown in Figure~\ref{fig:syntax}. The constructs supported include constructive assignments (increment or decrement of variable values), destructive assignments (overwriting of variable values), conditional branching, while loops, blocks, declaration and removal of local variables or arrays, and potentially recursive procedures. Let $\setOf{P}$ be the set of programs ranged over using \code{P}, \code{Q}, etc, and $\setOf{S}$ be the set of statements ranged over using \code{S}, \code{T}, etc. Let $\setOf{X}$ be the set of variables ranged over by \code{X}, \code{Y}, etc, and $\setOf{Z}$ be the set of integers such that \code{n} $\in \setOf{Z}$. The set of arithmetic expressions $\setOf{E}$ is such that $\code{E} \in \setOf{E}$, while the set of boolean expressions $\setOf{B}$ is such that $\code{B} \in \setOf{B}$.
\begin{figure}[t]
{\footnotesize \begin{align*}
\texttt{P} &::= \texttt{$\varepsilon$} ~|~ \texttt{S} ~|~ \texttt{P};\texttt{P} ~|~ \texttt{P par P} \\
\texttt{S} &::= \texttt{skip} ~|~  \texttt{X = E pa} ~|~ \texttt{name[E] = E pa} ~|~ \texttt{X += E pa} ~|~ \texttt{X -= E pa} ~|~ \texttt{name[E] += E pa} ~| \\& \phantom{\texttt{::==}} \texttt{name[E] -= E pa} ~|~ \texttt{if In B then P else Q end pa}  ~|~ \texttt{while Wn B do P end pa} ~| \\& \phantom{\texttt{::==}} \texttt{begin Bn BB end}  ~|~ \texttt{call Cn n pa} ~|~   \texttt{runc Cn P end} \\
\texttt{E}   &::=   \texttt{X} ~|~ \texttt{n} ~|~ \texttt{Arr[n]} ~|~ \texttt{(E)} ~|~ \texttt{E Op E} \qquad
\texttt{B}  ::= \texttt{T} ~|~ \texttt{F} ~|~ \lnot\texttt{B}~|~ \texttt{(B)} ~|~ \texttt{E == E} ~|~ \texttt{E > E} ~|~  \texttt{B $\land$ B}  \\
\texttt{BB} &::= \texttt{DV};\texttt{DA};\texttt{DP};\texttt{P};\texttt{RP};\texttt{RA};\texttt{RV} \\
\texttt{DV} &::= \texttt{$\varepsilon$} ~|~ \texttt{var X = E pa};\texttt{DV} \qquad  
\texttt{RV} ::= \texttt{$\varepsilon$} ~|~ \texttt{remove X = E pa};\texttt{RV} \\
\texttt{DP} &::= \texttt{$\varepsilon$} ~|~ \texttt{proc Pn name is P end pa};\texttt{DP} \qquad
\texttt{RP} ::= \texttt{$\varepsilon$} ~|~ \texttt{remove Pn name is P end pa};\texttt{RP} \\
\texttt{DA} &::= \texttt{$\varepsilon$} ~|~ \texttt{arr[n] name pa};\texttt{DA}  \qquad 
\texttt{RA} ::= \texttt{$\varepsilon$} ~|~ \texttt{remove arr[n] name pa};\texttt{RA}
\end{align*} }
\vspace{-.4cm}
\caption{Imperative concurrent programming language syntax, where \code{Op} = \{+, -, +=, -=\}.}
\label{fig:syntax}
\end{figure}
Our language syntax follows closely to other while languages, including that of \huttel~\cite{HH2010}. We do however introduce several notions not seen in common while languages. Firstly, we support concurrent execution through our parallel composition operator, and represent two programs \code{P} and \code{Q} in parallel as \code{P par Q}. We note here we also write this as \code{par \{P\} \{Q\}} in several examples for presentation purposes. Support for arrays is also first introduced here and was not part of our previous work~\cite{JHphd}. An array name acts as a pointer to the beginning of a block of memory for array elements. This increases the complexity of the language and allows for more interesting example programs. Secondly, constructs are uniquely named to aid identification of a specific construct during reversal (explained later). This is achieved using elements of the set of \emph{construct identifiers} $\setOf{C}$. Construct identifiers \code{In}, \code{Wn}, \code{Bn}, \code{Pn} and \code{Cn} ($\in \setOf{C}$) are used to name conditionals, loops, blocks, procedures and procedure calls respectively. Thirdly, the majority of statements have a \emph{statement path} associated to them that is used to determine the scope of local variables. For correct evaluation of a local variable, we must first determine in which block this local variable was declared. Given a statement, its statement path is the ordered sequence of all (including nested) block names in which it exists. The set of statement paths \code{Pa} is such that $\code{pa} \in \code{Pa}$. Statements that do not require any evaluation, such as a block statement, do not require a statement path. Fourthly, we include \emph{removal statements} and the \code{runc} construct. For reversal purposes and to avoid garbage, any local declarations within a block are also removed at the end of that block. For example, a local variable is declared and exists for the execution of that block body, before being removed prior to the closure of the block. The \code{runc} construct exists to distinguish between the procedure call that has not yet been executed, and the execution of a procedure body. 

Finally, we make several key assumptions about programs written in our syntax. Local variables, arrays and procedures are only declared within blocks, those in the same block are named uniquely, and local declarations are removed in exactly the inverted order. The \code{runc} construct is not allowed to appear in an original program, and is instead only introduced as a result of executing a procedure call.   

\subsection{Program State} \label{sec:prog-state}
We now introduce the program state. As in \cite{HH2010}, the program state contains a \emph{variable environment} and a \emph{data store}. 

\vspace{.2cm}
\noindent \textbf{Variable Environment} {\boldmath$\gamma$}. Each version of a variable ($\code{X} \in \setOf{V}$) and parent block name ($\code{Bn}$) is mapped to a unique memory location ($\code{l} \in \setOf{L}$). The notation $\gamma[\code{(X,Bn)} \Rightarrow \code{l}] $ maps \code{X} declared within \code{Bn} to \code{l}, while $\gamma[\code{(X,Bn)}]$ removes this. All global variables are assumed to exist each with a distinct memory location.

\vspace{.2cm}
\noindent \textbf{Data Store} {\boldmath$\sigma$}. Each memory location $\code{l}$ is mapped to the value $\code{v}$ it currently holds. The notation $\sigma[\code{l} \mapsto \code{v}]$ updates $\sigma$ such that \code{l} now holds \code{v}, while $\sigma(\code{l})$ returns the current integer value at memory location \code{l}. Expression $\sigma[\{\code{l1},\cdots,\code{ll+n}\} \mapsto \code{v}]$ represents multiple modifications to memory locations performed in a single step. Each memory location in the range \code{l1} up to \code{l1}+\code{n} is updated to hold some value~\code{v}. All required memory locations exist initially and contain \code{0}. The function $\func{nextLoc}{\sigma}$ returns the next available memory location, while $\func{nextLocBlock}{\sigma,\code{n}}$ reserves a consecutive block of \code{n} locations, returning the first.

We further introduce the \emph{while environment} and the \emph{procedure environment}, both of which are initially empty and help with reversal.

\vspace{.2cm}
\noindent \textbf{While Environment} {\boldmath$\beta$}. The unique construct identifier $\code{Wn}$ of each currently active while loop is mapped to a copy of that loop. The notation $\beta$[\code{Wn} $\Rightarrow$ \code{P}] inserts a mapping into $\beta$ between the loop uniquely named \code{Wn} and the copy \code{P} of the loop body, while $\beta$[\code{Wn}] removes this.

\vspace{.2cm}
\noindent  \textbf{Procedure Environment} {\boldmath$\mu$}. A construct identifier $\code{Pn}$ of each procedure is mapped to a copy of the procedure body. Expression $\mu[\code{Pn} \Rightarrow \code{(n,P)}]$ maps the identifier \code{Pn} to the procedure name \code{n} and copy of the body \code{P}, while $\mu[\code{Pn}]$ removes this. 

The program state consists of four environments that we have introduced so far. We use $\square$ to represent the tuple ($\gamma$,$\sigma$,$\beta$,$\mu$). By abuse of notation, we use $\square$ to represent any subset of this tuple in the following semantics. 

Given our syntax and program state, we can now present the operational semantics of traditional (and irreversible) execution, following closely to that of \huttel~\cite{HH2010}. To save space we do not show this here (though it can be seen in \cite{JHphd}), as there is similarity with semantics shown in Section~\ref{sec:seman-ann}. Ignoring any use of $\delta$ and stacks \code{A} (both of which are introduced later) in the semantics in Section~\ref{sec:seman-ann} gives the traditional semantics.

\subsection{Odd-Even Transposition Sort} \label{sec:sort-orig-sec}
We shall use a program that implements a sorting algorithm throughout the rest of the paper to explain and illustrate our method for reversing execution. More specifically, the program is an implementation of an Odd-even Transposition Sort, a concurrent version of a bubble sort 
which was originally introduced by Habermann \cite{NH1972}. It is shown in Figure~\ref{fig:odd-even-sort-trad}, with statement paths omitted as these can be read from the code. For example, the assignment on line~14 has the path \code{b2.0;b1:0} since this statement is within the block \code{b2.0}, which is within block \code{b1.0}. Block names will be modified for each iteration as part of the process to keep construct identifiers unique (meaning paths also change). For example, block \code{b2.0} becomes \code{b2.1} during the second iteration. We omit empty false branches of conditional statements.

Each pass of the algorithm contains an even and an odd phase. During even phases (lines 8--27), each element at an even index is compared with their neighbour and potentially switched. Each comparison is performed in parallel and can be interleaved in any order. During odd phases (lines 28--47), each element at an odd index is compared and potentially swapped with its neighbour again in parallel. As for a bubble sort, the worst case scenario requires $n$ complete loop iterations for an array of length $n$. After all required loop iterations, the array will be sorted into ascending order. Example~\ref{sec:sort-orig} describes the traditional (irreversible) execution of our program.

\begin{figure}[t!]
{\tiny \begin{lstlisting}[xleftmargin=2.0ex,mathescape=true,multicols=2,basicstyle=\scriptsize]
$\kw{begin b1.0}$
 $\kw{arr[5]}~ \code{l;}$
 $\code{l[0] = 7; l[1] = 3; l[2] = 4;}$
 $\code{l[3] = 1; l[4] = 6;}$
 $\code{count = 0;}$ 
 
 $\kw{while w1.0}~ \code{(count1 < 4)}~\kw{do}$
  ${\color{green} \code{//even phase}}$    
  $\kw{par \{}$
   $\kw{if i2.0}~ \code{(l[0] > l[1])}~\kw{then}$
    $\kw{begin b2.0}$
     $\kw{var}~ \code{temp = 0;}$
     $\code{temp = l[0]; l[0] = l[1];}$ 
     $\code{l[1] = temp;}$
     $\kw{remove}~ \code{temp = 0}$
    $\kw{end}$
   $\kw{end}$
  $\kw{\}}$ $\kw{\{}$
   $\kw{if i3.0}~ \code{(l[2] > l[3])}~\kw{then}$
    $\kw{begin b3.0}$
     $\kw{var}~ \code{temp = 0;}$
     $\code{temp = l[2]; l[2] = l[3];}$ 
     $\code{l[3] = temp;}$
     $\kw{remove}~ \code{temp = 0}$
    $\kw{end}$
   $\kw{end}$ 
  $\kw{\};}$     
  ${\color{green} \code{//odd phase}}$  
  $\kw{par \{}$
   $\kw{if i4.0}~ \code{(l[1] > l[2])}~\kw{then}$
    $\kw{begin b4.0}$
     $\kw{var}~ \code{temp = 0;}$
     $\code{temp = l[1]; l[1] = l[2];}$ 
     $\code{l[2] = temp;}$
     $\kw{remove}~ \code{temp = 0}$
    $\kw{end}$
   $\kw{end}$
  $\kw{\}}$ $\kw{\{}$
   $\kw{if i5.0}~ \code{(l[3] > l[4])}~\kw{then}$
    $\code{begin b5.0}$
     $\kw{var}~ \code{temp = 0}$
      $\code{temp = l[3]; l[3] = l[4];}$ 
      $\code{l[4] = temp;}$
      $\kw{remove}~ \code{temp = 0}$
     $\kw{end}$
    $\kw{end}$   
  $\kw{\};}$
  $\code{count += 1}$
 $\kw{end;}$
  
 ${\color{green} \code{//array removal}}$ 
 $\kw{remove arr[5]}~ \code{l}$ 
$\kw{end}$
\end{lstlisting} }
\caption{Odd-even transposition sort. Empty else branches and paths are omitted.}\vspace{-.5cm}
\label{fig:odd-even-sort-trad}
\end{figure}

\begin{example} \label{sec:sort-orig}
Let \code{l} be the array \code{[7,3,4,1,6]} as on lines 2--4 of Figure~\ref{fig:odd-even-sort-trad}. This is the array that we wish to sort via the odd-even transposition sort. Assuming ascending order, the first even phase will produce the array \code{[3,7,1,4,6]}, before the first odd phase produces \code{[3,1,7,4,6]}. After all required passes, the resulting array is \code{[1,3,4,6,7]}. We shall describe the annotated forward execution in Example~\ref{sec:sort-ann} and the reverse execution in Example~\ref{sec:sort-inv}. $\myqed$
\end{example}

%% file: Annotation.tex
\begin{figure}[t]
{\small
\begin{align*}
\texttt{AP} &::= \texttt{$\varepsilon$} ~|~ \texttt{AS} ~|~ \texttt{AP};\texttt{AP} ~|~ \texttt{AP par AP} \\
\texttt{AS} &::= \texttt{skip I} ~|~ \texttt{X = E (pa,A)} ~|~ \texttt{name[E] = E (pa,A)} ~|~ \texttt{X += E (pa,A)} ~| \\& \phantom{\texttt{::==}} \texttt{X -= E (pa,A)} ~|~ \texttt{name[E] += E (pa,A)} ~|~ \texttt{name[E] -= E (pa,A)} ~| \\& \phantom{\texttt{::==}} \texttt{if In B then AP else AP end (pa,A)} ~|~  \texttt{while Wn B do AP end (pa,A)} ~| \\& \phantom{\texttt{::==}} \texttt{begin Bn ABB end}  ~|~ \texttt{call Cn n (pa,A)} ~|~ \texttt{runc Cn AP end} \\
\texttt{E}   &::=   \texttt{X} ~|~ \texttt{n} ~|~ \texttt{Arr[n]} ~|~ \texttt{(E)} ~|~ \texttt{E Op E} \qquad
\texttt{B}  ::= \texttt{T} ~|~ \texttt{F} ~|~ \lnot\texttt{B}~|~ \texttt{(B)} ~|~ \texttt{E == E} ~|~ \texttt{E > E} ~|~  \texttt{B $\land$ B}  \\
\texttt{ABB} &::= \texttt{ADV};\texttt{ADA};\texttt{ADP};\texttt{AP};\texttt{ARP};\texttt{ARA};\texttt{ARV} \\
\texttt{ADV} & ::= \texttt{$\varepsilon$} ~|~ \texttt{var X = E (pa,A)};\texttt{ADV}   \qquad
\texttt{ARV} ::= \texttt{$\varepsilon$} ~|~ \texttt{remove X = E (pa,A);ARV} \\
 \texttt{ADP} & ::= \texttt{$\varepsilon$} ~|~ \texttt{proc Pn name is AP end (pa,A)};\texttt{ADP}  \quad 
 \texttt{ARP} ::= \texttt{$\varepsilon$} ~|~ \texttt{remove Pn name is AP end (pa,A)};\texttt{ARP} \\
\texttt{ADA} & ::= \texttt{$\varepsilon$} ~|~ \texttt{arr[n] name (pa,A)};\texttt{ADA} \qquad 
\texttt{ARA} ::= \texttt{$\varepsilon$} ~|~ \texttt{remove arr[n] name (pa,A)};\texttt{ARA}
\end{align*} }
\vspace{-.4cm}
\caption{Annotated programming syntax, where \code{Op} = \{+, -, +=, -=\}.}
\label{fig:syntax-ann}
\end{figure}

Annotation is the process of producing a version of a given program that both performs the expected forward execution and saves reversal information. This reversal information consists of two parts; the interleaving order and unrecoverable information lost during statement execution. To save this we define the semantics of annotated execution. We first introduce the syntax of annotated programs.

\subsection{Annotated Syntax}
We modify an original program \code{P} such that it is capable of recording both lost information and the interleaving order. The resulting program \code{AP}, an element of the set of all annotated programs $\setOf{AP}$, is of the syntax in Figure~\ref{fig:syntax-ann}. Compared to Figure~\ref{fig:syntax}, we use \code{AP} and \code{AS} in place of \code{P} and \code{S} respectively, and introduce stack names \code{A}, which we define in Section~\ref{sec:inter-order}. Now each path \code{pa} has a stack associated to it via the notation $(\code{pa}, \code{A})$. An annotated version of a program is produced via the function \emph{ann}: $\setOf{P} \rightarrow \setOf{AP}$ in Figure~\ref{fig:func-ann}.  
\begin{figure}[t]
{\footnotesize \begin{align*}
&ann(\texttt{$\varepsilon$}) = \texttt{$\varepsilon$} \\
&ann(\texttt{S;P}) = ann(\texttt{S}); ann(\texttt{P}) \\ 
&ann(\texttt{P par Q}) = ann(\texttt{P}) \texttt{ par } ann(\texttt{Q}) \\
&ann(\texttt{skip}) = \texttt{skip I} \\
&ann(\texttt{X = e pa}) = \texttt{X = e (pa,A)} \\
&ann(\texttt{if In b then P else Q end pa}) = \\ &\phantom{==} \texttt{if In b then $ann($P$)$ else $ann($Q$)$ end (pa,A)} \\
&ann(\texttt{while Wn b do P end pa}) = \texttt{while Wn b do $ann($P$)$ end (pa,A)} \\
&ann(\texttt{begin Bn P end}) = \texttt{begin Bn $ann($P$)$ end} \\
&ann(\texttt{var X = E pa}) = \texttt{var X = E (pa,A)} \\
&ann(\texttt{arr[n] name pa}) = \texttt{arr[n] name (pa,A)} \\
&ann(\texttt{proc Pn n is P end pa}) = \texttt{proc Pn n is $ann($P$)$ end (pa,A)} \\
&ann(\texttt{call Cn n pa}) = \texttt{call Cn n (pa,A)} \\
&ann(\texttt{runc Cn P end}) = \texttt{runc Cn AP end A} \\
&ann(\texttt{remove X = E pa}) = \texttt{remove X = E (pa,A)} \\
&ann(\texttt{remove arr[n] name pa}) = \texttt{remove arr[n] name (pa,A)} \\
&ann(\texttt{remove Pn n is P end pa}) = \texttt{remove Pn n is $ann($P$)$ end (pa,A)} 
\end{align*} }
\vspace{-.4cm}
\caption{Partial definition of the Annotation function $ann$, where \code{I} indicates a potential stack \code{A} (as some skip steps, including those produced when closing a block, will not have a stack) and shown in full in Appendix~\ref{full-ann-appen}. All statements omitted follow the instance of \code{X = e pa}, where \code{pa} becomes \code{(pa,A)}.}
\label{fig:func-ann}
\end{figure}

\subsection{Capturing Interleaving Order} \label{sec:inter-order}
The specific execution path of a forward concurrent program is crucial for correct reversal. This is the order in which the statements of a program are executed forwards. Due to the potential for data races, two forward executions with different interleaving orders may produce different program states. Reversing statements in a different order to that of the inverted forward order may lead to incorrect program states (a state not reached during forward execution). Such incorrect reversal is a problem for any possible application of reversible execution. Consider reversible debugging where restoring program states correctly is vital to help detect (and fix) software bugs. The following two scenarios illustrate some negative consequences of incorrect reversal order.

\begin{example} \label{ex:incorrect-reverse-order}
Figure~\subref*{ex:order-assign} contains a race between two assignments to the same variable \code{X}, where \code{X} is initially \code{1}. The left assignment executes first (overwriting \code{1}) and is followed by the right assignment (overwriting \code{3}). Assume reverse execution incorrectly reverses the left assignment first, restoring the value (it held before this assignment) to \code{1}. The right assignment is then reversed restoring \code{X} to \code{3}. This reversal is incorrect as it does not match the initial state. $\myqed$
\end{example}

\begin{example} \label{ex:incorrect-reverse-order-cond}
Consider Figure~\subref*{ex:order-cond}, and an execution where the left conditional executes the true branch (overwriting \code{Z}) and the right conditional executes the false branch (also overwriting \code{Z}). Assume an inverse execution that reverses the incorrect branch of each conditional. Both conditionals now reverse assignments to \code{Y} that were not performed forwards, meaning old values were not saved and will lead to a state not reached during forward execution. $\myqed$
\end{example}

\newsavebox\exOrderAssign
\newsavebox\exOrderCond

\begin{lrbox}{\exOrderAssign}
\begin{minipage}{0.35\textwidth}
\centering
{\small \begin{lstlisting}[xleftmargin=4.0ex,mathescape=true,basicstyle=\scriptsize]
$\code{X = 3} \hspace{.5cm} \kw{par} \hspace{.5cm} \code{X = 5}$
\end{lstlisting} }
\end{minipage}
\end{lrbox}

\begin{lrbox}{\exOrderCond}
\begin{minipage}{0.62\textwidth}
\centering
{\small \begin{lstlisting}[xleftmargin=4.0ex,mathescape=true,basicstyle=\scriptsize]
$\kw{if i1 } \code{(Z > 3) } \kw{then} \hspace{.5cm} \phantom{\code{par}} \hspace{.5cm} \kw{if i2 } \code{(Y > 3) } \kw{then}$
  $\code{Z = 2} \hspace{2cm} \phantom{\code{par}} \hspace{.5cm}   \code{Y = 4}$
$\kw{else} \hspace{2.3cm} \kw{par} \hspace{.5cm} \kw{else}$
  $\code{Y = 6} \hspace{2cm} \phantom{\code{par}} \hspace{.5cm}   \code{Z = 5}$
$\kw{end} \hspace{2.4cm} \phantom{\code{par}} \hspace{.5cm} \kw{end}$
\end{lstlisting} }
\end{minipage}

\end{lrbox}

\begin{figure}[t]
\subfloat[Racing assignments\label{ex:order-assign}]{ \usebox\exOrderAssign }
\subfloat[Conditional statements\label{ex:order-cond}]{ \usebox\exOrderCond }
\caption{Programs that reverse incorrectly under specific interleavings. Initially \code{X~=~1}, \code{Y~=~2} and \code{Z~=~4}. After the executions discussed in Examples~\ref{ex:incorrect-reverse-order}~and~\ref{ex:incorrect-reverse-order-cond}, \code{X~=~5}, \code{Y~=~2} and \code{Z~=~2}.  Statement paths and identifier stacks are omitted.}
\label{ex:inverse-identifiers}
\end{figure}

Such problems are overcome by recording the order in which statements are executed forwards with the help of \emph{identifiers} as in \cite{IP2007,IP2012}. Identifiers, which are consecutive natural numbers and are ranged over using \code{m} and \code{n} (where both are elements of the set of all identifiers $\setOf{K}$), are used in ascending order to capture the order in which statements execute. At each interleaving decision point the next unused available identifier is associated with the statement that is executing. An annotated version of a program must be capable of storing potentially many identifiers associated to each statement. Identifiers are maintained as an ordered sequence, with the function \func{next}{} returning the next unused identifier, and \func{previous}{} returning the previously used identifier. 

In order to save identifiers as they are used, the syntax\footnote{The identifier stacks would in practise be saved in a separate environment, named the identifier environment. Each statement would then include the stack name, meaning the stack of identifiers outlives the execution. To keep both the proof and examples concise, we display all stacks as part of the statement and omit use of the identifier environment.} of each statement that requires identifiers is expanded to include an \emph{identifier stack}. Each is uniquely named and represented as \code{Ai}. The stack notation \code{m:Ai} represents an identifier \code{m} and the remaining stack \code{Ai}. Example~\ref{ex:fwd-iden-use} demonstrates forward identifier use on a small program.   


\begin{figure}[t]
\centering
\begin{minipage}{0.75\textwidth}
{\small \begin{lstlisting}[xleftmargin=4.0ex,mathescape=true]
$\kw{while w1 }\code{(X > 3)} \kw{ do} \hspace{1cm} \phantom{\code{par}} \hspace{.85cm} \code{Y = 2} \quad {\color{blue}\langle \code{5} \rangle} \code{;}$
   $\code{X = X - 1} \quad {\color{blue}\langle \code{1,3,7} \rangle} \code{} \hspace{0.57cm} \code{par} \hspace{.96cm} \code{X = X - Y} \quad {\color{blue}\langle \code{6} \rangle} \code{;}$
$\kw{end} \quad {\color{blue}\langle \code{0,2,4,8} \rangle} \code{} \hspace{1.9cm} \phantom{\code{par}} \hspace{.8cm} \code{Y = X} \quad {\color{blue}\langle \code{9} \rangle} \code{}$
\end{lstlisting} }
\end{minipage}
\caption{Example use of identifiers to capture interleaving order. Let \code{X = 10}, \code{Y = 1} and \code{Z~=~3} initially. Statement paths and identifier stacks are omitted.}
\label{ex:fwd-identifiers}
\end{figure}

%
%
%
%

\begin{example} \label{ex:fwd-iden-use}
Figure~\ref{ex:fwd-identifiers} contains a while loop in parallel with three assignment statements. The notation ${\color{blue}\langle \phantom{\code{3}} \rangle}$ represents a stack used to store identifiers for each statement, with the left most element being the head. The interleaving order of this program is captured via the identifiers shown.   $\myqed$
\end{example}

\subsection{Capturing Unrecoverable Information} \label{sec:lost-info-ann}
Some data lost as statements execute forwards is unrecoverable. The simplest example of this is a destructive assignment, where the current value of the variable in question is overwritten and lost. Correct reversal requires this variable to be restored to this lost value, with no way of re-calculation.  We must therefore save such lost information, including control flow information (see Section~\ref{sec:seman-ann}). Therefore we introduce the \emph{auxiliary store} $\delta$ for this purpose.
The store $\delta$  is separate to the program state and consists of stacks (ideal for reversibility due to their Last-In, First-Out (LIFO) nature).

There exists a single stack in $\sigma$ for each variable within the program. There is also a single stack \code{B} and a single stack \code{W} in $\sigma$, each used to store the result of evaluating all conditional statement or while loop boolean expressions respectively. The stacks \code{WI} and \code{Pr} are used to store identifiers associated to renamed copies of loop or procedure bodies prior to their removal. The store $\delta$ is defined as $\delta: (\setOf{X} \rightarrow (\idens \times \ints)) \cup (\code{B} \rightarrow (\idens \times \setOf{B})) \cup (\code{W} \rightarrow (\idens \times \setOf{B})) \cup (\code{WI} \rightarrow (\idens \times \setOf{C})) \cup (\code{Pr} \rightarrow (\idens \times \setOf{C})) $, where $\setOf{C}$ is the set of identifier sequences retrieved from a loop or procedure copy prior to its removal. Each stack within $\delta$ stores pairs of natural numbers, with the first element being an identifier and the second being an old value or boolean. We note that throughout we use \code{T} and \code{F} to represent \code{1} and \code{0} (boolean values) respectively.

The notation $\delta$[\code{el} $\rightharpoonup$ \code{St}] represents pushing the element \code{el} to the stack \code{St}, and $\delta$[\code{St/St$'$}] represents popping the head of stack \code{St}, leaving the tail \code{St$'$}. Correspondingly, $\delta[\code{\{(\code{m},$\sigma$(\code{el1}),\ldots,(\code{m},$\sigma$(\code{el1+n}))\} $\rightharpoonup$ X}]$ is used to represent pushing many pairs onto a single stack \code{X}. Each pair contains the identifier \code{m}, but with the contents of memory locations ranging from \code{el1} to \code{el1}+\code{n}, in order. Finally, $\delta(\code{name},\code{m})$ is used to represent continually popping the stack \code{name} while the top element contains the identifier \code{m}, returning each of the old values saved there. 

By abuse of notation, we update $\square$ to now also contain $\delta$. 
The following defines when two program states are equivalent. 

\begin{definition}{(Equivalent states)} \label{def:prog-states}
Let $\square$ be the environment tuple \{$\sigma$,$\gamma$,$\mu$,$\beta$,$\delta$\} and $\square_1$ be the tuple of annotated environments \{$\sigma_1$,$\gamma_1$,$\mu_1$,$\beta_1$,$\delta_1$\}. Let \func{d}{$\gamma$} be the domain of $\gamma$. 
We say that $\square$ is \emph{equivalent} to $\square_1$, written $\square \approx \square_1$, if and only if:\vspace{-.2cm}
\begin{enumerate}
\item \func{d}{$\gamma$} = \func{d}{$\gamma_1$} and $\sigma(\gamma(\textup{\code{X}},\textup{\code{Bn}}))$ = $\sigma_1(\gamma_1(\textup{\code{X}},\textup{\code{Bn}}))$ for all \code{X} $\in$ \func{d}{$\gamma$} and block names \code{Bn}; \vspace{-.4cm}
\item \func{d}{$\mu$} = \func{d}{$\mu_1$}; $\mu(\textup{\code{Pn}})$ = (\textup{\code{n,P}}), $\mu_1(\textup{\code{Pn}})$~=~(\textup{\code{n,AP}}) and \code{AP} = $ann($\code{P}$)$ for all \code{Pn} $\in$ \func{d}{$\mu$};\vspace{-.4cm} 
\item \func{d}{$\beta$} = \func{d}{$\beta_1$};  $\beta(\textup{\code{Wn}})$ = \textup{\code{P}}, $\beta_1(\textup{\code{Wn}})$ = \textup{\code{AP}} and \code{AP} = $ann($\code{P}$)$ for all \code{Wn} $\in$ \func{d}{$\beta$}.
\end{enumerate}
\end{definition}
 
\subsection{Operational Semantics of Annotated Programs} \label{sec:seman-ann}
To define the behaviour of our programming language, we now use  \emph{Structural Operational Semantics} (SOS) \cite{GP2004} to define a \emph{transition system}. A transition system is a directed graph, with each vertex a \emph{configuration} and each edge a \emph{transition}. A \emph{configuration} represents a snapshot of an execution and is composed of a program \code{P} and program state $\square$, written as $(\code{P} \mid \square)$. Two configurations can be linked together as a transition, where the second configuration can be reached from the first via a single step of an execution. Each such transition is a member of the \emph{transition relation} $\atran{}$ (annotated execution), which is defined as the least relation generated by the set of \emph{transition rules} contained within this section.  Each transition rule consists of a \emph{conclusion} (the transition this rule represents), potentially many \emph{premises} (other transitions or claims that must hold for the conclusion to be valid) and an optional \emph{side condition} (constraints placed on components of the rule). Each transition rule that has no premises is called an \emph{axiom}. Let $\atran{}^*$ be the reflexive and transitive closure of $\atran{}$. 
We note that ignoring all use of $\delta$ and identifiers gives the traditional operational semantics of our language.

We now give a syntax-directed, \emph{small step semantics} that defines annotated execution. Each transition rule represents a single step of a larger execution, meaning the resulting configuration need not be \emph{terminal} (no further transitions from this configuration are possible), and instead can be \emph{intermediate}. We use $\hookrightarrow^*_{\code{a}}$ and $\hookrightarrow^*_{\code{b}}$ to represent a step of arithmetic and boolean expression evaluation respectively. As previously described, some transition rules will use identifiers. Such steps are called \emph{identifier steps} and are represented as $\atran{m}$, where \code{m} is the identifier used. Transition rules that do not use identifiers are called \emph{skip steps} and are represented as 
$\rightarrow_s$. In some cases, such as the execution of a loop body, a step may or may not use an identifier depending on what the body does. Such transitions are represented as $\atran{\circ}$. All components of the following rules that are included for reversibility are indicated using the colour blue. The removal of all blue components from the following rules gives the semantics of traditional (and irreversible) execution.

\vspace{.2cm}
\noindent \textbf{Assignment} Destructive assignments overwrite the value of a given variable (rule \sos{D1a} below) or array element (rule \sos{D2a}), with the old value being saved to $\delta$ with $\delta[\code{(m,$\sigma($l$)$) $\rightharpoonup$ X}]$ and 
$\delta[\code{(m,$\sigma$($\sigma$(\code{l})+\code{n})) $\rightharpoonup$ name)}]$ respectively. The given variable is evaluated via the function \code{$evalV($X,pa,$\gamma)$}, taking the variable name \code{X}, the statement path \code{pa} and variable environment $\gamma$, and producing the memory location \code{l}. Constructive assignments increment (\sos{C1a}/\sos{C3a}) or decrement (\sos{C2a}/\sos{C4a}) a variable or array element (e.g. $\sigma[\code{l} \mapsto \sigma(\code{l}) + \code{v}]$) with no reversal information saved.

{\footnotesize
\begin{align*}
&\sos{D1a} \quad \frac{\gs{\code{m} = \func{next}{}} \quad (\code{e pa} \mid \sigma,\gamma,\square) \hookrightarrow^*_{\code{a}} (\code{v} \mid \sigma,\gamma,\square) \quad \code{$evalV($X,pa,$\gamma)$} = \code{l}}{(\code{X = e (pa,A)} \mid \gs{\delta},\sigma,\gamma,\square) \atran{m} (\code{skip \gs{m:A}} \mid \gs{\delta[\code{(m,$\sigma($l$)$) $\rightharpoonup$ X}]},\sigma[\code{l} \mapsto \code{v}],\gamma,\square)} \\[6pt]
& \sos{D2a} \quad \frac{\begin{aligned} & \gs{\code{m} = \func{next}{}} \quad (\code{e pa} \mid \sigma,\gamma,\square) \hookrightarrow^*_{\code{a}} (\code{n} \mid \sigma,\gamma,\square) \quad (\code{e$'$ pa} \mid \sigma,\gamma,\square) \hookrightarrow^*_{\code{a}} (\code{v} \mid \sigma,\gamma,\square) \\[-3pt] & \code{$evalV($name,pa,$\gamma)$} = \code{l} \quad \sigma' = \sigma[(\sigma(\code{l})+\code{n}) \mapsto \code{v}] \quad \gs{\delta' = \delta[\code{(m,$\sigma$($\sigma$(\code{l})+\code{n})) $\rightharpoonup$ name)}]} \end{aligned}}{(\code{name[e] = e$'$ (pa,A)} \mid \sigma,\gamma,\delta,\square) \atran{m} (\code{skip \gs{m:A}} \mid \sigma',\gamma,\delta', \square)} \\[6pt]
&\sos{C1a} \quad \frac{\gs{\code{m} = \func{next}{}} \quad (\code{e pa} \mid \sigma,\gamma,\square) \hookrightarrow^*_{\code{a}} (\code{v} \mid \sigma,\gamma,\square) \quad \code{$evalV($X,pa,$\gamma)$} = \code{l}}{(\code{X += e (pa,A)} \mid \delta,\sigma,\gamma,\square) \atran{m} (\code{skip \gs{m:A}} \mid \sigma[\code{l} \mapsto \sigma(\code{l}) + \code{v}],\gamma,\square)} \\[6pt]
&\sos{C2a} \quad \frac{\gs{\code{m} = \func{next}{}} \quad (\code{e pa} \mid \sigma,\gamma,\square) \hookrightarrow^*_{\code{a}} (\code{v} \mid \sigma,\gamma,\square) \quad \code{$evalV($X,pa,$\gamma)$} = \code{l}}{(\code{X -= e (pa,A)} \mid \delta,\sigma,\gamma,\square) \atran{m} (\code{skip \gs{m:A}} \mid \sigma[\code{l} \mapsto \sigma(\code{l}) - \code{v}],\gamma,\square)} \\[6pt]
& \sos{C3a} \quad \frac{\begin{aligned} & \gs{\code{m} = \func{next}{}} \quad (\code{e pa} \mid \sigma,\gamma,\square) \hookrightarrow^*_{\code{a}} (\code{n} \mid \sigma,\gamma,\square) \quad (\code{e$'$ pa} \mid \sigma,\gamma,\square) \hookrightarrow^*_{\code{a}} (\code{v} \mid \sigma,\gamma,\square) \\[-3pt] & \code{$evalV($name,pa,$\gamma)$} = \code{l} \quad \sigma' = \sigma[(\sigma(\code{l})+\code{n}) \mapsto (\sigma(\code{l})+\code{n}) + \code{v}] \end{aligned}}{(\code{name[e] += e$'$ (pa,A)} \mid \sigma,\gamma,\square) \atran{m} (\code{skip \gs{m:A}} \mid \sigma',\gamma,\square)} \\[6pt]
& \sos{C4a} \quad \frac{\begin{aligned} & \gs{\code{m} = \func{next}{}} \quad (\code{e pa} \mid \sigma,\gamma,\square) \hookrightarrow^*_{\code{a}} (\code{n} \mid \sigma,\gamma,\square) \quad (\code{e$'$ pa} \mid \sigma,\gamma,\square) \hookrightarrow^*_{\code{a}} (\code{v} \mid \sigma,\gamma,\square) \\[-3pt] & \code{$evalV($name,pa,$\gamma)$} = \code{l} \quad \sigma' = \sigma[(\sigma(\code{l})+\code{n}) \mapsto (\sigma(\code{l})+\code{n}) - \code{v}] \end{aligned}}{(\code{name[e] -= e$'$ (pa,A)} \mid \sigma,\gamma,\square) \atran{m} (\code{skip \gs{m:A}} \mid \sigma',\gamma,\square)} 
\end{align*} }%

\vspace{.2cm}
\noindent \textbf{Conditional Statement}
A conditional statement begins with evaluation of the condition ($(\code{b pa} \mid \square) \hookrightarrow^*_{\code{b}} (\code{T} \mid \square)$) to true (\sos{I1aT} or false \sos{I1aF}). The true (\sos{I2a}) or false (\sos{I3a}) branch is then executed. Finally, the conditional closes by saving  a true (\sos{I4a}) or false (\sos{I5a}) boolean value to $\delta$ indicating the executed branch (e.g. $\delta[\code{(m,T) $\rightharpoonup$ B}]$).

{\footnotesize
\begin{align*}
&\sos{I1aT} \quad \frac{\gs{\code{m} = \func{next}{}} \quad (\code{b pa} \mid \square) \hookrightarrow^*_{\code{b}} (\code{T} \mid \square)}{(\code{if In b then AP else AQ end (pa,A)} \mid \square) \atran{m} (\code{if In T then AP else AQ end (pa,\gs{m:A})} \mid \square)} \\[6pt]
&\sos{I1aF} \quad \frac{\gs{\code{m} = \func{next}{}} \quad (\code{b pa} \mid \square) \hookrightarrow^*_{\code{b}} (\code{F} \mid \square)}{(\code{if In b then AP else AQ end (pa,A)} \mid \square) \atran{m} (\code{if In F then AP else AQ end (pa,\gs{m:A})} \mid \square)} \\[6pt]
&\sos{I2a} \quad \frac{(\code{AP} \mid \square) \atran{\circ} (\code{AP$'$} \mid \square')}{(\code{if In T then AP else AQ end (pa,A)} \mid \square) \atran{\circ} (\code{if In T then AP$'$ else AQ end (pa,A)} \mid \square')} \\[6pt]
&\sos{I3a} \quad \frac{(\code{AQ} \mid \square) \atran{\circ} (\code{AQ$'$} \mid \square')}{(\code{if In F then AP else AQ end (pa,A)} \mid \square) \atran{\circ} (\code{if In F then AP else AQ$'$ end (pa,A)} \mid \square')} \\[6pt]
&\sos{I4a} \quad \frac{\gs{\code{m} = \func{next}{}}}{(\code{if In T then skip I else AQ end (pa,A)} \mid \delta,\square) \atran{m} (\code{skip \gs{m:A}} \mid \gs{\delta[\code{(m,T) $\rightharpoonup$ B}]},\square)} \\[6pt]
&\sos{I5a} \quad \frac{\gs{\code{m} = \func{next}{}}}{(\code{if In F then AP else skip I end (pa,A)} \mid \delta,\square) \atran{m} (\code{skip \gs{m:A}} \mid \gs{\delta[\code{(m,F) $\rightharpoonup$ B}]},\square)}
\end{align*} }%

\vspace{.2cm}
\noindent \textbf{While Loop} While loops allow a varying number of iterations (with the condition not guaranteed to be invariant). To reverse a loop we record a sequence of pairs onto stack \code{W}, with an element for each iteration. Each pair consists of an identifier and a boolean value that reflects the evaluation of the condition. This sequence is used to determine reverse control flow, meaning these pairs are accessed in reverse order. As the first forward iteration will be the last reverse iteration (the place to stop), we save a false boolean during the first iteration ($\delta[\code{(m,F) $\rightharpoonup$ W}]$) and then true for all others ($\delta[\code{(m,T) $\rightharpoonup$ W}]$).

A loop not yet executed will have an undefined (\emph{und}) mapping in $\beta$. Loops have zero iterations (\sos{W1a}), or begin with a first iteration (\sos{W3a}). This first iteration defines a mapping (\emph{def}) in the while environment ($\beta[\code{Wn $\Rightarrow$ AR}]$) for this  loop. Any further iteration (\sos{W4a}) requires a true boolean value to be saved. The body of a loop is executed for each iteration (\sos{W5a}, where each identifier used is reflected into the while environment via \func{refW}{\code{Wn},\code{AP$'$}} to ensure they outline the execution) before this completes (\sos{W6a}). A loop body is renamed prior to every iteration via the function ($ \func{reL}{\code{AP}}$), which takes the loop body \code{AP} and updates each construct identifier to ensure they remain unique. Finally, the loop itself completes (\sos{W2a}), saving all identifiers used by this loop (retrieved via the function \func{getAI}{$\beta$(\code{Wn})} .

{\footnotesize
\begin{align*}
&\sos{W1a} \hspace{.2cm} \frac{\gs{\code{m} = \func{next}{}} \quad \code{$\beta($Wn$)$} = \emph{und} \quad (\code{b pa} \mid \beta,\square) \hookrightarrow^*_{\code{b}} (\code{F} \mid \beta,\square) }{(\code{while Wn b do AP end (pa,A)} \mid \delta,\beta,\square) \atran{m} (\code{skip \gs{m:A}} \mid \gs{\delta[\code{(m,F) $\rightharpoonup$ W}]},\beta,\square)} \\[6pt]
&\sos{W2a} \hspace{.2cm} \frac{\gs{\code{m} = \func{next}{}} \quad \code{$\beta($Wn$)$} = \emph{def} \quad (\code{b pa} \mid \beta,\square) \hookrightarrow^*_{\code{b}} (\code{F} \mid \beta,\square) }{(\code{while Wn b do AP end (pa,A)} \mid \delta,\beta,\square) \atran{m} (\code{skip \gs{m:A}} \mid \gs{\delta[\code{(m,T) $\rightharpoonup$ W}, \code{(m,C) $\rightharpoonup$ WI}]},\beta[\code{Wn}],\square)} \\[3pt] & \phantom{\text{[W3a]} \quad} \text{where } \gs{\code{C} = \func{getAI}{\beta(\code{Wn})})} 
\end{align*} }

{\footnotesize
\begin{align*}
&\sos{W3a} \hspace{.2cm} \frac{\gs{\code{m} = \func{next}{}} \quad \code{$\beta($Wn$)$} = \emph{und} \quad (\code{b pa} \mid \beta,\square) \hookrightarrow^*_{\code{b}} (\code{T} \mid \beta,\square) }{(\code{S} \mid \delta,\beta,\square) \atran{m} (\code{while Wn T do AP$'$ end (pa,\gs{m:A})} \mid \gs{\delta[\code{(m,F) $\rightharpoonup$ W}]},\beta[\code{Wn $\Rightarrow$ AR}],\square)} \\[3pt] & \phantom{\text{[W3a]} \quad} \text{where }\code{S} = \code{while Wn b do AP end (pa,A)} \text{, } \code{AP$'$} = \func{reL}{\code{AP}} \\[-2pt] & \phantom{\text{[W3a]} \quad} \text{and } \code{AR} = \code{while Wn b do AP$'$ end (pa,m:A)} \\[6pt]
&\sos{W4a} \hspace{.2cm} \frac{\gs{\code{m} = \func{next}{}} \quad \code{$\beta($Wn$)$} = \emph{def} \quad (\code{b pa} \mid \beta,\square) \hookrightarrow^*_{\code{b}} (\code{T} \mid \beta,\square)}{(\code{S} \mid \delta,\beta,\square) \atran{m} (\code{while Wn T do AP$'$ end (pa,\gs{m:A})} \mid \gs{\delta[\code{(m,T) $\rightharpoonup$ W}]},\beta[\code{Wn $\Rightarrow$ AR}],\square)} \\[3pt] & \phantom{\text{[W3a]} \quad} \text{where }\code{S} = \code{while Wn b do AP end (pa,A)} \text{, } \code{AP$'$} = \func{reL}{\code{AP}} \\[-2pt] & \phantom{\text{[W3a]} \quad} \text{and } \code{AR} = \code{while Wn b do AP$'$ end (pa,m:A)} \\[6pt]
&\sos{W5a} \hspace{.1cm} \frac{\code{$\beta($Wn$)$} = \emph{def} \quad (\code{AP} \mid \delta,\beta,\square) \atran{\circ} (\code{AP$'$} \mid \delta',\beta',\square')}{(\code{while Wn T do AP end (pa,A)} \mid \delta,\beta,\square) \atran{\circ} (\code{while Wn T do AP$'$ end (pa,A)} \mid \delta',\beta'',\square')} \\[3pt] & \phantom{\text{[W3a]} \quad} \text{where }  \beta'' = \beta'[\func{refW}{\code{Wn},\code{AP$'$}}] \\[6pt]
&\sos{W6a} \hspace{.1cm} \frac{\code{$\beta($Wn$)$} = \code{while Wn b do AP end (pa,A)}}{(\code{while Wn T do skip I end (pa,A)} \mid \delta,\beta,\square) \atran{}_s (\code{while Wn b do AP end (pa,A)} \mid \delta,\beta,\square)}
\end{align*} }

\vspace{.2cm}
\noindent \textbf{Block} The body of a block executes (\sos{B1a}) before this reaches skip and the block closes (\sos{B2a}) \ch{without} saving any reversal information.

{\footnotesize
\begin{align*}
&\sos{B1a} \quad \frac{(\code{AP} \mid \square) \atran{\circ} (\code{AP$'$} \mid \square')}{(\code{begin b1 AP end} \mid \square) \atran{\circ} (\code{begin b1 AP$'$ end} \mid \square')}  \\[6pt] &\sos{B2a} \quad \frac{}{(\code{begin b1 skip I end} \mid \square) \rightarrow_s (\code{skip} \mid \square)} 
\end{align*} }%

\vspace{.2cm}
\noindent \textbf{Variable, Procedure and Array Declaration} Variable declaration (\sos{L1a}) creates a corresponding mapping in the variable environment ($\gamma[(\code{X},\code{Bn}) \Rightarrow \code{l}]$) and data store ($\sigma[\code{l} \mapsto \code{v}]$). Procedure declaration (\sos{L2a}) inserts a mapping into the procedure environment ($\mu[\code{Pn} \Rightarrow \code{(n,AP)}]$). Array declaration (\sos{L3a}) creates a mapping in the variable environment ($\gamma[(\code{name},\code{Bn}) \Rightarrow \code{l}]$), and initialises a block of consecutive memory locations for its elements ($\sigma[\code{l} \mapsto \code{l1},\{\code{l1},\cdots,\code{l1}\code{+n}\} \mapsto 0]$). In each case, no reversal information is required. 

{\footnotesize
\begin{align*}
&\sos{L1a} \quad \frac{\gs{\code{m} = \func{next}{}} \quad (\code{e pa} \mid \sigma,\gamma,\square) \hookrightarrow^*_{\code{a}} (\code{v} \mid \sigma,\gamma,\square) \quad \func{nextLoc}{\sigma} = \code{l} \quad \code{pa} = \code{Bn:pa$'$} }{(\code{var X = e (pa,A)} \mid \sigma,\gamma,\square) \atran{m} (\code{skip \gs{m:A}} \mid \sigma[\code{l} \mapsto \code{v}],\gamma[(\code{X},\code{Bn}) \Rightarrow \code{l}],\square)} \\[6pt]
&\sos{L2a} \quad \frac{\gs{\code{m} = \func{next}{}}}{(\code{proc Pn n is AP (pa,A)} \mid \mu,\square) \atran{m} (\code{skip \gs{m:A}} \mid \mu[\code{Pn} \Rightarrow \code{(n,AP)}],\square)} \\[6pt]
& \sos{L3a} \quad \frac{\begin{aligned} & \gs{\code{m} = \func{next}{}} \quad \func{nextLoc}{\sigma} = \code{l} \quad \func{nextLocBlock}{\sigma,n} = \code{l1} \quad \code{pa} = \code{Bn:pa$'$} \\[-3pt]& \sigma' = \sigma[\code{l} \mapsto \code{l1},\{\code{l1},\ldots,\code{l1}\code{+n}\} \mapsto 0] \quad  \gamma' = \gamma[(\code{name},\code{Bn}) \Rightarrow \code{l}] \end{aligned}}{(\code{arr[n] name (pa,A)} \mid \sigma,\gamma,\square) \atran{m} (\code{skip \gs{m:A}} \mid \sigma',\gamma',\square)} 
\end{align*} }%

\vspace{.2cm}
\noindent \textbf{Variable, Procedure and Array Removal} Variable removal (\sos{H1a}) erases the mapping for this version of a variable ($\gamma[(\code{X},\code{Bn})]$), losing the final value it holds. We therefore save this value and the next identifier onto $\delta$ via $\delta[\code{(m,$\sigma($l$)$) $\rightharpoonup$ X}]$. Procedure removal (\sos{H2a}) deletes the mapping from the procedure environment ($\mu[\code{Pn}]$), though this can be recovered completely from the inverse version and only an identifier is required. Array removal (\sos{H3a}) destroys the mapping between this array and the memory locations it has used ($\sigma[\code{l} \mapsto \code{0},\{\code{l1},\cdots,\code{l1}\code{+n}\} \mapsto 0]$), and loses each of the final values held by the array elements. Each final value is saved onto $\delta$ alongside the same identifier ($\delta[\code{\{(\code{m},$\sigma$(\code{l1}),\ldots,(\code{m},$\sigma(\code{l1}\code{+n})$\} $\rightharpoonup$ X}]$).

{\footnotesize
\begin{align*}
&\sos{H1a} \quad \frac{\gs{\code{m} = \func{next}{}} \quad \code{pa} = \code{Bn:pa$'$} \quad \code{$\gamma($X,Bn$)$} = \code{l}}{(\code{remove X = e (pa,A)} \mid \delta,\sigma,\gamma,\square) \atran{m} (\code{skip \gs{m:A}} \mid \gs{\delta[\code{(m,$\sigma($l$)$) $\rightharpoonup$ X}]},\sigma[\code{l} \mapsto \code{0}],\gamma[(\code{X},\code{Bn})],\square)} \\[6pt]
&\sos{H2a} \quad \frac{\gs{\code{m} = \func{next}{}}  \quad \code{$\mu$(Pn)} = \emph{def}}{(\code{remove Pn n is AP (pa,A)} \mid \mu,\square) \atran{m} (\code{skip \gs{m:A}} \mid \mu[\code{Pn}],\square)} \\[6pt]
&\sos{H3a} \quad \frac{\begin{aligned} & \gs{\code{m} = \func{next}{}}  \quad \code{pa} = \code{Bn:pa$'$} \quad \code{$\gamma($name,Bn$)$} = \code{l} \quad \code{l1} = \sigma(\code{l}) \\[-3pt]& \gs{\delta' = \delta[\code{\{(\code{m},$\sigma$(\code{l1}),\ldots,(\code{m},$\sigma(\code{l1}\code{+n})$\} $\rightharpoonup$ X}]} \quad \sigma' = \sigma[\code{l} \mapsto \code{0},\{\code{l1},\cdots,\code{l1}\code{+n}\} \mapsto 0] \quad \gamma' = \gamma[(\code{name},\code{Bn})] \end{aligned}}{(\code{remove arr[n] name (pa,A)} \mid \delta,\sigma,\gamma,\square) \atran{m} (\code{skip (pa,\gs{m:A})} \mid \delta',\sigma',\gamma',\square)} 
\end{align*} }%

\vspace{.2cm}
\noindent \textbf{Procedure Call} The opening (\sos{G1a}) and closing (\sos{G3a}) of a procedure call do not lose any information that is not recoverable, meaning only an identifier is required in each case. A renamed version of the procedure body (with construct identifiers and stack names renamed to remain unique, and produced by the function \func{reP}{}) is stored into $\mu$ when opening ($\mu[\code{Cn} \Rightarrow \code{(n,AP$'$)}]$), and is removed when closed ($\mu[\code{Cn}]$). Procedures are evaluated to a unique construct identifier via the function \func{evalP}{\code{n},\code{pa}}, which takes the procedure name \code{n} used in the code and the statement path \code{pa}. To ensure identifiers are not lost when the mapping is removed, these are reflected to the procedure environment (via the function \func{refC}{\code{Cn},\code{AP$'$}}) and then saved to $\delta$ with $\delta[(\code{m,$getAI(\mu($Cn$)))$ $\rightharpoonup$ Pr}]$).

{\footnotesize
\begin{align*}
&\sos{G1a} \quad \frac{\gs{\code{m} = \func{next}{}} \quad \func{evalP}{\code{n,pa}} = \code{Pn} \quad \mu(\code{Pn}) = \code{(n,AP)} \quad \func{reP}{\code{AP},\code{Cn}} = \code{AP$'$}}{(\code{call Cn n (pa,A)} \mid \mu,\square) \atran{m} (\code{runc Cn AP$'$ end \gs{m:A}} \mid \mu[\code{Cn} \Rightarrow \code{(n,AP$'$)}],\square)} \\[6pt]
&\sos{G2a} \quad \frac{(\code{AP} \mid \mu,\square) \atran{\circ} (\code{AP$'$} \mid \mu',\square')}{(\code{runc Cn AP end A} \mid \mu,\square) \atran{\circ} (\code{runc Cn AP$'$ end A} \mid \mu'[\emph{refC}(\code{Cn},\code{AP$'$})],\square')} \\[6pt]
&\sos{G3a} \hspace{.2cm} \frac{\gs{\code{m} = \func{next}{}} \quad \code{$\mu$(Cn)} = \emph{def}}{(\code{runc Cn skip I end A} \mid \delta,\mu,\square) \atran{m} (\code{skip \gs{m:A}} \mid \gs{\delta[\code{(m,$getAI(\mu($Cn$)))$ $\rightharpoonup$ Pr}]},\mu[\code{Cn}],\square)}
\end{align*} }%

\vspace{.2cm}
\noindent \textbf{Sequential and Parallel Composition} A sequentially composed statement executes (\sos{S1a}) until it reaches skip, before this skip statement is removed (\sos{S2a}) (\code{I} represents a potential identifier stack). Programs in parallel interleave their execution, with a step made via the left (\sos{P1a}) or the right (\sos{P2a}) program, before both are skip and the composition can close (\sos{P3a}). 

{\footnotesize
\begin{align*}
&\sos{S1a} \quad \frac{(\code{AS} \mid \square) \atran{\circ} (\code{AS$'$} \mid \square')}{(\code{AS; AP} \mid \square) \atran{\circ} (\code{AS$'$; AP} \mid \square')} & \sos{S2a}& \quad \frac{}{(\code{skip I; AP} \mid \square) \rightarrow_s (\code{AP} \mid \square)} \\[6pt]
&\sos{P1a} \hspace{.2cm} \frac{(\code{AP} \mid \square) \atran{\circ} (\code{AP$'$} \mid \square')}{(\code{AP par AQ} \mid \square) \atran{\circ} (\code{AP$'$ par AQ} \mid \square')} &\sos{P2a}& \hspace{.2cm} \frac{(\code{AQ} \mid \square) \atran{\circ} (\code{AQ$'$} \mid \square')}{(\code{AP par AQ} \mid \square) \atran{\circ} (\code{AP par AQ$'$} \mid \square')} \\[6pt]
&\sos{P3a} \hspace{.2cm} \frac{}{(\code{skip I$_1$ par skip I$_2$} \mid \square) \rightarrow_s (\code{skip} \mid \square)} 
\end{align*} }%

\subsection{Annotated Odd-Even Transposition Sort}\label{annotated-OETS}

\begin{figure}[t!]
{\small \begin{lstlisting}[xleftmargin=2.0ex,mathescape=true,multicols=2,basicstyle=\scriptsize]
$\kw{begin b1.0}$
 $\kw{arr[5]}~ \code{l} \quad{\color{blue} ~\langle 0 \rangle};$
 $\code{l[0] = 7} \quad{\color{blue} ~\langle 1 \rangle};$ $\code{l[1] = 3}\quad {\color{blue} ~\langle 2 \rangle};$
 $\code{l[2] = 4} \quad{\color{blue} ~\langle 3 \rangle};$ $\code{l[3] = 1}\quad {\color{blue} ~\langle 4 \rangle};$
 $\code{l[4] = 6} \quad{\color{blue} ~\langle 5 \rangle};$
 $\code{count = 0} \quad{\color{blue} ~\langle 6 \rangle};$
 $\kw{while w1.0}~ \code{(count1 < 4)}~\kw{do}$ 
  ${\color{green} \code{//even phase}}$   
  $\kw{par \{}$
   $\kw{if i2.0}~ \code{(l[0] > l[1])}~\kw{then}$
    $\kw{begin b2.0}$
     $\kw{var}~ \code{temp = 0} \quad{\color{blue} ~\langle 35,10 \rangle};$
     $\code{temp = l[0]} \quad{\color{blue} ~\langle 36,12 \rangle};$
     $\code{l[0] = l[1]} \quad{\color{blue} ~\langle 38,14 \rangle};$
     $\code{l[1] = temp} \quad{\color{blue} ~\langle 40,17 \rangle};$
     $\kw{remove}~ \code{temp = 0} \quad{\color{blue} ~\langle 42,18 \rangle}$
    $\kw{end}$
   $\kw{end} \quad{\color{blue} ~\langle 70,68,61,58,45,33,21,8 \rangle}$
  $\kw{\}}$ $\kw{\{}$
   $\kw{if i3.0}~ \code{(l[2] > l[3])}~\kw{then}$
    $\kw{begin b3.0}$
     $\kw{var}~ \code{temp = 0} \quad{\color{blue} ~\langle 37,11 \rangle};$
     $\code{temp = l[2]} \quad{\color{blue} ~\langle 39,13 \rangle};$
     $\code{l[2] = l[3]} \quad{\color{blue} ~\langle 41,15 \rangle};$
     $\code{l[3] = temp} \quad{\color{blue} ~\langle 43,16 \rangle};$
     $\kw{remove}~ \code{temp = 0} \quad{\color{blue} ~\langle 44,19 \rangle}$
    $\kw{end}$
   $\kw{end} \quad{\color{blue} ~\langle 71,69,60,59,46,34,20,9 \rangle}$ 
  $\kw{\};}$ 
  ${\color{green} \code{//odd phase}}$       
  $\kw{par \{}$
   $\kw{if i4.0}~ \code{(l[1] > l[2])}~\kw{then}$
    $\kw{begin b4.0}$
     $\kw{var}~ \code{temp = 0} \quad{\color{blue} ~\langle 25 \rangle};$
     $\code{temp = l[1]} \quad{\color{blue} ~\langle 26 \rangle};$ 
     $\code{l[1] = l[2]} \quad{\color{blue} ~\langle 27 \rangle};$
     $\code{l[2] = temp} \quad{\color{blue} ~\langle 28 \rangle};$
     $\kw{remove}~ \code{temp = 0} \quad{\color{blue} ~\langle 29 \rangle}$
    $\kw{end}$
   $\kw{end} \quad{\color{blue} ~\langle 73,72,65,63,50,47,30,23 \rangle}$ 
  $\kw{\}}$ $\kw{\{}$
   $\kw{if i5.0}~ \code{(l[3] > l[4])}~\kw{then}$
    $\code{begin b5.0}$
     $\kw{var}~ \code{temp = 0} \quad{\color{blue} ~\langle 49 \rangle};$
     $\code{temp = l[3]} \quad{\color{blue} ~\langle 51 \rangle};$ 
     $\code{l[3] = l[4]} \quad{\color{blue} ~\langle 52 \rangle};$
     $\code{l[4] = temp} \quad{\color{blue} ~\langle 53 \rangle};$
     $\kw{remove}~ \code{temp = 0} \quad{\color{blue} ~\langle 54 \rangle}$
    $\kw{end}$
   $\kw{end} \quad{\color{blue} ~\langle 75,74,64,62,55,48,24,22 \rangle}$
  $\kw{\}};$   
  $\code{count += 1} \quad{\color{blue} ~\langle 76,66,56,31 \rangle}$
 $\kw{end} \quad{\color{blue} ~\langle 77,67,57,32,7 \rangle};$ 
 
 ${\color{green} \code{//array removal}}$   
 $\kw{remove arr[5]}~ \code{l} \quad{\color{blue} ~\langle 78 \rangle}$
$\kw{end}$
\end{lstlisting} }
\caption{Executed annotated Odd-even Transposition Sort. Empty else branches and paths are omitted.}
\label{fig:odd-even-sort-ann}
\end{figure}

\newcommand{\element}[1]{\code{$\langle$#1$\rangle$}}
\newcommand{\gap}{\hspace{.2cm}}

\begin{figure}[t]
\begin{minipage}{0.9\textwidth}
\vspace{15pt}
      \centering
       {\footnotesize \begin{tabular}{c | c | c | c | c | c }
\code{temp}			&	\code{count}	&	\code{l}			&	\code{B}			&	\code{W}	\\ \hline
\element{54,7} \gap \element{29,7}	&	\element{6,0}	&	\element{78,7}	\gap \element{52,7} \gap \element{27,7} \gap \element{4,0}	& \element{75,T} \gap	\element{61,F} \gap \element{24,F}	&	\element{77,T}	\\
\element{51,0} \gap \element{26,0}	&	\phantom{\element{66,2}}	&	\element{78,6}	\gap \element{43,4} \gap \element{17,3} \gap \element{3,0}	& \element{73,F} \gap	\element{55,T} \gap \element{21,T}	&	\element{67,T}	\\
\element{44,7} \gap \element{19,4}	&	\phantom{\element{56,1}}	&	\element{78,4}	\gap \element{41,7} \gap \element{16,1} \gap \element{2,0}	& \element{71,F} \gap	\element{50,F} \gap \element{20,T}	&	\element{57,T}	\\
\element{42,3} \gap \element{18,7}	&	\phantom{\element{31,0}}	&	\element{78,3}	\gap \element{40,1} \gap \element{15,4} \gap \element{1,0}	& \element{70,F} \gap	\element{46,T} \gap \phantom{\element{20,T}}	&	\element{32,T}	\\
\element{39,0} \gap \element{13,0}	&	\phantom{\element{6,0}}	&	\element{78,1}	\gap \element{38,3} \gap \element{14,7} \gap \phantom{\element{1,0}}	& \element{65,F} \gap	\element{45,T} \gap \phantom{\element{20,T}}	&	\element{7,F}	\\
\element{36,0} \gap \element{12,0}	&		&	\element{53,6}	\gap \element{28,1} \gap \element{5,0} \gap \phantom{\element{1,0}}	& \element{64,F} \gap	\element{30,T} \gap \phantom{\element{10,4}}	&	\phantom{\element{9}}	\\
        \end{tabular} }
\end{minipage}%
\caption{Populated auxiliary store after an execution of annotated Odd-Even Transposition Sort. The order of stack elements is read from top to bottom, beginning in the leftmost column. For example, $\langle$78,7$\rangle$ is the head of the stack \code{l}, while $\langle$1,0$\rangle$ is the bottom element. 
}
\label{fig:odd-even-sort-aux}
\end{figure}

With the operational semantics of annotated programs introduced, we now present an example of one step of annotated execution. Recall the sorting program introduced in Section~\ref{sec:sort-orig-sec}, and the executed annotated version of this program in Figure~\ref{fig:odd-even-sort-ann}. One important step of this execution is part of the array declaration that initialises the second element of \code{l} to 3, namely \code{l[1] = 3 $\langle 2 \rangle$} as on line~3. Firstly, we consider the full program \code{AP} where this statement is executed next, namely \code{AP} = \code{begin b1 l[1] = 3 (pa,$\langle~\rangle$); AQ end}, where \code{AQ} is the rest of the program (lines 4--56) and \code{pa} is the statement path \code{b1.0}. The inference tree of this step is given below.
 
{\footnotesize
$$ \begin{prooftree}
\Hypo{\code{m} = \func{next}{}}
\Hypo{\code{e pa} \mid \sigma,\gamma,\square) \hookrightarrow^*_{\code{a}} (\code{v} \mid \sigma,\gamma,\square)}
\Hypo{\code{$evalV($X,pa,$\gamma)$} = \code{l}}
\Infer3[\sos{D2a}]{(\code{l[1] = 3 (pa,$\langle\rangle$);} \mid \delta,\sigma,\square) \atran{m} (\code{skip $\langle\code{m}\rangle$;} \mid \delta',\sigma',\square)}  
\Infer1[\sos{S1a}]{(\code{l[1] = 3 (pa,$\langle\rangle$); AQ} \mid \delta,\sigma,\square) \atran{m} (\code{skip $\langle\code{m}\rangle$; AQ} \mid \delta',\sigma',\square)}  
\Infer1[\sos{B1a}]{(\code{begin b1 l[1] = 3 (pa,$\langle\rangle$); AQ end;} \mid \delta,\sigma,\square) \atran{m} (\code{begin b1 skip $\langle\code{m}\rangle$; AQ end;} \mid \delta',\sigma',\square)}           
\end{prooftree} $$ }

The program state produced after execution of this step is such that $\sigma'$ = $\sigma[(\sigma(\code{l})+\code{n}) \mapsto \code{v}]$ and $\delta'$ = $\delta[\code{(m,$\sigma$($\sigma$(\code{l})+\code{n})) $\rightharpoonup$ l)}]$. This shows that the old value of this memory location (namely \code{0}) is saved to $\delta$. This process can be repeated for all steps of this execution.

\begin{example} \label{sec:sort-ann}
Figure~\ref{fig:odd-even-sort-ann} shows an executed annotated version of the program discussed in  Section~\ref{sec:sort-orig-sec}. This annotated version is produced via \func{ann}{P}, where \code{P} is the program in Section~\ref{sec:sort-orig-sec}. As can be seen, the main difference here is that statements now contain identifier stacks. The interleaving order of this program can be read from the identifier stacks within each statement. After this execution, the array is sorted into ascending order, namely \code{l = [1,3,4,6,7]}.

As the program executes, the auxiliary store is populated as shown in Figure~\ref{fig:odd-even-sort-aux}. For example, the pair $\langle 29,7 \rangle$ on the stack \code{temp} indicates that the statement using identifier \code{29} (line 38 in 
Figure~\ref{fig:odd-even-sort-ann}) overwrites the variable \code{temp} and the value \code{7}, which be lost in traditional execution, is now saved. Not all previous values of \code{count} are saved as most assignments to it are constructive, with values re-calculated during reversal. The identifiers \code{76,66,56} and \code{31} are associated to \code{count += 1} to capture interleaving order, but do not appear in $\delta$. $\myqed$
\end{example}

%

%

\subsection{Proving Correctness of Annotation}
Our first result compares the execution of an original and the corresponding annotated program. The only difference is that the annotated execution also populates the auxiliary store. Theorem~\ref{theorem:ann-result} states the annotation result.


\begin{theorem}{(Annotation result)} \label{theorem:ann-result}
Let \textup{\code{P}} be an original program and  \textup{\code{AP}} be the corresponding annotated version \textup{\code{$ann($P$)$}}. Further let \textup{\code{$\square$}} be the tuple \{$\sigma$,$\gamma$,$\mu$,$\beta$\} of initial program state environments, and $\delta$ be the initial auxiliary store.

If an execution $ (\textup{\code{P} $\mid$ $\delta$, $\square$}) \hookrightarrow^* (\textup{\code{skip} $\mid$ $\delta$, $\square'$}) $ exists, for some program state $\square'$, then there exists an annotated execution $ (\textup{\code{AP} $\mid$ $\delta$, $\square_1$}) \atran{\circ}^* (\textup{\code{skip I} $\mid$ $\delta'$, $\square_1'$}) $ for some \textup{\code{I}}, $\square_1$, $\square_1'$ such that $\square_1 \approx \square$ and $\square_1' \approx \square'$, and some auxiliary store $\delta'$.
\end{theorem}

\begin{proof}
This is via induction on the length of the execution of an original program \code{P}, namely $(\code{P} \mid \delta,\square) \otran^* (\code{skip} \mid \delta,\square')$, and the corresponding execution of the annotated version \code{AP}, namely $(\code{AP} \mid \delta,\square) \atran{\circ}^* (\code{skip} \mid \delta',\square')$. Each rule \sos{R} within the traditional semantics has a corresponding rule \sos{Ra} within the forward semantics (Section~\ref{sec:seman-ann}). We first consider all base cases (executions with transition length 1), namely original executions via the rules \sos{P3}, \sos{D1} and \sos{W1}. Each of the corresponding annotated rules, namely \sos{P3a}, \sos{D1a} and \sos{W1a}, is shown to behave identically with respect to the program state while 
modifying the auxiliary store. Therefore each base case holds.

We briefly describe the induction step. We first assume that Theorem~\ref{theorem:ann-result} holds for all program executions of length $k$ such that $k \geq 1$, namely $(\code{Q} \mid \delta_1,\square_1) \otran^* (\code{skip} \mid \delta_1,\square_1')$. Then we assume our original program has execution length $k$ + $1$, namely $(\code{T} \mid \delta_2,\square_2) \otran^* (\code{skip} \mid \delta_2,\square_2')$. This assumed execution can then be rewritten to state the first step via some transition rule \sos{R} as $(\code{T} \mid \delta_2,\square_2) \otran (\code{T$'$} \mid \delta_2,\square_2'') \otran^* (\code{skip} \mid~\delta_2,\square_2')$.

Considering each rule \sos{R} from Section~\ref{sec:seman-ann}, we can then see that there exists a corresponding transition rule of the annotated execution that behaves identically w.r.t. the store. In some cases, for example \sos{S1a}, \sos{P1a} and \sos{I2a}, there is no difference other than the presence of the auxiliary store (which is not used). All other rules, including \sos{D1a}, \sos{I4a} and \sos{H1a}, behave identically w.r.t the program state while differing on the use of the auxiliary store (as required). Finally the induction hypothesis can be applied to the remaining, shorter execution (namely $(\code{T$'$} \mid \delta_2,\square_2'') \otran^* (\code{skip} \mid \delta_2,\square_2')$). Therefore Theorem~\ref{theorem:ann-result} is valid.
\end{proof}

%% file: Inversion.tex
The previous sections introduced, and proved correct, a modified forward execution that saves information necessary for reversal. In this section, we describe the process of using this information to reverse the execution. 

Inversion produces our inverted version. This is of the same syntax as annotated versions (Figure~\ref{fig:syntax-ann}), but with \code{IP} and \code{IS} representing inverted programs and statements respectively). The function \emph{inv}: $\setOf{AP} \rightarrow \setOf{IP}$ is defined in Figure~\ref{fig:func-inv}, takes an annotated program and produces the corresponding inverted version. The inversion of a destructive assignment produces the same statement. This may be counter-intuitive, however the additional reversal information saved onto $\delta$ allows the statement to be correctly reversed. We do not use RESTORE statements similar to those used by RCC and discussed in Section~\ref{sec:related-work}. Maintaining the original statement both helps our proof of correctness, and the application of our approach to debugging~\cite{JH2019}, helping a debugger to understand the current position within the reversal more easily. It is important to note here that the inverted version is a forward-executing program, simulating reversal by beginning in the final program state and restoring to the initial state. To achieve this, an inverse execution must use both the identifiers and lost information.

\begin{figure}[!t]
{\footnotesize \begin{align*}
&inv(\texttt{$\varepsilon$}) = \texttt{$\varepsilon$} \\
&inv(\texttt{AS;AP}) = inv(\texttt{AP}); inv(\texttt{AS}) \\
&inv(\texttt{AP par AQ}) = inv(\texttt{AP}) \texttt{ par } inv(\texttt{AQ}) \\
&inv(\texttt{skip I}) = \texttt{skip I} \\
&inv(\texttt{X = e (pa,A)}) = \texttt{X = e (pa,A)} \\
&inv(\texttt{X += e (pa,A)}) = \texttt{X -= e (pa,A)} \\
&inv(\texttt{X -= e (pa,A)}) = \texttt{X += e (pa,A)} \\
&inv(\texttt{if In b then AP else AQ end (pa,A)}) = \\&\phantom{==} \texttt{if In b then $inv($AP$)$ else $inv($AQ$)$ end (pa,A)} \\
&inv(\texttt{while Wn b do AP end (pa,A)}) = \texttt{while Wn b do $inv($AP$)$ end (pa,A)} \\
&inv(\texttt{begin Bn AP end}) = \texttt{begin Bn $inv($AP$)$ end} \\
&inv(\texttt{var X = e (pa,A)}) = \texttt{remove X = e (pa,A)} \\
&inv(\texttt{arr[n] name (pa,A)}) = \texttt{remove arr[n] name (pa,A)} \\
&inv(\texttt{proc Pn n is AP end (pa,A)}) = \texttt{remove Pn n is $inv($AP$)$ end (pa,A)} \\
&inv(\texttt{call Cn n (pa,A)}) = \texttt{call Cn n (pa,A)} \\
&inv(\texttt{runc Cn AP end A}) = \texttt{runc Cn $inv($AP$)$ A} \\
&inv(\texttt{remove Pn n is AP end (pa,A)}) = \texttt{proc Pn n is $inv($AP$)$ end (pa,A)} \\
&inv(\texttt{remove arr[n] name (pa,A)}) = \texttt{arr[n] name (pa,A)} \\
&inv(\texttt{remove X = e (pa,A)}) = \texttt{var X = e (pa,A)} 
\end{align*} }
\vspace{-.4cm}
\caption{Partial definition of the Inversion function $inv$, shown in full in Appendix~\ref{full-inv-appen}.}
\label{fig:func-inv}
\end{figure} 

\subsection{Using Captured Interleaving Order}
All interleaving decisions when running an inverted program are made using the identifiers assigned to statements during forward execution. Using the execution sequence and beginning with the statement with the highest previously used identifier (via \func{previous}{}), the reverse execution is determined by using identifiers in descending order. 

\begin{example} \label{ex:inv-iden}
Recall Example~\ref{ex:fwd-iden-use} of identifier use. The reverse execution begins with \func{previous}{} = \code{9} (the final identifier used during forward execution). The first execution step is either opening the inverted while loop (identifier \code{8}) or the assignment \code{Y = 12} (identifier \code{9}). Identifiers used in descending order determine which must happen, namely the assignment. This process repeats until all statements (and identifiers) have been inverted. $\myqed$
\end{example}

Using identifiers in this way follows \emph{backtracking order}, where statements are reversed in exactly the inverted order of the forward execution \cite{backtracking}. A possible relaxation of backtracking reversibility is named \emph{causal-consistent} reversibility \cite{VD2004,IP2007}, where a step of an execution can be reversed (out of backtracking order) provided all of its consequences have previously been undone. Currently we support a limited form of this relaxation. Skip steps, which have no effect on the program state, do not require identifiers and can be undone according to causal-consistent reversibility.

\subsection{Restoring Unrecoverable Information}
During an annotated execution we have saved unrecoverable information. Each statement can now be reversed, many of which will make use of this additional information. We note here that as reversal information is used, it is also removed from $\delta$. An initially empty auxiliary store will be empty again after complete inverse execution.

Assignments are reversed by restoring the variable or array element to the previous value it held. Destructive assignments retrieve this from $\delta$, while constructive assignments can easily be reversed (e.g. an increment becomes a decrement). Control flow of conditionals and while loops \ch{are} determined by retrieving boolean values from $\delta$. Removal statements are reversed by re-declaring the variable, array or procedure, using old values saved onto $\delta$. More details of this are given below when presenting operational semantics of inverted programs.  


\subsection{Operational Semantics of Inverted Programs}
We now use SOS to define a transition system describing inverse execution, correspondingly to in Section~\ref{sec:seman-ann}. Each \emph{configuration} is now composed of an inverse program \code{IP} and program state $\square$, written as $(\code{IP} \mid \square)$. Each transition between two configurations is part of the transition relation $\rtran{}$, defined as the least relation generated by the following set of transition rules. Let $\rtran{}^*$ be the reflexive and transitive closure of $\rtran{}$. We now give a syntax-directed, small step semantics that defines inverse execution. We again use $\hookrightarrow^*_{\code{a}}$ and $\hookrightarrow^*_{\code{b}}$ to represent a step of arithmetic and boolean expression evaluation respectively. As before \emph{identifier steps} are transitions using the arrow $\rtran{m}$, where \code{m} is an identifier, while skip steps use the arrow $\rtran{}_s$. Transitions that may, or may not, use an identifier are written using the arrow $\atran{\circ}$. As before, the components of the following rules introduced for reversal are indicated using the colour blue.


There is close correspondence between the rules here and in Section~\ref{sec:seman-ann}. Matching pairs of rules are named accordingly, for example \sos{D1a} and \sos{D1r} with \sos{D1r} being the version of \sos{D1a} adapted for reverse execution. 
The semantics of both sequential and parallel composition and of block statements are omitted. Taking the corresponding rules in Section~\ref{sec:seman-ann}, the annotated versions are produced by changing all arrows to $\rtran{}$ and substituting \code{IP} and \code{IS} for \code{AP} and \code{AS} respectively. We also note that evaluation is only performed during reversal for constructive assignments. No boolean expressions or destructive assignment expressions are evaluated. This potentially saves execution time when compared to reversible languages \cite{CL1986}, each of which evaluates post-conditions. 

\vspace{.2cm}
\noindent \textbf{Assignment} Reversing a destructive assignment restores the variable (\sos{D1r}) or array element (\sos{D2r}) to the old value saved on $\delta$ (e.g. $\sigma[\code{l} \mapsto \code{v}]$). A constructive assignment is reversed by evaluating the expression and incrementing (\sos{C1r}, \sos{C4r}) or decrementing (\sos{C2r}, \sos{C3r}) the variable or array element respectively. 

{\footnotesize
\begin{align*}
&\text{[D1r]} \quad \frac{\gs{\code{A} = \code{m:A$'$} \quad \code{m} = \func{previous}{} \quad \code{$\delta($X$)$} = \code{(m,v):X$'$}} \quad  \code{$evalV($X,pa,$\gamma)$} = \code{l}}{(\code{X = e (pa,A)} \mid \delta,\sigma,\square) \rtran{m} (\code{skip \gs{A$'$}} \mid \gs{\delta[\code{X/X$'$}]},\sigma[\code{l} \mapsto \code{v}],\square)} \\[6pt]
& \text{[D2r]} \quad \frac{\begin{aligned} & \gs{\code{A} = \code{m:A$'$}  \quad \code{m} = \func{previous}{}} \quad (\code{e pa} \mid \sigma,\gamma,\square) \hookrightarrow^*_{\code{a}} (\code{n} \mid \sigma,\gamma,\square) \quad \code{$\delta($name$)$} = \code{(m,v):name$'$} \\[-3pt] & \code{$evalV($name,pa,$\gamma)$} = \code{l} \quad \sigma' = \sigma[(\sigma(\code{l})+\code{n}) \mapsto \code{v}]  \quad \gs{\delta' = \delta[\code{name/name$'$)}]} \end{aligned}}{(\code{name[e] = e$'$ (pa,A)} \mid \sigma,\gamma,\delta,\square) \rtran{m} (\code{skip \gs{A$'$}} \mid \sigma',\gamma,\delta',\square)} \\[6pt]
&\sos{C1r} \quad \frac{\gs{\code{A} = \code{m:A$'$}  \quad \code{m} = \func{previous}{}} \quad (\code{e pa} \mid \sigma,\gamma,\square) \hookrightarrow^*_{\code{a}} (\code{v} \mid \sigma,\gamma,\square) \quad \code{$evalV($X,pa,$\gamma)$} = \code{l}}{(\code{X += e (pa,A)} \mid \delta,\sigma,\gamma,\square) \rtran{m} (\code{skip \gs{A$'$}} \mid \delta,\sigma[\code{l} \mapsto \sigma(\code{l}) + \code{v}],\gamma,\square)} \\[6pt]
&\sos{C2r} \quad \frac{\gs{\code{A} = \code{m:A$'$}  \quad \code{m} = \func{previous}{}} \quad (\code{e pa} \mid \sigma,\gamma,\square) \hookrightarrow^*_{\code{a}} (\code{v} \mid \sigma,\gamma,\square) \quad \code{$evalV($X,pa,$\gamma)$} = \code{l}}{(\code{X -= e (pa,A)} \mid \delta,\sigma,\gamma,\square) \rtran{m} (\code{skip \gs{A$'$}} \mid \delta,\sigma[\code{l} \mapsto \sigma(\code{l}) - \code{v}],\gamma,\square)} \\[6pt]
& \sos{C3r} \quad \frac{\begin{aligned} & \gs{\code{A} = \code{m:A$'$}  \quad \code{m} = \func{previous}{}} \quad (\code{e pa} \mid \sigma,\gamma,\square) \hookrightarrow^*_{\code{a}} (\code{n} \mid \sigma,\gamma,\square) \quad (\code{e$'$ pa} \mid \sigma,\gamma,\square) \hookrightarrow^*_{\code{a}} (\code{v} \mid \sigma,\gamma,\square) \\[-3pt] & \code{$evalV($name,pa,$\gamma)$} = \code{l} \quad \sigma' = \sigma[(\sigma(\code{l})+\code{n}) \mapsto (\sigma(\code{l})+\code{n}) - \code{v}] \end{aligned}}{(\code{name[e] -= e$'$ (pa,A)} \mid \delta,\sigma,\gamma,\square) \rtran{m} (\code{skip \gs{A$'$}} \mid \delta,\sigma',\gamma,\square)}  \\[6pt]
& \sos{C4r} \quad \frac{\begin{aligned} & \gs{\code{A} = \code{m:A$'$}  \quad \code{m} = \func{previous}{}} \quad (\code{e pa} \mid \sigma,\gamma,\square) \hookrightarrow^*_{\code{a}} (\code{n} \mid \sigma,\gamma,\square) \quad (\code{e$'$ pa} \mid \sigma,\gamma,\square) \hookrightarrow^*_{\code{a}} (\code{v} \mid \sigma,\gamma,\square) \\[-3pt] & \code{$evalV($name,pa,$\gamma)$} = \code{l} \quad \sigma' = \sigma[(\sigma(\code{l})+\code{n}) \mapsto (\sigma(\code{l})+\code{n}) + \code{v}] \end{aligned}}{(\code{name[e] += e$'$ (pa,A)} \mid \delta,\sigma,\gamma,\square) \rtran{m} (\code{skip \gs{A$'$}} \mid \delta,\sigma',\gamma,\square)} 
\end{align*} }%

\vspace{.2cm}
\noindent \textbf{Conditional Statement} Opening a conditional does not evaluate the condition and instead retrieves a true (\sos{I1rT}) or false (\sos{I1rF}) value from $\delta$ (e.g. \code{$\delta($B$)$} = \code{(m,T):B$'$}). The true (\sos{I2r}) or false (\sos{I3r}) branch is then executed, before the conditional closes (\sos{I4r} or \sos{I5r}) using an identifier.

{\footnotesize
\begin{align*}
&\sos{I1rT} \quad \frac{\gs{\code{A} = \code{m:A$'$} \quad \code{m} = \func{previous}{} \quad \code{$\delta($B$)$} = \code{(m,T):B$'$}}}{(\code{if In b then IP else IQ end (pa,A)} \mid \delta,\square) \rtran{m} (\code{if In T then IP else IQ end (pa,\gs{A$'$})} \mid \gs{\delta[\code{B/B$'$}]},\square)} \\[6pt]
&\sos{I1rF} \quad \frac{\gs{\code{A} = \code{m:A$'$} \quad \code{m} = \func{previous}{} \quad \code{$\delta($B$)$} = \code{(m,F):B$'$}}}{(\code{if In b then IP else IQ end (pa,A)} \mid \delta,\square) \rtran{m} (\code{if In F then IP else IQ end (pa,\gs{A$'$})} \mid \gs{\delta[\code{B/B$'$}]},\square)} \\[6pt]
&\sos{I2r} \quad \frac{(\code{IP} \mid \square) \rtran{\circ} (\code{IP$'$} \mid \square')}{(\code{if In T then IP else IQ end (pa,A)} \mid \square) \rtran{\circ} (\code{if In T then IP$'$ else IQ end (pa,A)} \mid \square')} \\[6pt]
&\sos{I3r} \quad  \frac{(\code{IQ} \mid \square) \rtran{\circ} (\code{IQ$'$} \mid \square')}{(\code{if In F then IP else IQ end (pa,A)},\square) \rtran{\circ} (\code{if In F then IP else IQ$'$ end (pa,A)} \mid \square')} \\[6pt]
&\sos{I4r} \quad \frac{\gs{\code{A} = \code{m:A$'$} \quad \code{m} = \func{previous}{}}}{(\code{if In T then skip I else IQ end (pa,A)} \mid \square) \rtran{m} (\code{skip \gs{A$'$}} \mid \square)} \\[6pt]
&\sos{I5r} \quad \frac{\gs{\code{A} = \code{m:A$'$} \quad \code{m} = \func{previous}{}}}{(\code{if In F then AP else skip I end (pa,A)} \mid \square) \rtran{m} (\code{skip \gs{A$'$}} \mid \square)}
\end{align*} }%

\vspace{.2cm}
\noindent \textbf{While Loop} A while loop either has zero iterations (\sos{W1r}), or begins with a first iteration (\sos{W3r}). This first iteration inserts a renamed copy (via \func{IreL}{}) of the loop body into $\beta$ ($\beta[\code{Wn $\Rightarrow$ IR}]$). The loop body is then executed step-by-step (\sos{W5r}) until the loop body reaches skip. At this point, the iteration is finished and the loop is reset (\sos{W6r}). All other iterations are performed (\sos{W4r}) while \code{T} remains at the head of the boolean sequence ($\code{$\delta($W$)$} = \code{(m,T):W$'$}$). When this sequence contains an \code{F}, no more iterations are required and the loop closes (\sos{W2r}). 

{\footnotesize
\begin{align*}
&\sos{W1r} \hspace{.2cm} \frac{\gs{\code{m} = \func{previous}{} \quad \code{A} = \code{m:A$'$}} \quad \code{$\beta($Wn$)$} = \emph{und} \quad \gs{\code{$\delta($W$)$} = \code{(m,F):W$'$} }}{(\code{while Wn b do IP end (pa,A)} \mid \delta,\beta,\square) \rtran{m} (\code{skip \gs{A$'$}} \mid \gs{\delta[\code{W/W$'$}]},\beta,\square)} \\[6pt]
&\sos{W2r} \hspace{.2cm} \frac{\gs{\code{m} = \func{previous}{} \quad \code{A} = \code{m:A$'$}} \quad \code{$\beta($Wn$)$} = \emph{def} \quad \gs{\code{$\delta($W$)$} = \code{(m,F):W$'$}} }{(\code{while Wn b do IP end (pa,A)} \mid \delta,\beta,\square) \rtran{m} (\code{skip \gs{A$'$}} \mid \gs{\delta[\code{W/W$'$}]},\beta[\code{Wn}],\square)} \\[6pt]
&\sos{W3r} \hspace{.2cm} \frac{\begin{aligned}  &\gs{\code{m} = \func{previous}{} \quad \code{A} = \code{m:A$'$}} \quad \code{$\beta($Wn$)$} = \emph{und} \quad  \gs{\code{$\delta($W$)$} = \code{(m,T):W$'$} \quad \code{$\delta($WI$)$} = \code{(m,C):WI$'$}} \\[-3pt] & \code{IR} = \code{while Wn T do IP$'$ end (pa,A$'$)} \quad \gs{\delta' = \delta[\code{W/W$'$}, \code{WI/WI$'$}]} \quad \beta' = \beta[\code{Wn $\Rightarrow$ IR}]  \end{aligned}}{(\code{while Wn b do IP end (pa,A)} \mid \delta,\beta,\square) \rtran{m} (\code{while Wn T do IP$'$ end (pa,\gs{A$'$})} \mid\delta',\beta',\square)} \\[3pt] & \phantom{\text{[W3a]} \quad} \text{where } \code{IP$'$} = \func{IreL}{\func{setAI}{\code{IP},\code{C}}} \\[6pt]
&\sos{W4r} \hspace{.2cm} \frac{\begin{aligned} &\gs{\code{m} = \func{previous}{} \quad \code{A} = \code{m:A$'$}} \quad \code{$\beta($Wn$)$} = \emph{def} \quad  \gs{\code{$\delta($W$)$} = \code{(m,T):W$'$} \quad \code{$\delta($WI$)$} = \code{(m,C):WI$'$}} \\[-3pt] &\code{IR} = \code{while Wn b do } \func{IreL}{\code{IP}} \code{ end (pa,A$'$)} \quad \gs{\delta' = \delta[\code{W/W$'$}]}  \quad \beta' = \beta [\code{Wn $\Rightarrow$ IR}] \end{aligned}}{(\code{while Wn b do IP end (pa,A)} \mid \delta,\beta,\square) \rtran{m} (\code{while Wn T do } \func{IreL}{\code{IP}} \code{ end (pa,\gs{A$'$})} \mid \delta',\beta',\square)}  \\[6pt]
&\sos{W5r} \hspace{.2cm} \frac{\code{$\beta($Wn$)$} = \emph{def} \quad (\code{IP} \mid \delta,\beta,\square) \rtran{\circ} (\code{IP$'$} \mid \delta',\beta',\square')}{(\code{while Wn T do IP end (pa,A)} \mid \delta,\beta,\square) \rtran{\circ} (\code{while Wn T do IP$'$ end (pa,A)} \mid \delta',\beta'',\square')} \\[3pt] & \phantom{\text{[W3a]} \quad} \text{where }  \beta'' = \beta'[\func{refW}{\code{Wn},\code{IP$'$}}] \\[6pt]
&\sos{W6r} \hspace{.2cm} \frac{\code{$\beta($Wn$)$} = \code{while Wn b do IP end (pa,A)}}{(\code{while Wn T do skip I end (pa,A)} \mid \delta,\beta,\square) \rtran{}_s (\code{while Wn b do IP end (pa,A)} \mid \delta,\beta,\square)}
\end{align*} }%

\vspace{.2cm}
\noindent \textbf{Variable, Procedure and Array Declaration} A declaration reverses a removal statement. Variable declaration (\sos{L1r}) recreates the necessary mappings and initialises this to the final value held on $\delta$ ($\sigma[\code{l} \mapsto \code{v$'$}]$ and $\gamma[(\code{X},\code{Bn}) \Rightarrow \code{l}]$). Procedure declaration (\sos{L2r}) recreates the basis mapping from the inverted program without using reversal information ($\mu[\code{Pn} \Rightarrow \code{(n,IP)}]$). Array declaration (\sos{L3r}) recreates the array by initialising a block of consecutive memory locations reserving a block of memory and initialising each element to their final values on $\delta$ ($\sigma[\code{l} \mapsto \code{l1},\{\code{l1},\ldots,\code{l1}\code{+n}\} \mapsto \delta(\code{name},\code{m})]$).

{\footnotesize
\begin{align*}
&\sos{L1r} \quad \frac{\gs{\code{A} = \code{m:A$'$} \quad \code{m} = \func{previous}{} \quad \code{$\delta($X$)$} = \code{(m,v$'$):X$'$}} \quad \func{nextLoc}{\sigma} = \code{l} \quad \code{pa} = \code{Bn:pa$'$} }{(\code{var X = e (pa,A)} \mid \delta,\sigma,\gamma,\square) \rtran{m} (\code{skip \gs{A$'$}} \mid \gs{\delta[\code{X/X$'$}]},\sigma[\code{l} \mapsto \code{v$'$}],\gamma[(\code{X},\code{Bn}) \Rightarrow \code{l}],\square)} \\[6pt]
&\sos{L2r} \quad \frac{\gs{\code{A} = \code{m:A$'$} \quad \code{m} = \func{previous}{}}}{(\code{proc Pn n is IP (pa,A)} \mid \mu,\square) \rtran{m} (\code{skip \gs{A$'$}} \mid \mu[\code{Pn} \Rightarrow \code{(n,IP)}],\square)} \\[6pt]
& \sos{L3r} \quad \frac{\begin{aligned} & \gs{\code{A} = \code{m:A$'$} \quad \code{m} = \func{previous}{}} \quad \func{nextLoc}{\sigma} = \code{l} \quad \func{nextLocBlock}{\sigma,n} = \code{l1} \quad \code{pa} = \code{Bn:pa$'$} \\[-3pt]& \sigma' = \sigma[\code{l} \mapsto \code{l1},\{\code{l1},\ldots,\code{l1}\code{+n}\} \mapsto \delta(\code{name},\code{m})]  \quad \gamma' = \gamma[(\code{name},\code{Bn}) \Rightarrow \code{l}]  \end{aligned}}{(\code{arr[n] name (pa,A)} \mid \sigma,\gamma,\square) \rtran{m} (\code{skip \gs{A$'$}} \mid \sigma',\gamma',\square)} 
\end{align*} }%

\vspace{.2cm}
\noindent \textbf{Variable, Procedure and Array Removal} Removal statements undo the effect of declaration statements. 
Variable removal (\sos{H1r}) removes the mappings for this variable and restores the memory location to \code{0} ($\sigma[\code{l} \mapsto \code{0}]$). Procedure removal (\sos{H2r}) simply removes the renamed copy of the body from $\mu$ ($\mu[\code{Pn}]$). Array removal (\sos{H3r}) removes the mapping for this array, restoring each memory location it used to \code{0} ($\sigma[\code{l} \mapsto \code{0},\{\code{l1},\ldots,\code{l1}\code{+n}\} \mapsto 0]$). 

{\footnotesize
\begin{align*}
&\sos{H1r} \quad \frac{\gs{\code{A} = \code{m:A$'$} \quad \code{m} = \func{previous}{}} \quad \code{$\gamma($X,Bn$)$} = \code{l} \quad \code{pa} = \code{Bn:pa$'$}}{(\code{remove X = e (pa,A)} \mid \delta,\sigma,\gamma,\square) \rtran{m} (\code{skip \gs{A$'$}} \mid \delta,\sigma[\code{l} \mapsto \code{0}],\gamma[(\code{X},\code{Bn})],\square)} \\[6pt]
&\sos{H2r} \quad \frac{\gs{\code{A} = \code{m:A$'$} \quad \code{m} = \func{previous}{}} \quad \code{$\mu$(Pn)} = \emph{def}}{(\code{remove Pn n is IP (pa,A)} \mid \mu,\square) \rtran{m} (\code{skip \gs{A$'$}} \mid \mu[\code{Pn}],\square)} \\[6pt]
& \sos{H3r} \quad \frac{\begin{aligned} & \gs{\code{A} = \code{m:A$'$} \quad \code{m} = \func{previous}{}} \quad \code{pa} = \code{Bn:pa$'$} \quad \code{$\gamma($name,Bn$)$} = \code{l} \quad \code{l1} = \sigma(\code{l}) \\[-3pt]& \sigma' = \sigma[\code{l} \mapsto \code{0},\{\code{l1},\ldots,\code{l1}\code{+n}\} \mapsto 0]  \quad \gamma' = \gamma[(\code{name},\code{Bn})]  \end{aligned}}{(\code{remove arr[n] name (pa,A)} \mid \sigma,\gamma,\square) \rtran{m} (\code{skip \gs{A$'$}} \mid \sigma',\gamma',\square)} 
\end{align*} }%

\vspace{.2cm}
\noindent \textbf{Procedure Call} Opening an inverse procedure call (\sos{G1r}) reverses the closure by recreating the copy of the procedure body in $\mu$, using the inverse renaming function \func{IreP}{} to include the appropriate stack names ($\mu[\code{Cn} \Rightarrow \code{(n,IP$'$)}]$). The body is then executed (\sos{G2r}), before the call is closed (\sos{G3a}) by removing the copy from $\mu$ ($\mu[\code{Cn}]$).

{\footnotesize
\begin{align*}
&\sos{G1r} \quad \frac{\gs{\code{A} = \code{m:A$'$} \quad \code{m} = \func{previous}{}} \quad \func{evalP}{\code{n,pa}} = \code{Pn} \quad \mu(\code{Pn}) = \code{(n,IP)} \quad \gs{\code{$\delta($Pr$)$} = \code{(m,C):Pr$'$}}  }{(\code{call Cn n (pa,A)} \mid \delta,\mu,\square) \rtran{m} (\code{runc Cn IP$'$ end \gs{A$'$}} \mid \gs{\delta[\code{Pr/Pr$'$}]}, \mu[\code{Cn} \Rightarrow \code{(n,IP$'$)}],\square)} \\&\phantom{[RG1] \quad} \text{where } \code{IP$'$} = \func{IreP}{\func{setAI}{\code{IP},\code{C}},\code{Cn}} \\[6pt]
&\sos{G2r} \quad \frac{(\code{IP} \mid \mu,\square) \rtran{\circ} (\code{IP$'$} \mid \mu',\square')}{(\code{runc Cn IP end A} \mid \mu,\square) \rtran{\circ} (\code{runc Cn IP$'$ end A} \mid \mu'[\emph{refC(}\code{Cn},\code{IP$'$}\emph{)}],\square')} \\[6pt]
&\sos{G3r} \quad \frac{\gs{\code{A} = \code{m:A$'$} \quad \code{m} = \func{previous}{}} \quad \code{$\mu$(Cn)} = \emph{def}}{(\code{runc Cn skip I end A} \mid \mu,\square) \rtran{m} (\code{skip \gs{m:A}} \mid \mu[\code{Cn}],\square)}
\end{align*} }%

\vspace{.2cm}

\subsection{Inverted Odd-Even Transposition Sort}
Before we discuss how inverted Odd-Even Transposition Sort works, we shall show the reversal of a single step of execution of the annotated version that we have presented in Section~\ref{annotated-OETS}. Recall that this step is the first assignment to the second array element, and the final program states $\delta'$, $\sigma'$ and $\square$ are produced by the forward execution of this statement. These are the starting states for the reversal of this assignment. We consider the inverse program \code{IP} whose next statement to execute (reverse) is this assignment, such that \code{IP} = \code{begin b1 l[1] = 3 (pa,$\langle$2$\rangle$); IQ end}. Note that \code{IQ} is the inverted version of the statements still left to reverse, namely those on lines~55--56. The inference tree is shown below. 

{\footnotesize
$$ \begin{prooftree}
\Hypo{\code{m} = \func{previous}{}}
\Hypo{\delta(\code{l}) = \code{(m,v):X$'$}}
\Hypo{\code{$evalV($X,pa,$\gamma)$} = \code{l}}
\Infer3[\sos{D2r}]{(\code{l[1] = 3 (pa,$\langle$2$\rangle$);} \mid \delta',\sigma',\square) \rtran{m} (\code{skip $\langle\rangle$;} \mid \sigma'',\delta'',\square)}    
\Infer1[\sos{S1a}]{(\code{l[1] = 3 (pa,$\langle$2$\rangle$); IQ} \mid \delta',\sigma',\square) \rtran{m} (\code{skip $\langle\rangle$; IQ} \mid \sigma'',\delta'',\square)}
\Infer1[\sos{B1a}]{(\code{begin b1 l[1] = 3 (pa,$\langle$2$\rangle$); IQ end} \mid \delta',\sigma',\square) \rtran{m} (\code{begin b1 skip $\langle\rangle$; IQ end} \mid \sigma'',\delta'',\square)}           
\end{prooftree} $$ }

This step produces the states $\sigma''$ = $\sigma[(\sigma(\code{l})+\code{n}) \mapsto \code{v}]$  and $\delta'$ = $\delta[\code{name/name$'$)}]$. Recalling that \code{v} = \code{0}, this means that $\sigma''$ and $\delta''$ are equivalent to $\sigma$ and $\delta$ respectively, and therefore that the reversal is correct. 

\begin{figure}[t!]
{\small \begin{lstlisting}[xleftmargin=2.0ex,mathescape=true,multicols=2,basicstyle=\scriptsize]
$\kw{begin b1.0}$   
 $\kw{arr[5]}~ \code{l;} \quad{\color{blue} ~\langle 78 \rangle};$
 $\kw{while w1.0}~ \code{(count1 < 4)}~\kw{do}$ 
  $\code{count -= 1} \quad{\color{blue} ~\langle 76,66,56,31 \rangle};$
  ${\color{green} \code{//odd phase}}$       
  $\kw{par \{}$
   $\kw{if i4.0}~ \code{(l[1] > l[2])}~\kw{then}$
    $\kw{begin b4.0}$
     $\kw{var}~ \code{temp = 0} \quad{\color{blue} ~\langle 29 \rangle};$
     $\code{l[2] = temp} \quad{\color{blue} ~\langle 28 \rangle};$
     $\code{l[1] = l[2]} \quad{\color{blue} ~\langle 27 \rangle};$
     $\code{temp = l[1]} \quad{\color{blue} ~\langle 26 \rangle};$
     $\kw{remove}~ \code{temp = 0} \quad{\color{blue} ~\langle 25 \rangle}$
    $\kw{end}$
   $\kw{end} \quad{\color{blue} ~\langle 73,72,65,63,50,47,30,23 \rangle}$
  $\kw{\}}$ $\kw{\{}$
   $\kw{if i5.0}~ \code{(l[3] > l[4])}~\kw{then}$
    $\code{begin b5.0}$
     $\kw{var}~ \code{temp = 0} \quad{\color{blue} ~\langle 54 \rangle};$
     $\code{l[4] = temp} \quad{\color{blue} ~\langle 53 \rangle};$
     $\code{l[3] = l[4]} \quad{\color{blue} ~\langle 52 \rangle};$
     $\code{temp = l[3]} \quad{\color{blue} ~\langle 51 \rangle};$
     $\kw{remove}~ \code{temp = 0} \quad{\color{blue} ~\langle 49 \rangle}$
    $\kw{end}$
   $\kw{end} \quad{\color{blue} ~\langle 75,74,64,62,55,48,24,22 \rangle}$
  $\kw{\}};$   
  ${\color{green} \code{//even phase}}$   
  $\kw{par \{}$
   $\kw{if i2.0}~ \code{(l[0] > l[1])}~\kw{then}$
    $\kw{begin b2.0}$
     $\kw{var}~ \code{temp = 0} \quad{\color{blue} ~\langle 42,18 \rangle};$
     $\code{l[1] = temp} \quad{\color{blue} ~\langle 40,17 \rangle};$
     $\code{l[0] = l[1]} \quad{\color{blue} ~\langle 38,14 \rangle};$
     $\code{temp = l[0]} \quad{\color{blue} ~\langle 36,12 \rangle};$ 
     $\kw{remove}~ \code{temp = 0} \quad{\color{blue} ~\langle 35,10 \rangle}$
    $\kw{end}$
   $\kw{end} \quad{\color{blue} ~\langle 70,68,61,58,45,33,21,8 \rangle}$
  $\kw{\}}$ $\kw{\{}$
   $\kw{if i3.0}~ \code{(l[2] > l[3])}~\kw{then}$
    $\kw{begin b3.0}$
     $\kw{var}~ \code{temp = 0} \quad{\color{blue} ~\langle 44,19 \rangle};$
     $\code{l[3] = temp} \quad{\color{blue} ~\langle 43,16 \rangle};$
     $\code{l[2] = l[3]} \quad{\color{blue} ~\langle 41,19 \rangle};$
     $\code{temp = l[2]} \quad{\color{blue} ~\langle 39,13 \rangle};$
     $\kw{remove}~ \code{temp = 0} \quad{\color{blue} ~\langle 37,118 \rangle}$
    $\kw{end}$
   $\kw{end} \quad{\color{blue} ~\langle 71,69,60,59,46,34,20,9 \rangle}$
  $\kw{\}}$ 
 $\kw{end} \quad{\color{blue} ~\langle 77,67,57,32 \rangle};$

 ${\color{green} \code{//inverting array declaration}}$
 $\code{count = 0} \quad{\color{blue} ~\langle 6 \rangle};$
 $\code{l[4] = 6} \quad{\color{blue} ~\langle 5 \rangle};$ $\code{l[3] = 1} \quad{\color{blue} ~\langle 4 \rangle};$
 $\code{l[2] = 4} \quad{\color{blue} ~\langle 3 \rangle};$ $\code{l[1] = 3} \quad {\color{blue} ~\langle 2 \rangle};$
 $\code{l[0] = 7} \quad{\color{blue} ~\langle 1 \rangle};$
 $\kw{remove[5]}~ \code{l} \quad{\color{blue} ~\langle 0 \rangle}$
$\kw{end}$

\end{lstlisting} }
\caption{Inverted Odd-Even Transposition Sort. Empty else branches and paths omitted.}
\label{fig:odd-even-sort-inv}
\end{figure}

\begin{example} \label{sec:sort-inv}
Recall the executed annotated version of the sorting program in ~Example~\ref{sec:sort-ann}, which we denote as \code{AP}. The inverted version of this program is produced via \func{inv}{\code{AP}}
and is shown in Figure~\ref{fig:odd-even-sort-inv}.

The overall statement order has been inverted, including within nested constructs. For example, the block is the highest level construct, containing local declarations (reversing forward removal statements), a block body and local removals (reversing declaration statements). The block body still consists of a loop which contains two sequentially composed parallel statements, but with the `odd' phase reversed first. 

Complete execution of the program reverses the annotated execution. The number of iterations of the loop, and the branches executed for all conditional statements, are determined using boolean values (\code{1} and \code{0} used for true and false respectively) from the stacks \code{W} and \code{B}. Pairs stored on the stacks \code{temp}, \code{count} and \code{l} are used to restore variables or array elements to their previous values. For example, the pair \code{(54,7)} indicates the statement using identifier \code{54} overwrote the value \code{7} of the variable \code{temp}. So, the variable \code{temp} is therefore restored to \code{7}. 

As reversal information is used, it is removed from $\delta$. Given that $\delta$ was initially empty, after complete inverse execution $\delta$ is again empty. This means our reversibility is garbage-free. $\myqed$
\end{example}

\subsection{Proving Correctness of Inversion}
We now state our inversion result. Theorem~\ref{theorem:inv-result} says that given the final state produced by the matching annotated execution, the reverse execution restores this to a state equivalent to that of before the forward execution. Equivalence between program states, where one is used in a reverse execution, is defined below (correspondingly to Definition~\ref{def:prog-states}). The inversion result also restores the auxiliary store, showing our reversibility to be garbage-free. 


\begin{definition}{(Equivalent inverse states)} \label{def:equiv-states-aux}
Let $\square$ be the tuple of environments \{$\sigma$,$\gamma$,$\mu$,$\beta$,$\delta$\} and $\square_1$ be the tuple of annotated environments \{$\sigma_1$,$\gamma_1$,$\mu_1$,$\beta_1$,$\delta_1$\}. $\square$ is \emph{equivalent} to $\square_1$, written $\square \approx \square_1$, if and only if:\vspace{-.2cm} 
\begin{enumerate}
\item Same as part 1 in Definition~\ref{def:prog-states}.\vspace{-.4cm} 
\item \func{d}{$\beta$} = \func{d}{$\beta_1$};  $\beta(\textup{\code{Wn}})$ = \textup{\code{AP}}, $\beta_1(\textup{\code{Wn}})$ = \textup{\code{IP}}, \code{IP} = $inv($\code{AP}$)$ for all \code{Wn} $\in$ \func{d}{$\beta$}.\vspace{-.4cm} 
\item \func{d}{$\mu$} = \func{d}{$\mu_1$};  $\mu(\textup{\code{Pn}})$ = (\textup{\code{n,AP}}), $\mu_1(\textup{\code{Pn}})$~=~(\textup{\code{n,IP}}) and \code{IP} = $inv($\code{AP}$)$ for all \code{Pn} $\in$ \func{d}{$\mu$}.\vspace{-.4cm}
\item $\delta$ = $\delta_1$
\end{enumerate}
\end{definition}
 
\begin{theorem} \label{theorem:inv-result}
Let \textup{\texttt{P}} be an original program, \textup{\texttt{AP}} be the corresponding annotated version \textup{\texttt{$ann($P$)$}}, and \code{AP$'$} be the executed version of \code{AP}. Further let \textup{\texttt{$\square$}} be the tuple \{$\sigma$,$\gamma$,$\mu$,$\beta$,$\delta$\} of all initial state environments. 

If $ (\code{P} \mid \square) \hookrightarrow^* (\code{skip} \mid \square') $ for some program state $\square'$, and therefore by Theorem~\ref{theorem:ann-result} the annotated execution $ (\code{AP} \mid \square_1) \atran{\circ}^* (\code{skip I} \mid \square_1') $ exists for some \code{I} and program state $\square_1'$ such that $\square_1' \approx \square'$, then there exists a corresponding inverse execution $ (\func{inv}{\code{AP$'$}} \mid \square_2') \rtran{\circ}^* (\code{skip I$'$} \mid \square_2) $, for some program states $\square_2'$, $\square_2$, such that $\square_2' \approx \square_1'$ and $\square_2 \approx \square_1$. Provided this holds, Definition~\ref{def:equiv-states-aux} states that $\delta_2$ = $\delta_1$, showing the reversal to be garbage-free.
\end{theorem}

Before proving this result, we consider the following challenges:
\begin{enumerate}
\item Parallel composition allows different interleaving orders, \vspace{-.3cm}
\item Partially executed programs may not reach skip when inverted, or may not stop at the corresponding position.
\end{enumerate}

\begin{figure}[t] 
\centering
{\small $\begin{array}{ccccc}
(\code{P} \mid \square) \atran{\circ}^* (\code{skip I} \mid \square') & & ${\Huge $\Rightarrow$}$1 & & (\code{P} \mid \square) \uatran{\circ}{^*} (\code{skip I} \mid \square') \\
& &  \\
\text{? }${\Huge$ \downarrow $}$ & & & & ${\Huge $\Downarrow$}$2 \\
& &  \\
(\func{inv$^+$}{\code{P}} \mid \square_1') \rtran{\circ}^* (\texttt{Q} \mid \square_1) & & 3${\Huge $\Leftarrow$}$ & & (\func{inv$^+$}{\code{P}} \mid \square_1') \urtran{\circ}{^*} (\texttt{Q} \mid \square_1) 
\end{array}$ }
\caption{Diagram representation of proof outline}
\label{diagram-theorem2-proof}
\end{figure}

The first challenge considers different interleaving orders of the same program introduced via our parallel operator. Statements of two programs in parallel may interleave in different orders, producing different behaviours. Some interleavings however may be ``equivalent'', typically as a result of skip steps performed in a different order. This can lead to a mismatch of skip steps between a forward and reverse execution, impacting our proof of correctness. In Section~\ref{subsec:interl} we introduce the notion of \emph{uniform execution} to avoid this issue. 

The second challenge concerns partially executed programs (defined later). When reversing a complete program, we judge success as whether the inverted version executes completely to skip. Imagine proving this via induction, where after the first step of an execution we no longer have a complete execution. The remaining execution will not reverse to skip, but to some intermediate position, from which the first step can be reversed to reach skip. In Section~\ref{subsec:abort} we overcome this by introducing \emph{abort} statements and \emph{skip equivalents}. 

Given the two challenges, our method of proving the reversibility described here to be correct is illustrated in Figure~\ref{diagram-theorem2-proof}. We first convert an arbitrary forward execution into an equivalent uniform execution ($\Rightarrow_1$), then prove our property ($\Downarrow_2$) for uniform executions, before converting back into a standard reverse execution. We now discuss this in more detail.

\subsubsection{Interleaving Parallel Composition} \label{subsec:interl}
The semantics of parallel composition allows multiple different interleaving orders of a given program. Different executions are equivalent when the order of identifier steps is identical while the order of skip steps is different. Since no skip steps alter the program state, each of these executions produces the same final program state (hence executions are equivalent). Inverse skip steps can also be applied in several different orders while preserving an equivalent execution. This results in the potential for constructs to remain ``open'' (so to speak) during an inverse execution for longer than during the forward execution, adding difficultly to the task of matching rule uses between a forward and reverse execution.  

We avoid these difficulties by considering executions as a composition of segments. Firstly let any arbitrary execution be named a \emph{standard execution}. We now introduce a \emph{uniform execution}, where each identifier step is followed by all of the possible closing skip steps. A sequence starting with an identifier step and followed by all its closing skip steps is called  a \emph{segment}. All skip steps within a segment are said to have been `caused' by the preceding identifier step.

\begin{definition}{(Uniform execution)} \label{def:uniform-ex}
A \emph{uniform execution} is a standard program execution where all skip steps are performed as soon as they are available. A step of a uniform execution is represented using$\uatran{\circ}{}$and$\urtran{\circ}{}$for forward and reverse execution respectively. 
\end{definition}

We remark that skip steps within a uniform execution are not interleaved across a parallel composition. This is due to identifier steps executed on one side not being able to make skip steps available on the other. Therefore any two equivalent standard executions (that is two transition sequences where the order of identifier steps is identical, with skip steps potentially different) will have the same uniform version. All skip steps from each sequence will be performed as soon as they are available, which will be identical in both cases. Consider Example~\ref{ex:uniform-exec} that shows a uniform version of a standard execution.

\begin{example} \label{ex:uniform-exec}
Consider the concurrent program in Figure~\subref*{ex:standUni-program} and one possible standard execution captured with identifiers \code{0}--\code{2}. Figure~\subref*{ex:standUni-executions} shows first the order of rule uses under this standard execution (with the identifier it used or that caused it). \sos{S2a} and \sos{B2a} are not executed as soon as available. The uniform execution contains unchanged order of identifier steps (an equivalent execution), but with skip steps performed as soon as possible.  $\myqed$
\end{example}

\newsavebox\exStandUniProgram
\newsavebox\exStandUniEx

\begin{lrbox}{\exStandUniProgram}
\begin{minipage}{0.25\textwidth}
\centering
{\small \begin{lstlisting}[xleftmargin=4.0ex,mathescape=true,basicstyle=\footnotesize]
$\kw{par \{} $
  $\kw{begin b1.0}$
    $\code{X = 20} \quad{\color{blue}\langle \code{1} \rangle} $ 
  $\kw{end}$
$\kw{\} \{}$
  $\code{Y = 2} \quad{\color{blue}\langle \code{0} \rangle} \code{;}$
  $\code{Z = 4} \quad{\color{blue}\langle \code{2} \rangle} $
$\kw{\}}$
\end{lstlisting} }
\end{minipage}
\end{lrbox}

\begin{lrbox}{\exStandUniEx}
\begin{minipage}{0.70\textwidth}

{\small \quad\textbf{Standard Execution}

 \phantom{\quad\quad} \sos{D1a} \quad \sos{D1a} \quad \sos{S2a} \quad \sos{D1a} \quad \sos{B1a} \quad \sos{P3a} \\
 \phantom{\quad\quad} \phantom{\hspace{.1cm}} $\atran{0}$ \hspace{.7cm} $\atran{1}$ \hspace{.5cm} $\atran{}_s$ \hspace{.5cm} $\atran{2}$ \hspace{.6cm} $\atran{}_s$ \hspace{.5cm} $\atran{}_s$}

{\small \quad\textbf{Uniform Execution}

 \phantom{\quad\quad} \sos{D1a} \quad \sos{S2a} \quad \sos{D1a} \quad \sos{B1a} \quad \sos{D1a} \quad \sos{P3a} \\
 \phantom{\quad\quad} \phantom{\hspace{.1cm}} $\atran{0}$ \hspace{.7cm} $\atran{}_s$ \hspace{.5cm} $\atran{1}$ \hspace{.5cm} $\atran{}_s$ \hspace{.6cm} $\atran{2}$ \hspace{.5cm} $\atran{}_s$}

\end{minipage}

\end{lrbox}

\begin{figure}[t]
\subfloat[Small program\label{ex:standUni-program}]{ \usebox\exStandUniProgram }
\subfloat[Standard and uniform execution\label{ex:standUni-executions}]{ \usebox\exStandUniEx }
\caption{The standard execution and a corresponding uniform execution of a given program. Each execution is displayed as an ordered sequence of rule labels, with the arrow from the corresponding transition shown below each. For example, the rule \sos{D1a} above the arrow $\atran{1}$ represents executing the statement \code{X = 20}, while the rule \sos{B1a} above the arrow $\atran{}_s$ represents closing the block on line~4. Intermediate programs and states are omitted.}
\label{ex:standUni}
\end{figure}

One further consideration is of programs that begin with skip steps. This is a special case of a uniform execution, where all skip steps performed at the beginning have not been `caused' by any identifier step. Such an execution performs all initial skip steps in a fixed order, namely from the left side of a parallel before those of the right, and in a depth first manner. This ensures all equivalent executions have a single uniform version. Lemma~\ref{lem:equiv-uniform-fwds} states that for any standard forward execution, there exists a uniform forward execution that has the same behaviour (with respect to the program state).  A proof of this property uses Lemma~\ref{claim:1} to show that statement reordering is valid.

\begin{lemma} \label{claim:1}
Given a sequence of two transitions where each is from the opposite side of a parallel, these two transitions may be reordered in two specific circumstances (depending on the type of each transition). Let \code{P} be an annotated program and $\square$ be the tuple $\allstores$ of initial program state environments.

\begin{enumerate}
\item If $ (\code{P} \mid \square) \atran{}_{S_1} (\code{P$'$} \mid \square) \atran{}_{S_2} (\code{P$'''$} \mid \square) $ exists, for some programs \code{P$'$} and \code{P$'''$}, such that the transition $\atran{}_{S_1}$ is from the opposite side of a parallel composition to the transition $\atran{}_{S_2}$, then there exists the reordered execution $ (\code{P} \mid \square) \atran{}_{S_2} (\code{P$''$} \mid \square) \atran{}_{S_1} (\code{P$'''$} \mid \square) $ for some program \code{P$''$}.

\item If $ (\code{P} \mid \square) \atran{m} (\code{P$'$} \mid \square') \atran{}_{S} (\code{P$'''$} \mid \square') $ exists, for some programs \code{P$'$}, \code{P$'''$} and program state $\square'$, such that the transition $\atran{m}$ is from the opposite side of a parallel composition to the transition $\atran{}_{S}$, then there exists the reordered execution $(\code{P} \mid \square) \atran{}_{S} (\code{P$''$} \mid \square) \atran{m} (\code{P$'''$} \mid \square')$ for some program \code{P$''$}.
\end{enumerate}

\begin{proof}
Recall that skip steps do not change the program state in any way, and the observation that any step from one side of a parallel composition cannot `cause' skip steps on the other side of the same parallel.  

Consider Part~1 of Lemma~\ref{claim:1}. The execution of two skip steps with one on each side of a parallel means that neither skip step is the cause of the other. With one of the steps requiring \sos{P1a} and the other needing \sos{P2a} to occur within its inference tree, the definition of these rules shows the execution of one cannot affect the execution of the other. This means the order of these skip steps cannot change the final program state or structure, meaning the reordering of these is permitted. The intermediate program \code{P$''$} is produced and is as \code{P$'''$} but with the side of the parallel that contains $\atran{}_{S_1}$ still in its initial form. 

Consider Part~2 of Lemma~\ref{claim:1}. With each step being on opposite sides of a parallel, the skip step $\atran{}_s$ cannot have been caused by the identifier step $\atran{m}$. This means that the skip step must have been caused by an identifier step $\atran{k}$ such that $k < m$. As in Part~1, it is therefore guaranteed that both steps were available prior to the execution of either. The rules \sos{P1a} and \sos{P2a} show that the program structure of one side of a parallel cannot be changed via the other side. With the skip step having no effect on the program state, the identifier step $\atran{m}$ and the skip step $\atran{}_s$ can therefore be reordered. The intermediate program \code{P$''$} is as \code{P$'''$} but with the side containing the identifier step $\atran{m}$ still in its original form. Both parts are valid meaning Lemma~\ref{claim:1} holds, as required.
\end{proof}

\end{lemma}

\begin{lemma}{(Equivalent uniform execution)} \label{lem:equiv-uniform-fwds} Given an arbitrary standard forward or reverse program execution, there exists an equivalent uniform execution that modifies the program state equivalently. Let \code{P} be an arbitrary program and $\square$ be the tuple $\allstores$ of initial program state environments.
\begin{enumerate}
\item If $(\code{P} \mid \square) \atran{\circ}^* (\code{skip I} \mid \square')$ for some \code{I} and final program state $\square'$, then there exists an equivalent uniform execution $(\code{P} \mid \square) \uatran{\circ}{^*} (\code{skip I$'$} \mid \square'') $ for \code{I$'$} and program state $\square''$ such that \code{I$'$} = \code{I} and $\square''$ = $\square'$.
\item If $(\code{P} \mid \square) \rtran{\circ}^* (\code{skip I} \mid \square') $ for some \code{I} and final program state $\square'$, then there exists an equivalent uniform execution $(\code{P} \mid \square) \urtran{\circ}{^*} (\code{skip I$'$} \mid \square'') $ for \code{I$'$} and program state $\square''$ such that \code{I$'$} = \code{I} and $\square''$ = $\square'$.
\end{enumerate} 
\end{lemma}

\begin{proof}
A proof of Lemma~\ref{lem:equiv-uniform-fwds} is by induction on the length of the execution using Lemma~\ref{claim:1}.  All executions are assumed to be of complete programs. We consider each of the two possible cases of execution, beginning with the special case of those that start with any number of skip steps. Such a complete execution, written $(\code{P} \mid \square) \atran{}_s^* (\code{P$'$} \mid \square) \atran{\circ}^* (\code{skip I} \mid \square')$, exists such that $(\code{P} \mid \square) \atran{}_s^* (\code{P$'$} \mid \square)$ contains all possible skip steps and that $(\code{P$'$} \mid \square) \atran{\circ}^* (\code{skip I} \mid \square')$ therefore begins with an identifier step. Since this execution is complete, all of the initial skip steps do not have a preceding identifier step that caused them. Therefore the initial skip steps can be reordered as shown in Part~1 of Lemma~\ref{claim:1}, such that all skip steps from the left side of a parallel occur before those from the right. With the execution $(\code{P} \mid \square) \atran{}_s^* (\code{P$'$} \mid \square)$ reordered into a uniform equivalent, application of the induction hypothesis to the remaining shorter execution $(\code{P$'$} \mid \square) \atran{\circ}^* (\code{skip I} \mid \square')$ (guaranteed to begin with an identifier step) shows the existence of a uniform equivalent.

We now consider executions of the form $(\code{P} \mid \square) \atran{m} (\code{P$'$} \mid \square') \atran{}_s^* (\code{P$''$} \mid \square') \atran{m+1} (\code{P$'''$} \mid \square'') \atran{\circ}^* (\code{skip I} \mid \square') $ that begin with an identifier step. To support use of the induction hypothesis, we do not require this execution to be complete. The execution $(\code{P$'$} \mid \square') \atran{}_s^* (\code{P$''$} \mid \square')$ may contain skip steps that were caused by a previous identifier step (using an identifier $k$ such that $k$ $<$ $m$ provided it exists). Using Part~2 of Lemma~\ref{claim:1}, all such skip steps can be reordered to appear prior to the identifier step using $m$. All other skip steps within the execution $(\code{P$'$} \mid \square') \atran{}_s^* (\code{P$''$} \mid \square')$ are guaranteed to have been caused by the transition $\atran{m}$. These cannot have been caused by an identifier rule that has not yet happened, meaning their current position already respects uniformity. The final consideration is of skip steps within $(\code{P$'''$} \mid \square'') \atran{\circ}^* (\code{skip I} \mid \square')$. It is possible for some of these steps to have been caused by any of the preceding identifier steps and are therefore required to be reordered. As before, this can be achieved using Part~2 of Lemma~\ref{claim:1}. From here, the remaining shorter execution beginning with the identifier step using $m+1$ can be shown to be uniform via the induction hypothesis. Therefore Lemma~\ref{lem:equiv-uniform-fwds} holds. 
\end{proof}

Lemma~\ref{lem:equiv-uniform-rev} corresponds to Lemma~\ref{lem:equiv-uniform-fwds} and shows that any uniform reverse execution can be relaxed into an equivalent standard version.

\begin{lemma}{(Equivalent standard execution)} \label{lem:equiv-uniform-rev} Given an arbitrary uniform forwards or reverse program execution, there exists an equivalent standard execution that modifies the program state equivalently. Let \code{P} be an arbitrary program and $\square$ be the tuple $\allstores$ of initial program state environments.
\begin{enumerate}
\item If $(\code{P} \mid \square) \uatran{\circ}{^*} (\code{skip I} \mid \square')$ for some \code{I} and program state $\square'$, then there exists an equivalent standard execution $ (\code{P} \mid \square) \atran{\circ}^* (\code{skip I$'$} \mid \square'') $ for \code{I} and program state $\square''$ such that \code{I$'$} = \code{I} and $\square''$ = $\square'$.
\item If $(\code{P} \mid \square) \urtran{\circ}{^*} (\code{skip I} \mid \square')$ for some \code{I} and program state $\square'$, then there exists an equivalent standard execution $ (\code{P} \mid \square) \rtran{\circ}^* (\code{skip I$'$} \mid \square'') $ for \code{I} and program state $\square''$ such that \code{I$'$} = \code{I} and $\square''$ = $\square'$.
\end{enumerate} 

\begin{proof}
Correspondingly to the proof of Lemma~\ref{lem:equiv-uniform-fwds}. This is via induction on the length of an execution using a version of Lemma~\ref{claim:1} for reverse execution.
\end{proof}
\end{lemma}

\subsubsection{Reversing Partially Executed Programs} \label{subsec:abort}
\begin{figure}[t]
{\small
\begin{align*}
\code{PP} &::= \code{$\varepsilon$} ~|~ \code{PS} ~|~ \code{PP};\code{AP} ~|~ \code{PP par PP} \\
\code{PS} &::= \code{skip I} ~|~  \code{X = E (pa,A)} ~|~  \code{name[E] = E (pa,A)}  ~|~ \texttt{X += E (pa,A)} ~|~ \texttt{X -= E (pa,A)} ~| \\ &\phantom{\code{::=}} \texttt{name[E] += E (pa,A)} ~|~ \texttt{name[E] -= E (pa,A)} ~| \\ &\phantom{\code{::=}} \code{if In B B then AP AP else AP AP end (pa,A)} ~| \\ &\phantom{\code{::=}} \code{if In T B then PP AP else AP AP end (pa,A)} ~| \\ &\phantom{\code{::=}} \code{if In F B then AP AP else PP AP end (pa,A)} ~| \\ &\phantom{\code{::=}}  \code{while Wn B do AP end (pa,A)} ~|~ \code{while Wn T do PP end (pa,A)} ~| \\ &\phantom{\code{::=}}  \code{begin Bn ABB ABB end} ~|~ \code{begin Bn PBB ABB end}  ~| \\ &\phantom{\code{::=}}\code{call Cn n (pa,A)} ~|~   \code{runc Cn AP AP end A} ~|~   \code{runc Cn PP AP end A} \\ 
\code{PBB} &::= \code{ADV};\code{ADA};\code{ADP};\code{AP};\code{ARP};\code{ARA};\code{ARV} ~~|~~ \code{ADA};\code{ADP};\code{AP};\code{ARP};\code{ARA};\code{ARV} ~~|~~ \code{ADP};\code{AP};\code{ARP};\code{ARA};\code{ARV} ~~| \\ &\phantom{\code{::=}} \code{AP};\code{ARP};\code{ARA};\code{ARV} ~~|~~ \code{ARP};\code{ARA};\code{ARV} ~~|~~ \code{ARA};\code{ARV} ~~|~~ \code{ARV} ~~|~~ \code{PP};\code{ARP};\code{ARA};\code{ARV}
\end{align*}}
\vspace{-.4cm}
\caption{Syntax of partially executed programs, where \code{PBB} is a partially executed block statement and \code{PP} is a partially executed program. Recall the definition of \code{AP} from Figure~\ref{fig:syntax-ann}.}
\label{fig:partial-prog-syntax}
\end{figure}

Recall that Theorem~\ref{theorem:inv-result} states that given an annotated execution that completes and reaches skip, the corresponding inverted program also completes and reaches skip. This is valid for all executions of complete programs. 

Consider the process of proving this by induction on the length of an execution. This typically requires a small number of initial execution steps to be explicitly stated, before using the induction hypothesis on the remaining execution. After the initial steps, the remaining program may no longer be complete:

\begin{definition}{(Partially executed program)} \label{def:partial-ex}
A partially executed program is a program produced as a result of partial execution of a complete program, respecting uniformity. Figure~\ref{fig:partial-prog-syntax} gives the syntax of partially executed programs. 
\end{definition}

A partially executed program, often referred to as a partial program, can contain conditions and expressions in evaluated form and does not require a one-to-one match between declaration and removal statements within a block. Figure~\ref{fig:partial-prog-syntax} defines a partially executed block \code{PBB}, which can be of eight different forms. The first form is the sequence of three types of declaration, then body and finally three types of removal statements. The subsequent six forms are obtained by removing progressively longer initial subsequences. Finally, the block body could be a partially executed program (\code{PP}) with all removal statements following.  Conditional and \code{runc} statements contain a second copy of their sub-programs, used within functions defined later to determine whether the execution of a statement has already started. These copies are introduced to aid our proof and are not necessary for reversal. In the following proofs, we only include second copies when required. For example, the function \emph{inv$^+$} introduced later must invert a conditional statement differently if its execution has already started, compared to when its execution has not. Partial programs respect uniformity meaning that the initial execution that produces it contains all of the available closing skip steps. Therefore executions only stop between segments, where each segment is an identifier step and all skip steps caused by it. Each partial program begins with an identifier step as a result, and is executed using the same (previously defined) operational semantics.

\begin{example} \label{ex:partial-prog}
The complete program (containing only a block) in Figure~\subref*{ex:partial-prog-original} has a matching removal statement for each declaration. A partial execution of this in Figure~\subref*{ex:partial-prog-partial} violates an assumption of complete programs as there is not a matching declaration statement for each removal. $\myqed$ 
\end{example}


\newsavebox\exPartialProgOrig
\newsavebox\exPartialProgPartial

\begin{lrbox}{\exPartialProgOrig}
\begin{minipage}{0.40\textwidth}
\centering
{\small \begin{lstlisting}[xleftmargin=4.0ex,mathescape=true,basicstyle=\scriptsize]
$\kw{begin b1.0} $
   $\kw{var} \code{ X = 20} \quad{\color{blue}\langle  \rangle} \code{;} $
   $\kw{var} \code{ Y = 40} \quad{\color{blue}\langle  \rangle} \code{;} $   
   $\code{Z = Z + X} \quad{\color{blue}\langle  \rangle} \code{;} $
   $\code{Z = Z + Y} \quad{\color{blue}\langle  \rangle} \code{;} $
   $\kw{remove} \code{ X = 20} \quad{\color{blue}\langle  \rangle} \code{;} $
   $\kw{remove} \code{ Y = 40} \quad{\color{blue}\langle  \rangle} \code{} $   
$\kw{end}$
\end{lstlisting} }
\end{minipage}
\end{lrbox}

\begin{lrbox}{\exPartialProgPartial}
\begin{minipage}{0.40\textwidth}
\centering
{\small \begin{lstlisting}[xleftmargin=4.0ex,mathescape=true,basicstyle=\scriptsize]
$\kw{begin b1.0} $
   $\code{Z = Z + Y} \quad{\color{blue}\langle  \rangle} \code{;} $
   $\kw{remove} \code{ X = 20} \quad{\color{blue}\langle  \rangle} \code{;} $
   $\kw{remove} \code{ Y = 40} \quad{\color{blue}\langle  \rangle} \code{} $   
$\kw{end}$
\end{lstlisting} }
\end{minipage}

\end{lrbox}

\begin{figure}[t]
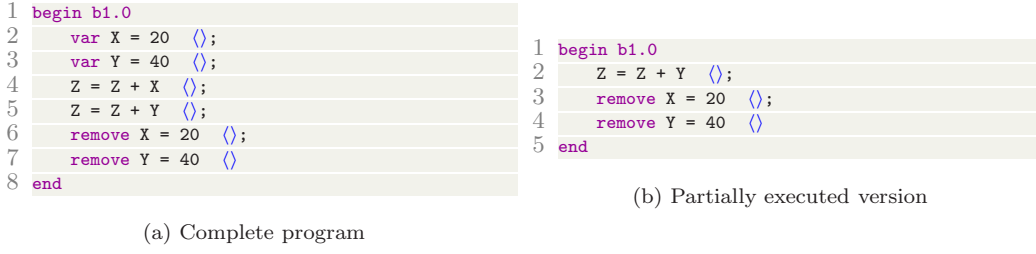

\subfloat[Complete program\label{ex:partial-prog-original}]{ \usebox\exPartialProgOrig }
\subfloat[Partially executed version\label{ex:partial-prog-partial}]{ \usebox\exPartialProgPartial }
\caption{An original program and a possible partially executed version}
\label{ex:partial-prog-and-original}
\end{figure}


The reversal of a partial program may not reach skip correctly. Consider the forward execution of a conditional statement beginning from the middle of the true branch. Reversal of only the remaining execution (the statements that finish this forward conditional) returns to a position in which the inverse conditional remains open and cannot close (the identifiers will likely not allow this). We note here that this is not sufficient if the branch contains skip statements (such as opening and closing of a block). Since these statements do not use identifiers, we cannot rely on identifiers to stop the execution correctly. In most cases this prevents the use of the induction hypothesis on a partial program. The program in Figure~\subref*{ex:partial-prog-partial} demonstrates a case that incorrectly reaches skip. After reversing the three assignments, the block could incorrectly close (as it is a skip step). Reversing a partial program requires the assumption of complete forward execution. As such, all reversal information will have been saved meaning the reverse execution may not stop at the desired position. To prevent such incorrect execution happening, we introduce an \emph{abort statement} to our syntax of partially executed programs (see Figure~\ref{fig:partial-prog-syntax}), written as \code{abort}. This has undefined behaviour and forcibly stops an execution whose only next step is an abort. Such \code{abort} statements cannot appear in original programs and are only used within the proof to follow.

Definition~\ref{def:skip-equiv} introduces \emph{skip equivalent} as the name given to the corresponding program code of the inverted execution of a partial program. 
\begin{definition}{(Skip equivalent)} \label{def:skip-equiv}
A skip equivalent of a program is either a single skip, a single abort or an identical version of the program where one or more sub-programs are either skip, abort or suitable skip equivalents. The function $\func{se}{$\code{P}$}$ in Figure~\ref{func-def:se} produces a skip equivalent of \code{P}.
\end{definition}

\begin{figure}[t]
{\footnotesize \begin{align*}
& \emph{se}(\code{$\varepsilon$})=\mbox{$\varepsilon$} \\
&\emph{se}(\code{AS;AP}) = \emph{se}(\code{AS$)$;AP} \\
&\emph{se}(\code{AP par AP}) = \emph{se}(\code{AP})~\code{par}~\emph{se}(\code{AP}) \\
&\emph{se}(\code{skip}) = \code{skip} \\
&\emph{se}(\code{abort}) = \code{abort} \\
&\emph{se}(\code{X = e (pa,A)}) = \code{skip (pa,A)} \\
&\emph{se}(\code{if In b b then AP AP else AQ AQ end (pa,A)}) = \code{skip (pa,A)}  \\
&\emph{se}(\code{if In T b then AP$'$ AP else AQ AQ end (pa,A)}) \\ \nonumber &\phantom{=} = \code{if In T b then $\emph{se}($AP$')$ AP else AQ AQ end (pa,A)}  \\
&\emph{se}(\code{while Wn b do AP end (pa,A)}) = \code{while Wn b do AP end (pa,A)} \text{ if started}  \\
&\emph{se}(\code{begin Bn AP AP end (pa,A)}) = \code{skip (pa,A)}   \\
&\emph{se}(\code{begin Bn AP AOP end (pa,A)}) = \code{begin Bn $\emph{se}($AP$)$ AOP end (pa,A)} \text{ if } \code{AP} \neq \code{AOP} \\
&\emph{se}(\code{runc Cn AP AP (pa,A)}) = \code{runc Cn $\emph{se}($AP$)$ AP (pa,A)}
\end{align*} }
\vspace{-.4cm}
\caption{A partial definition of the function \emph{se}, shown in full in Appendix~\ref{full-se-appen}. }
\label{func-def:se}
\end{figure}

\begin{example} \label{ex:skip-equiv}
Consider an annotated program \code{AP} = \code{if i1 T then AQ else AR end (pa,A)}, where \code{AQ} is a partially executed version of the true branch. The inverse version \code{IP} = \func{inv}{\code{AP}} should only invert \code{AQ}, leaving the conditional open. The skip equivalent that corresponds to this starting position is \code{if i1 T then skip I else} \func{inv}{\code{AR}} \code{end (pa,A)} (omitting second copies).  $\myqed$ 
\end{example}

With inverted partial programs not guaranteed to reach skip and instead a skip equivalent, we must ensure a reverse execution stops at the required position. For the partial forward execution to perform correctly, all previous steps of execution are assumed to have happened. With all reversal information of this previous execution available, the reverse execution may (unintentionally) continue past the desired skip equivalent.

\begin{example} \label{ex:stop-inv-ex-forcibly}
Recall the partial program in Figure~\subref*{ex:partial-prog-partial}. Inverse execution should invert only the remaining execution, namely the final assignment and both removals. Generating the inverse version based on this partial program produces a block with only those statements to invert. The block could incorrectly close (\sos{B2r}) leading to an incorrect program structure. $\myqed$
\end{example} 


\begin{figure}[t!]
{\footnotesize \begin{align} 
& \func{inv$^+$}{\code{$\varepsilon$},\square} = \varepsilon \\[1pt]
&\label{inv+:sc} \func{inv$^+$}{\code{S;P},\square} = \func{inv}{\code{P}}; \func{inv$^+$}{\code{S},\square} \\[1pt]
&\label{inv+:pc} \func{inv$^+$}{\code{P par Q},\square} = \func{inv$^+$}{\code{P},\square}~\code{par}~\func{inv$^+$}{\code{Q},\square} \\[1pt]
&\label{inv+:skip} \func{inv$^+$}{\code{skip I},\square} = \code{skip I; abort} \\[1pt]
&\label{inv+:cond-full} \nonumber \func{inv$^+$}{\code{if In ob ob then AP AP else AQ AQ end (pa,A)},\square} \\ &\phantom{=}= \code{if In ob ob then } \func{inv}{\code{AP}}~ \func{inv}{\code{AP}} \code{ else } \func{inv}{\code{AQ}}~\func{inv}{\code{AQ}} \code{ end (pa,A); abort} \\[1pt]
&\label{inv+:cond-true} \nonumber \func{inv$^+$}{\code{if In T ob then AP AP$'$ else AQ AQ end (pa,A)},\square} \\ &\phantom{=}= \code{if In ob ob then } \func{inv$^+$}{\code{AP}}~ \func{inv}{\code{AP$'$}} \code{ else } \func{inv}{\code{AQ}}~\func{inv}{\code{AQ}} \code{ end (pa,A)} \\ &\nonumber \text{ where } \code{AP} \neq \code{skip I} \text{ and } \code{AP $\neq$ AP$'$}\\[1pt]
&\label{inv+:while-diff-body} \func{inv$^+$}{\code{while Wn b do AP$'$ end (pa,A)},\square} = \code{while Wn ob do } \func{inv}{\code{AP}} \code{ end (pa,A)} \\ &\nonumber \text{ where } \code{AP} \text{ is such that }\beta(\code{Wn}) = \code{while Wn ob do AP end (pa,A)} \\[1pt]
&\func{inv$^+$}{\code{begin Bn AP AP$'$ end},\square} = \code{begin Bn }  \func{inv$^+$}{\code{AP}}~ \func{inv}{\code{AP}} \code{ end} \\ &\nonumber  \text{ where } \code{AP} \neq \code{AP$'$}\\[1pt]
&\func{inv$^+$}{\code{runc Cn AP AP$'$ (pa,A)},\square} = \code{call Cn n (pa,A)} 
\end{align} }%
\caption{A partial definition of the function \emph{inv$^+$}, shown in full in Appendix~\ref{full-inv+-appen}.}
\label{func-def:inv+}%
\end{figure}

We now define an extended inversion function \emph{inv$^+$}: $\setOf{AP} \rightarrow \setOf{IP}$ in Figure~\ref{func-def:inv+}, that takes an executed annotated program (either complete or partial) and returns the corresponding inverted version that stops at the desired skip equivalent. The function \emph{inv$^+$} calls the original inversion function \emph{inv} on all of the given program statements except the first statement. Should this statement be complete, the abort statement is sequentially composed. If the statement is partially executed, the abort statement is inserted at the appropriate position.

Manually inserting the abort statement via the inversion function \emph{inv$^+$} is sufficient in all cases except those of while loops or procedure calls. A partially executed while loop may have performed many iterations, with many left to execute (or reverse). We therefore define new semantics of loops and procedures that use abort statements to stop execution at a desired position. These rules are not shown here and instead are deferred to Appendix~\ref{appen-rules}.

\subsubsection{Proof of Theorem~\ref{theorem:inv-result}}
With the challenges associated with our inversion result addressed, we are now ready to prove Theorem~\ref{theorem:inv-result}. Recall Figure~\ref{diagram-theorem2-proof}, where proving each of the arrows $\Rightarrow_i$, for $1\leq i \leq 3$,  gives us Theorem~\ref{theorem:inv-result}. We first note that $\Rightarrow_1$ represents the restriction of a standard execution into an equivalent uniform execution, and recall the validation of this via Lemma~\ref{lem:equiv-uniform-fwds} in Section~\ref{subsec:interl}.

We focus on the arrow $\Rightarrow_2$. 
This represents a stronger version of Theorem~\ref{theorem:inv-result}, where all executions are uniform and the annotated program may be either complete or partially executed. The inverted version is now generated using the function \emph{inv$^+$} and its execution may reach either skip or a skip equivalent. Proposition~\ref{prop:stronger-theorem-4} states this stronger result.  

\begin{prop} \label{prop:stronger-theorem-4}
Let \textup{\texttt{P}} be an original program, \textup{\texttt{AP}} be the annotated program \textup{\texttt{$ann($P$)$}}, and \code{AP$'$} be the executed version of \code{AP}. Further let \textup{\texttt{$\square$}} be the tuple \{$\sigma$,$\gamma$,$\mu$,$\beta$,$\delta$\} of all initial program state environments. 

If $ (\code{P} \mid \square) \hookrightarrow^* (\code{skip} \mid \square') $ exists for some program state $\square'$, and therefore by Theorem~\ref{theorem:ann-result} the annotated execution $ (\code{AP} \mid \square_1) \uatran{\circ}{^*} (\code{skip I} \mid \square_1') $ for some \code{I} and program state $\square_1'$ such that $\square_1' \approx \square'$, then there exists a corresponding inverse execution $ (\func{inv$^+$}{\code{AP$'$}} \mid \square_2') \urtran{\circ}{^*} (\code{AQ} \mid \square_2) $, for the program \code{AQ} that is \code{skip I$'$} if \code{AP} was complete and a skip equivalent otherwise, and program states $\square_2'$, $\square_2$, such that $\square_2' \approx \square_1'$ and $\square_2 \approx \square_1$. 
\end{prop}

Our proof of this proposition uses two smaller results. The first, named the \emph{Statement Property} (Lemma~\ref{lem:sp}), considers all statement executions that begin with an identifier step. The second, named the \emph{Program Property } (Lemma~\ref{lem:pp}), is more general and considers all program executions that begin with either an identifier or a skip step.  

\begin{restatable}[Statement Property]{lemma}{lemsp} \label{lem:sp} 
Let \textup{\texttt{AS}} be a complete or partially executed annotated statement and \textup{\texttt{$\square$}} be the initial environments \{$\sigma$,$\gamma$,$\mu$,$\beta$,$\delta$\}. 

If a uniform execution $ (\code{AS} \mid \square) \uatran{m}{} (\code{AS$'$} \mid \square'') \uatran{}{_s^*}  (\code{AS$''$} \mid \square'') \uatran{\circ}{^*} (\code{skip I} \mid \square') $ exists for some statements \code{AS$'$}, \code{AS$''$} and program states $\square''$, $\square'$, then there exists a uniform inverse execution $ (\func{inv$^{+}$}{\code{AS}} \mid \square_1') \urtran{\circ}{^*} (\code{AT$''$} \mid \square_1'') \urtran{m}{} (\code{AT$'$} \mid \square_1) \urtran{}{_s^*} (\code{AT} \mid \square_1) $ for some statements \code{AT$''$}, \code{AT$'$}, \code{AT} where \code{AT} = \code{skip I} for some \code{I} if \code{AS} is a complete statement and \code{AT} is a skip equivalent if \code{AS} is partially executed, and states $\square_1'$, $\square_1''$, $\square_1$ such that $\square'_1 \approx \square'$, $\square''_1 \approx \square''$ and $\square_1 \approx \square$. 
\end{restatable} 

\begin{proof}
This proof is deferred to Appendix~\ref{appen-sp}.
\end{proof}

It is sufficient in Lemma~\ref{lem:sp} to consider only statement executions beginning with identifier steps. The Program Property considers executions beginning with skip steps (Part~1) and beginning with identifier steps (Part~2). Respecting uniformity means we only consider complete programs beginning with skip steps.

\begin{restatable}[Program Property]{lemma}{lempp} \label{lem:pp}
Let \textup{\texttt{AP}} be a complete program in Part~1, and either a complete or partially executed annotated program in Part~2. Further let \textup{\texttt{$\square$}} be the tuple \{$\sigma$,$\gamma$,$\mu$,$\beta$,$\delta$\} of initial program environments.

\vspace{-.2cm}
\begin{enumerate}[label=\textbf{Part \arabic*.},leftmargin=*]
\item If a uniform forward execution $ (\code{AP} \mid \square) \uatran{}{_s^*} (\code{AP$'$} \mid \square'') \uatran{\circ}{^*} (\code{skip I} \mid \square') $ exists for some program \code{AP$'$} and program states $\square''$, $\square'$ such that $\square''$ = $\square$, then there exists a uniform inverse execution $ (\func{inv$^{+}$}{\code{AP}} \mid \square_1') \urtran{\circ}{^*} (\code{AQ$'$} \mid \square_1'') \urtran{}{_s^*} (\code{AQ} \mid \square_1)  $ for some programs \code{AQ$'$}, \code{AQ} such that \code{AQ} is skip if \code{AP} is a complete program and a skip equivalent otherwise, and program states $\square_1'$, $\square_1''$, $\square_1$ such that $\square_1' \approx \square'$, $\square_1'' \approx \square''$ and $\square_1 \approx \square$. 
\item If a uniform forward execution $ (\code{AP} \mid \square) \uatran{m}{} (\code{AP$'$} \mid \square'') \uatran{}{_s^*} (\code{AP$''$} \mid \square'') \uatran{\circ}{^*} (\code{skip I} \mid \square') $ exists for some program \code{AP$'$} and program states $\square''$ and $\square'$, then there exists a uniform inverse execution $ (\func{inv$^{+}$}{\code{AP}} \mid \square_1') \urtran{\circ}{^*} (\code{AQ$''$} \mid \square_1'') \urtran{m}{} (\code{AQ$'$} \mid \square_1) \urtran{}{_s^*} (\code{AQ} \mid \square_1) $ for some programs \code{AQ$''$}, \code{AQ$'$}, \code{AQ} is skip if \code{AP} is a complete program and a skip equivalent otherwise, and states $\square_1'$, $\square_1''$, $\square_1$ such that $\square_1' \approx \square'$, $\square_1'' \approx \square''$ and $\square_1 \approx \square$.
\end{enumerate}
\end{restatable}

\begin{proof}
This proof is deferred to Appendix~\ref{appen-pp}.
\end{proof}

With our two smaller results proved, we are ready to establish Proposition~\ref{prop:stronger-theorem-4}.
\begin{proof}
This proof is by induction on the length of a uniform execution, namely $ (\code{AP} \mid \square)  \uatran{\circ}{^*} (\code{skip I} \mid \square') $, for some program \code{AP} and states $\square$ and $\square'$. It can be split into two parts. The first part considers executions that begin with skip steps, while the second considers executions that begin with identifier steps. Each part is now considered in turn.

\vspace{.3cm}
\noindent \textbf{Executions beginning with skip steps} This is done as in the proof of Part~1 of Lemma~\ref{lem:pp}. We consider only skip steps that can begin a complete program execution. The base cases (already shown in the proof of Part~1 of Lemma~\ref{lem:pp}) are Case~\ref{pp-part1-s2a} (sequential skips), Case~\ref{pp-part1-p3a} (empty parallel) and Case~\ref{pp-part1-b2a} (empty block). The inductive cases are those that correspond to Case~\ref{pp-part1-s1a} (sequential composition) and Case~\ref{pp-part1-p1a} (parallel composition) of the proof of Part~1 of Lemma~\ref{lem:pp}. With all shown to be valid in~\ref{appen-pp}, this part holds. 

\vspace{.3cm}
\noindent \textbf{Executions beginning with identifier steps} This is done as in the proof of Part~2 of Lemma~\ref{lem:pp}, where only those identifier steps that can be the beginning of a complete execution are considered. The base cases (already used in the proof of Part~2 of Lemma~\ref{lem:pp} and shown in the proof of Lemma~\ref{lem:sp}) are Case~\ref{sp-d1a} (assignment) and Case~\ref{sp-w1a} (loop with zero iterations). The inductive cases are those corresponding to Case~\ref{pp-part2-s1a} (sequential composition) and Case~\ref{pp-part2-p1a} (parallel composition) from the proof of Part~2 of Lemma~\ref{lem:pp}. Since the Program Property is proved via mutual induction with the Statement Property, we also consider several cases (of first steps) from the proof of Lemma~\ref{lem:sp}. This does not need any further base cases (since the two such cases are already used above), but does require the inductive cases corresponding to Case~\ref{sp-i1at} and Case~\ref{sp-i1af} (conditional statement), Case~\ref{sp-w3a} (while loop) and Case~\ref{sp-b1a} (block) from the proof of Lemma~\ref{lem:sp}. With all holding, this part is therefore valid. All possible executions have been shown to hold, meaning Proposition~\ref{prop:stronger-theorem-4} is also valid.
\end{proof}

With Proposition~\ref{prop:stronger-theorem-4} shown, the arrow $\Rightarrow_2$ from Figure~\ref{diagram-theorem2-proof} is therefore also correct. We finally consider the arrow $\Rightarrow_3$ from Figure~\ref{diagram-theorem2-proof}. As previously noted, all uniform executions are also standard executions by definition. This means $\Rightarrow_3$ exists in all cases as required. 
We have therefore completed the proof outline in Figure~\ref{diagram-theorem2-proof} meaning Theorem~\ref{theorem:inv-result} is valid.

%% file: Simulation.tex
Our method of reversibility has been implemented in the form of a simulator named Ripple~\cite{JHphd} (representing ``Reversing an Imperative Parallel Programming Language'') available at \url{https://github.com/jameshoey/ripple}. We plan to present Ripple more fully in an upcoming paper, which will include experimental data regarding the performance of Ripple, and efficiency of our approach. We note here that some aspects of our approach were chosen to aid the proof, including the use of stacks for reversal information and both the while and procedure environments. We believe performance can be improved via more efficient implementation.  

Ripple implements the operational semantics presented here and was created initially to confirm our proof of correctness. This tool takes an original program written in pseudocode, automatically converts this to our syntax (including construct identifiers and removal statements) and then simulates reversibility. Execution in both directions can be either full or step-by-step, with the program state and reversal information viewable at any intermediate position. Parallel composition can be executed with random or manual interleaving. 

Ripple was designed to demonstrate our application to debugging~\cite{JH2019}. This made use of reversibility to step backwards over a misbehaving execution, as well as specific debugging features such as a ``record mode'' that kept an interleaving and semantics history of any execution. 
An interesting extension is to compare the logs produced here (including the use of identifiers and reversal information) to the logs introduced by Lanese et al in \cite{IL2019}. The identifiers (and reversal information) used by Ripple captures an execution sufficiently to replay the execution to any intermediate position. This can be done, however, by executing statements in the same order as during the original forward execution, so currently does not support the causally equivalent replay as in \cite{IL2019}. Such replay ability, and exploration of a closer link to work above, will be developed as future work. 
%
%

%% file: Conclusion.tex
We have introduced a method of reversing the execution of programs written in an imperative concurrent programming language. This includes defining a process named annotation, capable of taking an irreversible program and producing a version that saves any information during forward execution. The process of inversion has also been defined, producing a version that uses this saved information to perform backtracking reversibility. We have overcome challenges introduced by parallel composition operator, such as keeping track of the many different interleaving orders, by assigning identifiers to statements in the order that they are executed. We have presented three sets of operational semantics, defining traditional, annotated and inverted execution respectively. 

We have shown our method of reversibility to be correct. This involved demonstrating that the annotated version performs the same forward execution as the original program, with extra reversal information saved. We have proved that the inverted version performs correct reversal and restores to a program state equivalent to that of before the forward execution. All information saved for reversal is used and the approach is therefore garbage-free.

The work presented here proposes a solution to a problem of correctly reversing an imperative concurrent program. To the best of our knowledge, no other approach to reversing imperative concurrent programs is proved to be correct. The previous sections describe several challenges we have overcome, including the definition of a state-saving approach to reversibility that has two main aims. Firstly, we have strived to store only minimal amount of data. Secondly, we have stored this data using uniform structures (namely stacks) which greatly assist us in the proofs of correctness. Clearly, the two aims are potentially orthogonal but we think that our solution offers a good balance. We note, however, that the format and the amount of information saved is not optimal. For example, for the while loops a single integer representing the number of iterations might be preferable to a stack of boolean values. 

Our concurrent programming language supports a shared memory form of concurrency using a parallel composition operator. We have propose a method of capturing an interleaving order of execution via identifiers. This has been inspired by CCSK (see Section~\ref{sec:related-work}), a method of reversing CCS using Communication Keys (identifiers). An interesting area of future work is to replace parallel composition with threads, adopt the basis of our approach to different forms of concurrency, such as those in the message-passing concurrent programming languages or in distributed computing (as recent advancements with Cauder have). One of our motivations for this work (and the reason we do not propose a reversible language) is to add reversibility to an irreversible programming language. In its current form, our approach is fairly specific to our programming language. However, we believe that the basis of this approach could be applied to other sequential or concurrent
programming languages. More generally, it might be worthwhile to consider adapting our approach to information saving and proof of correctness to developing reversible versions of general algorithms.

While the approach presented here is a solution to reversing the executions of concurrent programs, there are limitations to this. Firstly, several programming constructs are not currently supported. Examples of constructs missing that we aim to implement in future work include data structures such as arrays, stacks or heaps, and pointers, floating point arithmetic and functions. Secondly, there is the potential for a performance bottleneck when maintaining the order of identifiers assigned to executing statements. Currently the next identifier is retrieved from a central counter. In future work we plan to address this inefficiency possibly using the approach in \cite{AubertMedic2021}.
%
%
Thirdly, the approach presented here implements backtracking reversibility. As highlighted in our previous work on reversible debugging~\cite{JH2019,JH2020}, strictly following backtracking can be less efficient. We are currently exploring the modifications necessary to support causal-consistent reversibility~\cite{JHWiP}.
%
%
 Finally, the current implementation has not been optimised. Many choices that have been made to aid our correctness proof are not necessarily the most efficient solution. One example is the use of stacks for our reversal information. Further development will focus on a more efficient implementation.

%% file: Full-functions.tex
The following are full definitions of functions introduced in the main text, where some parts were omitted.

\subsection{Annotation} \label{full-ann-appen}
Figure~\ref{fig:func-ann-full} shows the complete definition of the annotation function $ann$.

\begin{figure}[hb!]
{\footnotesize \begin{align*}
&ann(\texttt{$\varepsilon$}) = \texttt{$\varepsilon$} \\
&ann(\texttt{S;P}) = ann(\texttt{S}); ann(\texttt{P}) \\ 
&ann(\texttt{P par Q}) = ann(\texttt{P}) \texttt{ par } ann(\texttt{Q}) \\
&ann(\texttt{skip}) = \texttt{skip I} \\
&ann(\texttt{X = e pa}) = \texttt{X = e (pa,A)} \\
&ann(\texttt{name[e] = e$'$ pa}) = \texttt{name[e] = e$'$ pa} \\
&ann(\texttt{X += e pa}) = \texttt{X += e (pa,A)} \\
&ann(\texttt{X -= e pa}) = \texttt{X -= e (pa,A)} \\
&ann(\texttt{name[e] += e$'$ pa}) = \texttt{name[e] += e$'$ pa} \\
&ann(\texttt{name[e] -= e$'$ pa}) = \texttt{name[e] -= e$'$ pa} \\
&ann(\texttt{if In b then P else Q end pa}) = \\ &\phantom{==} \texttt{if In b then $ann($P$)$ else $ann($Q$)$ end (pa,A)} \\
&ann(\texttt{while Wn b do P end pa}) = \texttt{while Wn b do $ann($P$)$ end (pa,A)} \\
&ann(\texttt{begin Bn P end}) = \texttt{begin Bn $ann($P$)$ end} \\
&ann(\texttt{var X = E pa}) = \texttt{var X = E (pa,A)} \\
&ann(\texttt{arr[n] name pa}) = \texttt{arr[n] name (pa,A)} \\
&ann(\texttt{proc Pn n is P end pa}) = \texttt{proc Pn n is $ann($P$)$ end (pa,A)} \\
&ann(\texttt{call Cn n pa}) = \texttt{call Cn n (pa,A)} \\
&ann(\texttt{runc Cn P end}) = \texttt{runc Cn AP end A} \\
&ann(\texttt{remove X = E pa}) = \texttt{remove X = E (pa,A)} \\
&ann(\texttt{remove arr[n] name pa}) = \texttt{remove arr[n] name (pa,A)} \\
&ann(\texttt{remove Pn n is P end pa}) = \texttt{remove Pn n is $ann($P$)$ end (pa,A)} 
\end{align*} }
\vspace{-.4cm}
\caption{Annotation function $ann$.}
\label{fig:func-ann-full}
\end{figure}

\subsection{Inversion} \label{full-inv-appen}
Figure~\ref{fig:func-inv-full} shows the complete definition of the inversion function $inv$.

\begin{figure}[!h]
{\footnotesize \begin{align*}
&inv(\texttt{$\varepsilon$}) = \texttt{$\varepsilon$} \\
&inv(\texttt{AS;AP}) = inv(\texttt{AP}); inv(\texttt{AS}) \\
&inv(\texttt{AP par AQ}) = inv(\texttt{AP}) \texttt{ par } inv(\texttt{AQ}) \\
&inv(\texttt{skip I}) = \texttt{skip I} \\
&inv(\texttt{X = e (pa,A)}) = \texttt{X = e (pa,A)} \\
&inv(\texttt{name[e] = e$'$ (pa,A)}) = \texttt{name[e] = e$'$ (pa,A)} \\
&inv(\texttt{X += e (pa,A)}) = \texttt{X -= e (pa,A)} \\
&inv(\texttt{X -= e (pa,A)}) = \texttt{X += e (pa,A)} \\
&inv(\texttt{name[e] += e$'$ (pa,A)}) = \texttt{name[e] -= e$'$ (pa,A)} \\
&inv(\texttt{name[e] -= e$'$ (pa,A)}) = \texttt{name[e] += e$'$ (pa,A)} \\
&inv(\texttt{if In b then AP else AQ end (pa,A)}) = \\&\phantom{==} \texttt{if In b then $inv($AP$)$ else $inv($AQ$)$ end (pa,A)} \\
&inv(\texttt{while Wn b do AP end (pa,A)}) = \texttt{while Wn b do $inv($AP$)$ end (pa,A)} \\
&inv(\texttt{begin Bn AP end}) = \texttt{begin Bn $inv($AP$)$ end} \\
&inv(\texttt{var X = e (pa,A)}) = \texttt{remove X = e (pa,A)} \\
&inv(\texttt{arr[n] name (pa,A)}) = \texttt{remove arr[n] name (pa,A)} \\
&inv(\texttt{proc Pn n is AP end (pa,A)}) = \texttt{remove Pn n is $inv($AP$)$ end (pa,A)} \\
&inv(\texttt{call Cn n (pa,A)}) = \texttt{call Cn n (pa,A)} \\
&inv(\texttt{runc Cn AP end A}) = \texttt{runc Cn $inv($AP$)$ A} \\
&inv(\texttt{remove Pn n is AP end (pa,A)}) = \texttt{proc Pn n is $inv($AP$)$ end (pa,A)} \\
&inv(\texttt{remove arr[n] name (pa,A)}) = \texttt{arr[n] name (pa,A)} \\
&inv(\texttt{remove X = e (pa,A)}) = \texttt{var X = e (pa,A)} 
\end{align*} }
\vspace{-.4cm}
\caption{Inversion function $inv$.}
\label{fig:func-inv-full}
\end{figure}

\subsection{Skip Equivalent} \label{full-se-appen}
Figure~\ref{func-def:se-full} shows the complete definition of the skip equivalent function $se$. In two rules below, we use ``if started'' to represent whether or not the execution of the loop body has begun or not.

\begin{figure}[h!]
{\footnotesize \begin{align*}
& \emph{se}(\code{$\varepsilon$})=\mbox{$\varepsilon$} \\
&\emph{se}(\code{AS;AP}) = \emph{se}(\code{AS$)$;AP} \\
&\emph{se}(\code{AP par AP}) = \emph{se}(\code{AP})~\code{par}~\emph{se}(\code{AP}) \\ 
&\emph{se}(\code{skip}) = \code{skip} \\
&\emph{se}(\code{abort}) = \code{abort} \\
&\emph{se}(\code{X = e (pa,A)}) = \code{skip (pa,A)} \\
&\emph{se}(\code{name[e] = e$'$ (pa,A)}) = \code{skip (pa,A)} \\
&\emph{se}(\code{X += e (pa,A)}) = \code{skip (pa,A)} \\
&\emph{se}(\code{X -= e (pa,A)}) = \code{skip (pa,A)} \\
&\emph{se}(\code{name[e] += e$'$ (pa,A)}) = \code{skip (pa,A)} \\
&\emph{se}(\code{name[e] -= e$'$ (pa,A)}) = \code{skip (pa,A)} \\
&\emph{se}(\code{if In b b then AP AP else AQ AQ end (pa,A)}) = \code{skip (pa,A)}  \\
&\emph{se}(\code{if In T b then AP$'$ AP else AQ AQ end (pa,A)}) \\ \nonumber &\phantom{=} = \code{if In T b then $\emph{se}($AP$')$ AP else AQ AQ end (pa,A)}  \\
&\emph{se}(\code{if In F b then AP AP else AQ$'$ AQ end (pa,A)}) \\ \nonumber &\phantom{=}= \code{if In F b then AP AP else $\emph{se}($AQ$')$ AQ end (pa,A)}  \\
&\emph{se}(\code{while Wn b do AP end (pa,A)}) = \code{skip} \text{ if not started}  \\
&\emph{se}(\code{while Wn b do AP end (pa,A)}) = \code{while Wn b do AP end (pa,A)} \text{ if started}  \\
&\emph{se}(\code{while Wn T do AP end (pa,A)}) = \code{while Wn T do $\emph{se}($AP$)$ end (pa,A)} \text{ if started \code{AP}} \\
&\emph{se}(\code{begin Bn AP AP end (pa,A)}) = \code{skip (pa,A)}   \\
&\emph{se}(\code{begin Bn AP AOP end (pa,A)}) = \code{begin Bn $\emph{se}($AP$)$ AOP end (pa,A)} \text{ if } \code{AP} \neq \code{AOP} \\
&\emph{se}(\code{var X = E (pa,A)}) = \code{skip (pa,A)} \\
&\emph{se}(\code{arr[n] name (pa,A)}) = \code{skip (pa,A)} \\
&\emph{se}(\code{proc Pn n is AP (pa,A)}) = \code{skip (pa,A)} \\
&\emph{se}(\code{remove X = E (pa,A)}) = \code{skip (pa,A)} \\
&\emph{se}(\code{remove arr[n] name pa}) = \code{skip (pa,A)} \\
&\emph{se}(\code{remove Pn n is AP (pa,A)}) = \code{skip (pa,A)} \\
&\emph{se}(\code{call Cn n (pa,A)}) = \code{skip (pa,A)} \\
&\emph{se}(\code{runc Cn AP AP (pa,A)}) = \code{runc Cn $\emph{se}($AP$)$ AP (pa,A)}
\end{align*} }
\vspace{-.4cm}
\caption{Definition of the function \emph{se}.}
\label{func-def:se-full}
\end{figure}

\allowdisplaybreaks

\subsection{Extended Inversion} \label{full-inv+-appen}
Figure~\ref{func-def:inv+-full} contains the complete definition of the extended inversion function $inv^+$.

\begin{figure}[h!]
\vspace{-.4cm}{\scriptsize \begin{align*} 
& \func{inv$^+$}{\code{$\varepsilon$},\square} = \varepsilon \\[1pt]
&\func{inv$^+$}{\code{S;P},\square} = \func{inv}{\code{P}}; \func{inv$^+$}{\code{S},\square} \\[1pt]
&\func{inv$^+$}{\code{P par Q},\square} = \func{inv$^+$}{\code{P},\square}~\code{par}~\func{inv$^+$}{\code{Q},\square} \\[1pt]
&\func{inv$^+$}{\code{skip I},\square} = \code{skip I; abort} \\[1pt]
&\func{inv$^+$}{\code{X = e (pa,A)},\square} = \code{X = e (pa,A); abort} \\
&\func{inv$^+$}{\code{name[e] = e$'$ (pa,A)},\square} = \code{name[e] = e$'$ (pa,A); abort} \\
&\func{inv$^+$}{\code{X += e (pa,A)},\square} = \code{X -= e (pa,A); abort} \\
&\func{inv$^+$}{\code{X -= e (pa,A)},\square} = \code{X += e (pa,A); abort} \\
&\func{inv$^+$}{\code{name[e] += e$'$ (pa,A)},\square} = \code{name[e] -= e$'$ (pa,A); abort} \\
&\func{inv$^+$}{\code{name[e] -= e$'$ (pa,A)},\square} = \code{name[e] += e$'$ (pa,A); abort} \\
&\func{inv$^+$}{\code{if In ob ob then AP AP else AQ AQ end (pa,A)},\square} \\ &\phantom{=}= \code{if In ob ob then } \func{inv}{\code{AP}}~ \func{inv}{\code{AP}} \code{ else } \func{inv}{\code{AQ}}~\func{inv}{\code{AQ}} \code{ end (pa,A); abort} \\[1pt]
&\func{inv$^+$}{\code{if In T ob then AP AP$'$ else AQ AQ end (pa,A)},\square} \\ &\phantom{=}= \code{if In ob ob then } \func{inv$^+$}{\code{AP}}~ \func{inv}{\code{AP$'$}} \code{ else } \func{inv}{\code{AQ}}~\func{inv}{\code{AQ}} \code{ end (pa,A)} \text{ where } \code{AP} \neq \code{skip I} \text{ and } \code{AP $\neq$ AP$'$}\\[1pt]
&\func{inv$^+$}{\code{if In F ob then AP AP else AQ AQ$'$ end (pa,A)},\square} \\ &\phantom{=}= \code{if In ob ob then } \func{inv}{\code{AP}}~ \func{inv}{\code{AP}} \code{ else } \func{inv$^+$}{\code{AQ}}~\func{inv}{\code{AQ$'$}} \code{ end (pa,A)} \text{ where } \code{AQ} \neq \code{skip I} \text{ and } \code{AQ $\neq$ AQ$'$} \\[1pt]
&\func{inv$^+$}{\code{if In T ob then skip I AP else AQ AQ end (pa,A)},\square} \\  &\phantom{=}= \code{if In ob ob then abort } \func{inv}{\code{AP}} \code{ else } \func{inv}{\code{AQ}}~\func{inv}{\code{AQ}} \code{ end (pa,A)} \\[1pt]
&\func{inv$^+$}{\code{if In F ob then AP AP else skip I AQ end (pa,A)},\square} \\ &\phantom{=}= \code{if In ob ob then } \func{inv}{\code{AP}}~\func{inv}{\code{AP}} \code{ else abort } \func{inv}{\code{AQ}} \code{ end (pa,A)} \\[1pt]
&\func{inv$^+$}{\code{while Wn b do AP end (pa,A)},\square} = \code{while Wn b do } \func{inv}{\code{AP}} \code{ end (pa,A); abort} \text{ where } \beta(\code{Wn}) = \code{und} \\[1pt]
&\func{inv$^+$}{\code{while Wn b do AP$'$ end (pa,A)},\square} = \code{while Wn ob do } \func{inv}{\code{AP}} \code{ end (pa,A)} \\ &\nonumber \text{ where } \code{AP} \text{ is such that }\beta(\code{Wn}) = \code{while Wn ob do AP end (pa,A)} \\[1pt]
&\func{inv$^+$}{\code{begin Bn AP AP end},\square} = \code{begin Bn } \func{inv}{\code{AP}}~ \func{inv}{\code{AP}} \code{ end; abort} \\[1pt]
&\func{inv$^+$}{\code{begin Bn AP AP$'$ end},\square} = \code{begin Bn }  \func{inv$^+$}{\code{AP}}~ \func{inv}{\code{AP}} \code{ end} \text{ where } \code{AP} \neq \code{AP$'$}\\[1pt]
&\func{inv$^+$}{\code{var X = E (pa,A)},\square} = \code{remove X = E (pa,A); abort} \\[1pt]
&\func{inv$^+$}{\code{arr[n] name (pa,A)},\square} = \code{remove arr[n] name (pa,A); abort} \\[1pt]
&\func{inv$^+$}{\code{proc Pn n is AP (pa,A)},\square} = \code{remove Pn n is}~\func{inv}{\code{AP}}~\code{(pa,A); abort} \\[1pt]
&\func{inv$^+$}{\code{remove X = E (pa,A)},\square} = \code{var X = E (pa,A); abort} \\[1pt]
&\func{inv$^+$}{\code{remove arr[n] name (pa,A)},\square} = \code{arr[n] name (pa,A); abort} \\[1pt]
&\func{inv$^+$}{\code{remove Pn n is AP (pa,A)},\square} = \code{proc Pn n is}~\func{inv}{\code{AP}}~\code{(pa,A); abort} \\[1pt]
&\func{inv$^+$}{\code{call Cn n (pa,A)},\square} = \code{call Cn n (pa,A); abort} \\[1pt]
&\func{inv$^+$}{\code{runc Cn AP AP$'$ (pa,A)},\square} = \code{call Cn n (pa,A)} 
\end{align*} }
\vspace{-.4cm}
\caption{Extended inversion function \emph{inv$^+$}.}
\label{func-def:inv+-full}%
\end{figure}%

%% file: Rules-Extended.tex
In this section, we introduce additional operational semantics used during our proof of the inversion result (Theorem~\ref{theorem:inv-result}). The following rules are used when considering partially executed while loops or procedure calls. 
For each appropriate rule \sos{r}, the modified versions are named \sos{rP} and \sos{rPI}.

\noindent \textbf{While Loop}
Consider a partial execution beginning with the next evaluation of the condition (not the first). This can be the last check (rule \sos{W2a}) or any other except the first (rule \sos{W4a}). The modified version of each, namely \sos{W2aP} and \sos{W4aP}, saves a true abort flag (to indicate reverse execution should stop here) within the element of the boolean sequence needed for inversion (the extended stack \code{W} that now stores triples such that $\delta[\code{(m,T,T) $\rightharpoonup$ W}$). Note all following iterations use the original rule and save a false abort flag (or nothing) (each modified rule requires the loop body to be \code{abort}, which is disallowed in normal execution and only within the proof). The matching inverse rules, namely \sos{W3rP} and \sos{W4rP}, can only be used in cases where a true abort flag is present (meaning one of the rules above must have been executed) and return a loop containing abort as its body. 

{\footnotesize
\begin{align*}
&\sos{W2aP} \quad \frac{\begin{aligned} & \code{m} = \func{next}{} \quad \code{$\beta($Wn$)$} = \emph{def} \quad (\code{b pa} \mid \beta,\square) \hookrightarrow^*_{\code{b}} (\code{F} \mid \beta,\square) \quad \code{AP} = \code{abort} \\[-3pt] & \code{C} = \func{getAI}{\beta(\code{Wn})} \quad \delta' = \delta[\code{(m,T,T) $\rightharpoonup$ W}, \code{(m,C) $\rightharpoonup$ WI}] \quad \beta' = \beta[\code{Wn}] \end{aligned}}{(\code{while Wn b do AP end (pa,A)} \mid \delta,\beta,\square) \atran{m} (\code{skip m:A} \mid \delta',\beta',\square)} \\[8pt]
&\sos{W4aP} \quad \frac{\begin{aligned} & \code{m} = \func{next}{} \quad \code{$\beta($Wn$)$} = \code{while Wn b do AP$'$ end (pa,A)} \quad (\code{b pa} \mid \beta,\square) \hookrightarrow^*_{\code{b}} (\code{T} \mid \beta,\square) \quad \code{AP} = \code{abort} \\[-3pt] & \code{AR} = \code{while Wn b do } \func{reL}{\code{AP$'$}} \code{ end (pa,m:A)} \quad \delta' = \delta[\code{(m,T,T) $\rightharpoonup$ W}] \quad \beta' = \beta[\code{Wn $\Rightarrow$ AR}] \end{aligned}}{(\code{while Wn b do AP end (pa,A)} \mid \delta,\beta,\square) \atran{m} (\code{while Wn T do } \func{reL}{\code{AP$'$}} \code{ end (pa,m:A)} \mid \delta',\beta',\square)} \\[8pt]
&\sos{W3rP} \quad \frac{\begin{aligned}  &\code{m} = \func{previous}{}  \quad \code{A} = \code{m:A$'$} \quad \code{$\beta($Wn$)$} = \emph{und} \quad  \code{$\delta($W$)$} = \code{(m,T,T):W$'$} \quad \code{$\delta($WI$)$} = \code{(m,C):WI$'$} \\[-3pt] & \code{AR} = \code{while Wn b do } \func{IreL}{\func{setAI}{\code{IP},\code{C}}} \code{ end (pa,A$'$)} \quad \delta' = \delta[\code{W/W$'$}, \code{WI/WI$'$}] \quad \beta' = \beta[\code{Wn $\Rightarrow$ AR}] \end{aligned}}{(\code{while Wn b do IP end (pa,A)} \mid \delta,\beta,\square) \rtran{m} (\code{while Wn T do abort end (pa,A$'$)} \mid \delta',\beta',\square)} \\[8pt]
&\sos{W4rP} \quad \frac{\begin{aligned} &\code{m} = \func{previous}{} \quad \code{A} = \code{m:A$'$} \quad \code{$\beta($Wn$)$} = \emph{def} \quad  \code{$\delta($W$)$} = \code{(m,T,T):W$'$} \\[-3pt] & \code{AR} = \code{while Wn b do } \func{reL}{\code{IP}} \code{ end (pa,A$'$)} \quad \delta' = \delta[\code{W/W$'$}] \quad \beta' = \beta[\code{Wn $\Rightarrow$ AR}] \end{aligned}}{(\code{while Wn b do IP end (pa,A)} \mid \delta,\beta,\square) \rtran{m} (\code{while Wn T do abort end (pa,A$'$)} \mid \delta',\beta',\square)} 
\end{align*} }%


Consider a partial execution beginning part way through a loop iteration (the rules \sos{W2aPI} and \sos{W4aPI}). Likewise to above, the use of an abort flag can indicate the inverse iteration that requires partial inversion. The difference here is that the partial loop body performed forwards must also be reversed. In order to determine when to stop within the loop body, this partial version of the loop body \code{PP} is saved alongside the abort flag (note the further extension to stack \code{W} on $\delta$ such that $\delta[\code{(m,T,T,PP) $\rightharpoonup$ W}$). The while loop is extended to contain a copy of the partial program \code{PP} that remains unchanged during execution. Such statements do not appear during normal execution (due to abort) and only within our proof. The matching inverse rules, namely \sos{W3rPI} and \sos{W4rPI}, are only available provided a true abort flag and a partial program have been saved.  

{\footnotesize \begin{align*}
&\sos{W2aPI} \quad \frac{\begin{aligned} & \code{m} = \func{next}{} \quad \code{$\beta($Wn$)$} = \emph{def} \quad (\code{b pa} \mid \beta,\square) \hookrightarrow^*_{\code{b}} (\code{F} \mid \beta,\square) \quad \code{AP} = \code{abort} \quad \code{PP} \neq \code{skip} \\[-3pt] & \code{C} = \func{getAI}{\beta(\code{Wn})} \quad \delta' = \delta[\code{(m,T,T,PP) $\rightharpoonup$ W}, \code{(m,C) $\rightharpoonup$ WI}] \quad \beta' = \beta[\code{Wn}] \end{aligned}}{(\code{while Wn b do AP PP end (pa,A)} \mid \delta,\beta,\square) \atran{m} (\code{skip m:A} \mid \delta',\beta',\square)} \\[8pt]
&\sos{W4aPI} \quad \frac{\begin{aligned} & \code{m} = \func{next}{} \quad \code{$\beta($Wn$)$} = \code{while Wn b do AP$'$ AP$''$ end (pa,A)} \quad (\code{b pa} \mid \beta,\square) \hookrightarrow^*_{\code{b}} (\code{T} \mid \beta,\square)  \\[-3pt] & \code{AP} = \code{abort} \quad \code{PP} \neq \code{skip} \quad \code{AR} = \code{while Wn b do } \func{reL}{\code{AP$'$}} \code{ skip end (pa,m:A)} \\[-3pt] & \delta' = \delta[\code{(m,T,T,PP) $\rightharpoonup$ W}] \quad \beta' = \beta[\code{Wn $\Rightarrow$ AR}] \end{aligned}}
{(\code{S} \mid \delta,\beta,\square) \atran{m} (\code{while Wn T do AP$_1$ skip end (pa,m:A)} \mid \delta',\beta',\square)} \\ &\phantom{\sos{W4aPI} \quad } \text{where } \code{S} = \code{while Wn b do AP PP end (pa,A)}  \\[8pt]
&\sos{W3rPI} \quad \frac{\begin{aligned}  &\code{m} = \func{previous}{} \quad \code{A} = \code{m:A$'$} \quad \code{$\beta($Wn$)$} = \emph{und} \quad  \code{$\delta($W$)$} = \code{(m,T,T,PP):W$'$} \quad \code{$\delta($WI$)$} = \code{(m,C):WI$'$} \\[-3pt] & \code{PP} \neq \code{skip} \quad \code{AR} = \code{while Wn b do } \func{IreL}{\func{setAI}{\code{IP},\code{C}}} \code{ skip end (pa,A$'$)} \\[-3pt] & \delta' = \delta[\code{W/W$'$}, \code{WI/WI$'$}] \quad \beta' = \beta[\code{Wn $\Rightarrow$ AR}] \quad \code{PP$_1$} = \func{inv$^+$}{\code{PP}} \end{aligned}}
{(\code{S} \mid \delta,\beta,\square) \rtran{m} (\code{while Wn T do }\func{IreL}{\code{PP$_1$}} \code{ skip end (pa,A$'$)} \mid \delta',\beta',\square)} \\ &\phantom{\sos{W4rPI} \quad } \text{where } \code{S} = \code{while Wn b do IP PP$'$ end (pa,A)}\\[8pt]
&\sos{W4rPI} \quad \frac{\begin{aligned} &\code{m} = \func{previous}{} \quad \code{A} = \code{m:A$'$} \quad \code{$\beta($Wn$)$} = \emph{def} \quad  \code{$\delta($W$)$} = \code{(m,T,T,PP):W$'$} \quad \code{PP} \neq \code{skip} \\[-3pt] &  \code{AR} = \code{while Wn b do } \func{IreL}{\code{IP}} \code{ skip end (pa,A$'$)}\quad \delta' = \delta[\code{W/W$'$}] \\[-3pt] & \beta' = \beta[\code{Wn $\Rightarrow$ AR}] \quad \code{PP$_1$} = \func{inv$^+$}{\code{PP}} \end{aligned}}
{(\code{S} \mid \delta,\beta,\square) \rtran{m} (\code{while Wn T do }\func{IreL}{\code{PP$_1$}} \code{ skip end (pa,A$'$)} \mid \delta',\beta',\square)} \\ &\phantom{\sos{W4rPI} \quad } \text{where } \code{S} = \code{while Wn b do IP PP$'$ end (pa,A)} \end{align*} }%

\noindent \textbf{Procedure Call}
Consider a procedure call beginning at the end of the execution of the procedure body. The inverse execution should open the inverted call and then stop. Using rule \sos{G3aP}, the second program within the runc construct is manually set to abort (a situation that can never arise during normal execution and only within the proof). Alongside the saving of all identifiers prior to removal, a true abort flag is recorded. As for loops, consider a partial program beginning within the procedure body. The rule \sos{G3aPI} handles this situation, where the second program of the runc construct is the partial program. This partial program is saved alongside the abort flag. The inverse versions, namely \sos{G1rP} and \sos{G1rPI}, are only available provided a true abort flag exists.

{\footnotesize \begin{align*}
&\sos{G3aP} \quad \frac{\begin{aligned} & \texttt{m} = \func{next}{} \quad \texttt{$\mu$(Cn)} = \texttt{OP} \quad \texttt{AOP} = \texttt{abort} \\[-3pt] & \delta' = \delta[\texttt{(m,$getAI(\mu($Cn$)$,T$)$) $\rightharpoonup$ Pr}] \end{aligned}}{(\texttt{runc Cn skip I AOP end A} \mid \delta,\mu,\square) \atran{m} (\texttt{skip m:A} \mid \delta',\mu[\texttt{Cn}],\square)} \\[8pt]
&\sos{G3aPI} \quad \frac{\begin{aligned} &\code{m} = \func{next}{} \quad \code{$\mu$(Cn)} = \emph{def} \quad \code{AOP} \neq \code{abort} \quad \code{PP} \neq \code{skip} \\[-3pt] & \delta' = \delta[\code{(m,$getAI(\mu($Cn$)$,T,PP$)$) $\rightharpoonup$ Pr}] \quad \mu' = \mu[\code{Cn}] \end{aligned}}{(\code{runc Cn skip I AOP PP end A} \mid \delta,\mu,\square) \atran{m} (\code{skip m:A} \mid \delta',\mu',\square)} \\[8pt]
\end{align*} }%


\vspace{-.5cm}
{\footnotesize \begin{align*}
&\sos{G1rP} \quad \frac{\begin{aligned} & \texttt{m} = \func{previous}{} \quad \texttt{A} = \texttt{m:A$'$} \quad  \texttt{$\mu(evalP($n,pa$)$)} = \texttt{(n,IP)} \quad \texttt{$\delta($Pr$)$} = \texttt{(m,C,T):Pr$'$} \\[-3pt] & \delta' = \delta[\texttt{Pr/Pr$'$}] \quad \mu' = \mu[\texttt{Cn} \Rightarrow \texttt{(n,IP$'$)}] \quad \texttt{IP$'$} =  \func{IreP}{\func{setAI}{\code{IP},\code{C}},\code{Cn}} \end{aligned}}{(\texttt{call Cn n (pa,A)} \mid \delta,\mu,\square) \rtran{m} (\texttt{runc Cn abort IP$'$ end A$'$} \mid \delta', \mu',\square)} \\[8pt]
&\sos{G1rPI} \quad \frac{\begin{aligned} & \code{m} = \func{previous}{} \quad \code{A} = \code{m:A$'$} \quad  \texttt{$\mu(evalP($n,pa$)$)} = \texttt{(n,IP)} \quad \code{$\delta($Pr$)$} = \code{(m,C,T,PP):Pr$'$} \\[-3pt] & \delta' = \delta[\code{Pr/Pr$'$}] \quad \mu' = \mu[\code{Cn} \Rightarrow \code{(n,IP$'$)}] \quad \code{IP$'$} = \func{IreP}{\func{setAI}{\code{IP},\code{C}},\code{Cn}} \\[-3pt]& \code{PP$_1$} = \func{IreP}{\func{setAI}{\func{inv$^+$}{\code{PP}},\code{C}},\code{Cn}} \end{aligned}}
{(\code{call Cn n (pa,A)} \mid \delta,\mu,\square) \rtran{m} (\code{runc Cn PP$_1$ IP$'$ skip end A$'$} \mid \delta', \mu',\square)}
\end{align*} }%

%% file: sp-proof.tex
We now give a proof of Lemma~\ref{lem:sp} (Statement Property). By abuse of notation, we use \code{AS} to represent either a complete or partially executed statement. In cases considering full statement execution (e.g. Case~\ref{sp-d1a}), we ignore the abort statement introduced by \func{inv$^+$}{} as the program stops as desired.

\lemsp*
\setcounter{lemma}{9}

All executions from here are uniform and so we omit the specific notation. A uniform execution $(\code{P} \mid \square) \uatran{m}{} (\code{P$'$} \mid \square')$ is now written as $(\code{P} \mid \square) \atran{m} (\code{P$'$} \mid \square')$.

\begin{proof}
This proof is by mutual induction of the Statement Property (this lemma) and the Program Property (Lemma~\ref{lem:pp}), on the length of the executions $ (\code{AS} \mid \square) \atran{\circ}^* (\code{skip I} \mid \square') $ and $ (\code{AP} \mid \square) \atran{\circ}^* (\code{skip I} \mid \square') $ respectively.

With no executions of length 0, our base cases are any executions of length 1. The two base cases are a single assignment statement (Case~\ref{sp-d1a}, with equivalent cases for each type of assignment omitted) and a single while loop statement that performs zero iterations (Case~\ref{sp-w1a}). 

\begin{case}{(Assignment \sos{D1a})} \label{sp-d1a}
Consider a single assignment of the value of the expression \code{e} to the variable \code{X}. Let \code{AS} = \code{X~=~e~(pa,A)} and $\square$ be the tuple of initial program state environments $\allstores$. Assume an execution exists with the first transition via \sos{D1a} such that $$ (\code{AS} \mid \square) = (\code{X~=~v~(pa,A)} \mid \square) \atran{m} (\code{skip (pa,m:A)} \mid \square') $$ for identifier \code{m} and program state $\square'$ such that $\square'  = \square[\sigma[\code{l $\mapsto$ v}], \delta[\code{(m,v1) $\rightarrow$ X}]]$, where \code{l} is the memory location for the variable \code{X}, \code{v} is the result of evaluating the expression \code{e} (the new value) and \code{v1} is the current value of \code{X} (the old value). Application of the rule \sos{D1a} means the premises were valid at that point, namely that \code{m} = \func{next}{}, $(\code{e pa} \mid \square) \hookrightarrow^*_{\code{a}} (\code{v} \mid \square) $ (the expression evaluates to \code{v}) and $ \code{$evalV($X,pa,$\gamma)$} = \code{l} $ (the variable is bound to memory location \code{l}).

We need to show that there exists an execution $$(\func{inv$^+$}{\code{AS}} \mid \square_1') \rtran{\circ}^* (\code{AT$''$} \mid \square_1'') \rtran{m} (\code{AT$'$} \mid \square_1) \rtran{}_s^* (\code{AT} \mid \square_1) $$ for some statements \code{AT$''$}, \code{AT$'$}, \code{AT}, and program states $\square_1'$, $\square_1''$, $\square_1$ such that $\square_1' \approx \square'$ and $\square_1 \approx \square$. Firstly, we have that \func{inv$^+$}{\code{AS}} = \code{X = e (pa,m:A)}. Beginning in the program state $\square_1'$, such that $\square_1' \approx \square'$, we note that no further execution was performed forwards and so does not require inversion. Therefore \code{AT$''$} = \code{X = e (pa,m:A)}, and $\square_1'$ = $\square_1''$, giving the execution $(\func{inv$^+$}{\code{AS}} \mid \square_1') \rtran{\circ}^* (\code{AT$''$} \mid \square_1'') $ with length 0.

From the program state $\square_1'$, we have that \code{m} = \func{previous}{} and that the identifier stack for this assignment is equal to \code{m:A} (\code{m} as its head). Also, $\delta_1'$(\code{X}) = \code{(m,v1):X$'$} (the stack \code{X} contains the pair \code{(m,v1)} as its head) and $ \code{$evalV($X,pa,$\gamma)$} = \code{l} $ (the variable is bound to the memory location \code{l} or equivalent). With all premises of the rule \sos{D1r} valid, and that no other rule is applicable, we have the execution $$ (\func{inv$^+$}{\code{AS}} \mid \square_1') = (\code{X~=~e~(pa,m:A)} \mid \square_1') \rtran{m} (\code{skip (pa,A)} \mid \square_1) $$ for state $\square_1$ such that $\square_1 = \square_1'[\sigma_1'[\code{l $\mapsto$ v1}]$ and $\delta_1'[\code{X/X$'$}]]$, meaning $\square_1 \approx \square$ (as required). This execution has reached skip and so no further execution is available. Therefore we take \code{AT$'$} = \code{skip (pa,A)} and \code{AT} = \code{AT$'$}. With program states matching, and \code{AT} being skip, we have shown our desired inverse execution and therefore this case to hold.
\end{case}

\begin{case}{\textbf{(Loop \sos{W1a})}} \label{sp-w1a}
Consider a while loop with zero iterations (loop condition immediately evaluates to false). This cannot be partially executed as it completes in a single step, therefore no copies are included. Let \code{AS} = \code{while Wn b do AP end (pa,A)} and $\square$ be the tuple of initial program state environments $\allstores$. Assume an execution exists with the first transition via \sos{W1a} $$ (\code{AS} \mid \square) = (\code{while Wn b do AP end (pa,A)} \mid \square) \atran{m} (\code{skip (pa,m:A)} \mid \square') $$ for program state $\square'$ such that $\square'$ = $\square[\delta[\code{(m,F,F) $\rightarrow$ W}]]$  (noting the false abort flag that is of no consequence here). Application of the rule \sos{W1a} means all premises were valid at this point, namely that \code{m} = \func{next}{}, $\beta$(\code{Wn}) = \emph{und} (there is no mapping for this loop within $\beta$ since it has not started) and $(\code{b pa} \mid \square) \hookrightarrow^*_{\code{b}} (\code{F} \mid \square) $ (the condition evaluates to \code{F}).

We need to show that there exists an execution $$(\func{inv$^+$}{\code{AS}} \mid \square_1') \rtran{\circ}^* (\code{AT$''$} \mid \square_1'') \rtran{m} (\code{AT$'$} \mid \square_1) \rtran{}_s^* (\code{AT} \mid \square_1) $$ for some statements \code{AT$''$}, \code{AT$'$}, \code{AT}, and program states $\square_1'$, $\square_1''$, $\square_1$ such that $\square_1' \approx \square'$ and $\square_1 \approx \square$. We have that \func{inv$^+$}{\code{AS}} =  \code{while Wn b do} \func{inv}{\code{AP}} \code{end (pa,m:A)} (note can ignore the abort in this case). Starting in the program state $\square_1'$ such that $\square_1' \approx \square'$, the forward execution above does not contain any further execution and so no use of the induction hypothesis is required. We therefore take \code{AT$''$} =  \code{while Wn b do} \func{inv}{\code{AP}} \code{end (pa,m:A)}, and $\square_1'$ = $\square_1''$, giving the execution $(\func{inv$^+$}{\code{AS}} \mid \square_1') \rtran{\circ}^* (\code{AT$''$} \mid \square_1'') $ with length 0.

From the program state $\square_1'$, we have that \code{m} = \func{previous}{} (since no further identifiers have been used) and that the identifier stack for this loop is equal to \code{m:A} (\code{m} as its head). We also have that $\beta_1'$(\code{Wn}) = \emph{und} (since the while loop has not started) and $\delta_1'$(\code{W}) = \code{(m,F,F):W$'$} (the stack \code{W} has the triple \code{(m,F,F)} as its head). This shows that all premises of \sos{W1r} are valid, and that no other rule is applicable. This gives the execution $$ (\code{while Wn b do}~\func{inv}{\code{AP}}~\code{end (pa,m:A)} \mid \square_1'') \rtran{m} (\code{skip (pa,A)} \mid \square_1) $$ for program state $\square_1$ such that $\square_1$ = $\square_1''[\delta_1''[\code{W/W$'$}]]$ (with $\delta_1''(\code{W}) = \code{(m,F,F):W$'$}$), resulting in $\square_1 = \square$ as required. This execution has reached skip and so no further execution is available. Therefore we take \code{AT$'$} = \code{skip (pa,A)} and \code{AT} = \code{AT$'$}. With \code{AT} having reached skip and program states that match at each required position, we have our desired inverse execution meaning this case holds.
\end{case}

With Lemma~\ref{lem:sp} valid for all base cases, it is valid for all executions of length 1. We now consider inductive cases. Assume Lemma~\ref{lem:sp} holds for all statements \code{AR} and program states $\square^*$ such that the execution is of length $k$ (where $k \geq 1$), namely $ (\code{AR} \mid \square) \atran{m} (\code{AR$'$} \mid \square'') \atran{}_s^*  (\code{AR$''$} \mid \square'') \atran{\circ}^* (\code{skip I} \mid \square')$ (induction hypothesis). Now assume that the execution $ (\code{AS} \mid \square) \atran{m} (\code{AS$'$} \mid \square'') \atran{}_s^*  (\code{AS$''$} \mid \square'') \atran{\circ}^* (\code{skip I} \mid \square')$ has length $l$ such that $l > k$.

Inductive cases considered are for conditional statements (Case~\ref{sp-i1at}--\ref{sp-i5a}), loops (Cases~\ref{sp-w2a}--\ref{sp-w5a}), blocks (Case~\ref{sp-b1a}--\ref{sp-h2a}) and call statements (Cases~\ref{sp-g1a}--\ref{sp-g3a}). 

\vspace{.3cm}
\noindent \textbf{Conditional Statements}

\begin{case}{\textbf{(Conditional \sos{I1aT})}} \label{sp-i1at} 
Consider the opening of a conditional statement. Let \code{AS} = \code{if In b b then AP AP else AQ AQ end (pa,A)} and the initial program state $\square$ = $\allstores$. Assume the following execution exists where the first transition is via \sos{I1aT} $$ \begin{aligned} &(\code{AS} \mid \square) = (\code{if In b b then AP AP else AQ AQ end (pa,A)} \mid \square) \\ &\phantom{(\code{AS} \mid \square)} \atran{m} (\code{if In T b then AP AP else AQ AQ end (pa,m:A)} \mid \square'') \\ &\phantom{(\code{AS} \mid \square)} \atran{}_s^* (\code{if In T b then AP$'$ AP else AQ AQ end (pa,m:A)} \mid \square'') \\ &\phantom{(\code{AS} \mid \square)} \atran{\circ}^* (\code{skip (pa,A$'$)} \mid \square')  \end{aligned} $$ for some program \code{AP$'$} and program states $\square''$, $\square'$ such that $\square''$ = $\square$ (as \sos{I1aT} does not change the program state), and an identifier stack \code{A$'$} that ends with the sub-stack \code{(m:A)}. Application of \sos{I1aT} means that all premises were valid at the time, namely \code{m} = \func{next}{} and $ (\code{b pa} \mid \square) \hookrightarrow_b^* (\code{T} \mid \square) $. Since this transition opens a conditional, \code{AS} must be a complete statement and therefore \code{AP} must also be a complete program. As such, it is possible for \code{AP} to begin with skip steps meaning the execution $$ \begin{aligned} &(\code{if In T b then AP AP else AQ AQ end (pa,m:A)} \mid \square'') \\ &\phantom{} \atran{}_s^* (\code{if In T b then AP$'$ AP else AQ AQ end (pa,m:A)} \mid \square'') \end{aligned} $$ exists for some \code{AP$'$}. This group of skip steps are said to have been caused by the identifier rule using \code{m}. The inverted version of this group will be contained within the uniform inverse execution of the remaining forward program (corresponding inverse skip steps will be caused by the previous inverse identifier rule). From here, we use \code{AS$'$} to represent \code{if In T b then AP$'$ AP else AQ AQ end (pa,m:A)}.   

We need to show that there exists an execution $$(\func{inv$^+$}{\code{AS}} \mid \square_1') \rtran{\circ}^* (\code{AT$''$} \mid \square_1'') \rtran{m} (\code{AT$'$} \mid \square_1) \rtran{}_s^* (\code{AT} \mid \square_1) $$ for some statements \code{AT$''$}, \code{AT$'$}, \code{AT}, and program states $\square_1'$, $\square_1''$, $\square_1$ such that $\square_1' \approx \square'$, $\square_1'' \approx \square''$ and $\square_1 \approx \square$. From Figure~\ref{func-def:inv+}, we have that \func{inv$^+$}{\code{AS}} = \code{if In b b then} \func{inv}{\code{AP}}~\func{inv}{\code{AP}} \code{else} \func{inv}{\code{AQ}}~\func{inv}{\code{AQ}} \code{end (pa,A$'$)} (where the sequentially composed abort statement can be ignored here as there is no further execution and will stop at the required position).

From our assumed execution above, we note that the execution $$ \begin{aligned} &(\code{AS$'$} \mid \square'') = (\code{if In T b then AP$'$ AP else AQ AQ end (pa,m:A)} \mid \square'') \\ & \phantom{(\code{AS$'$} \mid \square)} \atran{\circ}^* (\code{skip (pa,A$'$)} \mid \square') \end{aligned} $$ is a uniform execution and must begin with an identifier step. Since \code{AS$'$} must be a partially executed program, from Figure~\ref{func-def:inv+}, we have that \func{inv$^+$}{\code{AS$'$}} = \code{if In b b then} \func{inv$^+$}{\code{AP$'$}} \func{inv}{\code{AP}} \code{else} \func{inv}{\code{AQ}} \func{inv}{\code{AQ}} \code{end (pa,m:A)}, where an abort is added to the end of the inverted true branch in order to stop the execution at the required position (after entire execution of the true branch but before the closure of the conditional). The induction hypothesis of the Statement Property (Lemma~\ref{lem:sp}) applied to this shorter execution gives us $$ \begin{aligned} & (\func{inv$^+$}{\code{AS$'$}} \mid \square_1') = (\code{if In b b then}~\func{inv$^+$}{\code{AP}}~\func{inv}{\code{AP}} \\&\phantom{(\func{inv$^+$}{\code{AS$'$}} \mid \square_1')~=~ } \code{else } \func{inv}{\code{AQ}}~\func{inv}{\code{AQ}} \code{ end (pa,m:A)} \mid \square_1') \\&\phantom{(\func{inv$^+$}{\code{AS$'$}} \mid \square_1') } \rtran{\circ}^* (\code{AR$'$} \mid \square_1'') \end{aligned} $$ by a sequence of rule applications SR, for a skip equivalent \code{AR$'$} and states $\square_1'$, $\square_1''$ such that $\square_1' \approx \square'$ and $\square_1'' \approx \square''$. Since \code{AP} must be a full program, \func{inv$^+$}{\code{AP}} sequentially composes an abort statement to the end. As a result, the induction hypothesis use above performs the entire inverted true branch until the abort is reached. Therefore \code{AR$'$} = \code{if In T b then abort} \func{inv}{\code{AP}} \code{else} \func{inv}{\code{AQ}}~\func{inv}{\code{AQ}} \code{end (pa,m:A)}. 

Now compare the statements \func{inv$^+$}{\code{AS}} and \func{inv$^+$}{\code{AS$'$}}. They are the same with the exception that \func{inv$^+$}{\code{AS$'$}} contains an abort statement at the end of the true branch. Since both will begin with the same execution (inversion of the rest of the forward program after the identifier plus skip steps), then by the same sequence of rules SR we obtain $$ \begin{aligned} & (\func{inv$^+$}{\code{AS}} \mid \square_1') = (\code{if In b b then } \func{inv$^+$}{\code{AP}}~\func{inv}{\code{AP}} \\ & \phantom{(\func{inv$^+$}{\code{AS}} \mid \square_1') = } \code{ else } \func{inv}{\code{AQ}}~\func{inv}{\code{AQ}} \code{ end (pa,A$'$)} \mid \square_1') \\ & \phantom{(\func{inv$^+$}{\code{AS}} \mid \square_1')} \rtran{\circ}^* (\code{AR} \mid \square_1'') \end{aligned} $$ for program \code{AR}. The same sequence of rules mean that this transition sequence will perform the execution of the entire true branch, reaching the point at which the execution of \func{inv$^+$}{\code{AS$'$}} hit the abort statement (minus the trivial \sos{S1a} rule that has no effect on the program state). Therefore \code{AR} = \code{if In T b then skip I$'$} \func{inv}{\code{AP}} \code{else} \func{inv}{\code{AQ}}~\func{inv}{\code{AQ}} \code{end (pa,m:A)} for some \code{I$'$}. With the program states matching to our desired execution, we can take \code{AT$''$} = \code{AR}. 

From program state $\square_1''$ (recall $\square_1'' \approx \square''$), we have that \code{m} = \func{previous}{} (since no further identifiers have been used) and the identifier stack for this statement is equal to \code{m:A}. This shows all premises of the rule \sos{I4r} to be valid, and that no other rule is applicable. Using \sos{I4r}, we have $$ \begin{aligned} & (\code{AT$''$} \mid \square_1'') = (\code{if In T b then skip I$'$ } \func{inv}{\code{AP}} \code{ else } \func{inv}{\code{AQ}}~\func{inv}{\code{AQ}} \\ & \phantom{(\code{AT$''$} \mid \square_1'') = } \code{ end (pa,m:A)} \mid \square_1'')    \rtran{m} (\code{skip (pa,A)} \mid \square_1) \end{aligned} $$ for program state $\square_1$ such that $\square_1$ = $\square_1''$ (since this rule does not change the program state), meaning $\square_1 \approx \square$ as required. With the program states matching, we can take \code{AT$'$} = \code{skip (pa,A)}. With no skip steps to apply (already reached skip), we can take \code{AT} = \code{AT$'$}. Since \code{AT} is skip, we can conclude this case holds.
\end{case}

\begin{case}{\textbf{(Conditional \sos{I1aF})}} \label{sp-i1af} 
Consider the opening of a conditional statement where the condition evaluates to false. This follows Case~\ref{sp-i1at} using rules \sos{I1aF} and \sos{I5a} in place of \sos{I1aT} and \sos{I4a} respectively.
\end{case}

\begin{case}{\textbf{(Conditional \sos{I2a})}} \label{sp-i2a} 
Consider an identifier step from within the true branch of a conditional statement. Let \code{AS} = \code{if In T b then AP AOP else AQ AQ end (pa,A)} and the initial program state $\square$ = $\allstores$. Assume the following execution exists with the first transition via the rule \sos{I2a} $$ \begin{aligned} &(\code{AS} \mid \square) = (\code{if In T b then AP AOP else AQ AQ end (pa,A)} \mid \square) \\&\phantom{(\code{AS} \mid \square)} \atran{m} (\code{if In T b then AP$''$ AOP else AQ AQ end (pa,A)} \mid \square'') \\&\phantom{(\code{AS} \mid \square)} \atran{}_s^* (\code{if In T b then AP$'$ AOP else AQ AQ end (pa,A)} \mid \square'') \\&\phantom{(\code{AS} \mid \square)} \atran{\circ}^* (\code{skip (pa,A$'$)} \mid \square') \end{aligned} $$ for programs \code{AP$''$}, \code{AP$'$}, and program states $\square''$ and $\square'$. We note that this assumed execution can be rewritten to highlight the point at which the true branch finishes, specifically  $$ \begin{aligned} &(\code{AS} \mid \square) = (\code{if In T b then AP AOP else AQ AQ end (pa,A)} \mid \square) \\&\phantom{(\code{AS} \mid \square)} \atran{m} (\code{if In T b then AP$''$ AOP else AQ AQ end (pa,A)} \mid \square'') \\&\phantom{(\code{AS} \mid \square)} \atran{}_s^* (\code{if In T b then AP$'$ AOP else AQ AQ end (pa,A)} \mid \square'') \\&\phantom{(\code{AS} \mid \square)} \atran{\circ}^* (\code{if In T b then skip I$_1$ AOP else AQ AQ end (pa,A)} \mid \square''') \\&\phantom{(\code{AS} \mid \square)} \atran{n} (\code{skip (pa,n:A)} \mid \square') \end{aligned} $$ for some program state $\square'''$. From here, let \code{AS$'$} = \code{if In T b then skip I$_1$ AOP else AQ AQ end (pa,A)}.

We need to show that there exists an execution $$(\func{inv$^+$}{\code{AS}} \mid \square_1') \rtran{\circ}^* (\code{AT$''$} \mid \square_1'') \rtran{m} (\code{AT$'$} \mid \square_1) \rtran{}_s^* (\code{AT} \mid \square_1) $$ for some statements \code{AT$''$}, \code{AT$'$}, \code{AT}, and program states $\square_1'$, $\square_1''$, $\square_1$ such that $\square_1' \approx \square'$, $\square_1'' \approx \square''$ and $\square_1 \approx \square$. From Figure~\ref{func-def:inv+}, we have that \func{inv$^+$}{\code{AS}} = \code{if In b b then} \func{inv$^+$}{\code{AP}} \func{inv}{\code{AOP}} \code{else} \func{inv}{\code{AQ}} \func{inv}{\code{AQ}} \code{end (pa,n:A)}, with abort inserted to stop the inverse execution at the end of the inverted true branch. 

From our rewritten assumed execution, $$ \begin{aligned} &(\code{AS$'$} \mid \square''') = (\code{if In T b then skip I$_1$ AOP else AQ AQ end (pa,A)} \mid \square''')  \\ &\phantom{(\code{AS$'$} \mid \square''')} \atran{n} (\code{skip (pa,n:A)} \mid \square') \end{aligned} $$ must be the closing of the conditional, via the rule \sos{I4a}. There can only be one matching inverted execution step, namely the opening of the inverted conditional statement via the rule \sos{I1rT} (the very first step of the inverse execution). Therefore we have the execution \begin{equation} \label{eq-sp-case-i2a-part1} \begin{aligned} &(\func{inv$^+$}{\code{AS}} \mid \square_1') = (\code{if In b b then } \func{inv$^+$}{\code{AP}}~\func{inv}{\code{AOP}} \\&\phantom{(\func{inv$^+$}{\code{AS}} \mid \square_1') ===} \code{ else } \func{inv}{\code{AQ}} ~\func{inv}{\code{AQ}} \code{ end (pa,n:A)} \mid \square_1') \\&\phantom{(\func{inv$^+$}{\code{AS}} \mid \square_1')} \rtran{n} (\code{if In T b then } \func{inv$^+$}{\code{AP}}~\func{inv}{\code{AOP}} \\& \phantom{(\func{inv$^+$}{\code{AS}} \mid \square_1') ===} \code{ else } \func{inv}{\code{AQ}} ~\func{inv}{\code{AQ}} \code{ end (pa,A)} \mid \square_1''') \\&\phantom{(\func{inv$^+$}{\code{AS}} \mid \square_1')} \rtran{}_s^* (\code{if In T b then IP$'$} ~\func{inv}{\code{AOP}} \code{ else } \func{inv}{\code{AQ}} ~\func{inv}{\code{AQ}} \\& \phantom{(\func{inv$^+$}{\code{AS}} \mid \square_1') ===}\code{ end (pa,A)} \mid \square_1''')  \end{aligned} \end{equation} for some program \code{IP$'$} and program states $\square_1'$, $\square_1'''$ such that $\square_1' \approx \square'$ and $\square_1''' \approx \square'''$.

Returning to our assumed execution, repeated use of \sos{I2a} (from conclusion to premises) allows us to obtain the shorter execution (as the conditional must close) $$ \begin{aligned} & (\code{AP} \mid \square)  \atran{m} (\code{AP$''$} \mid \square'') \atran{}_s^* (\code{AP$'$} \mid \square'') \atran{\circ}^* (\code{skip I$_1$} \mid \square''')\end{aligned}  $$ With this guaranteed to be a uniform execution and to begin with an identifier step, application of the induction hypothesis of Part~2 of the Program Property (Lemma~\ref{lem:pp}) gives $$ \begin{aligned} & (\func{inv$^+$}{\code{AP}}  \mid \square_1''')  \rtran{\circ}^* (\code{AR$''$} \mid \square_1'') \rtran{m} (\code{AR$'$} \mid \square_1) \rtran{}_s^* (\code{AR} \mid \square_1)\end{aligned} $$ for some programs \code{AR$''$}, \code{AR$'$}, \code{AR} such that \code{AR} is a skip equivalent, and program states $\square_1'''$, $\square_1''$, $\square_1$ such that $\square_1''' \approx \square'''$, $\square_1'' \approx \square''$ and $\square_1 \approx \square$. Through repeated use of the corresponding inverse rule \sos{I2r} (premises to conclusion), we have the execution 
\begin{equation} \label{eq-sp-case-i2a-part2}
\begin{split} &(\code{if In T b then } \func{inv$^+$}{\code{AP}} ~\func{inv}{\code{AOP}} \code{ else } \func{inv}{\code{AQ}} ~\func{inv}{\code{AQ}} \code{ end (pa,A)} \mid \square_1''') \\&  \rtran{\circ}^* (\code{if In T b then } \code{AR$''$} ~\func{inv}{\code{AOP}} \code{ else } \func{inv}{\code{AQ}} ~\func{inv}{\code{AQ}} \code{ end (pa,A)} \mid \square_1'') \\& \rtran{m} (\code{if In T b then } \code{AR$'$} ~\func{inv}{\code{AOP}} \code{ else } \func{inv}{\code{AQ}} ~\func{inv}{\code{AQ}} \code{ end (pa,A)} \mid \square_1) \\& \rtran{}_s^* (\code{if In T b then } \code{AR} ~\func{inv}{\code{AOP}} \code{ else } \func{inv}{\code{AQ}} ~\func{inv}{\code{AQ}} \code{ end (pa,A)} \mid \square_1) \end{split}
\end{equation}
Since \code{AR} is either skip or a skip equivalent produced via the induction hypothesis, by the definition of skip equivalents, we can conclude that \code{if In T b then} \code{AR} \func{inv}{\code{AOP}} \code{else} \func{inv}{\code{AQ}} \func{inv}{\code{AQ}} \code{end (pa,A)} is also a skip equivalent. Note that if there are skip steps to apply within \ref{eq-sp-case-i2a-part1}, then those same steps in the same order will be present within the first transition of \ref{eq-sp-case-i2a-part2}. Through the composition of the executions \ref{eq-sp-case-i2a-part1} and \ref{eq-sp-case-i2a-part2} in that order, we get $$ \begin{aligned} &(\code{if In b b then } \func{inv$^+$}{\code{AP}}~\func{inv}{\code{AOP}} \code{ else } \func{inv}{\code{AQ}} ~\func{inv}{\code{AQ}} \\&\phantom{=} \code{ end (pa,n:A)} \mid \square_1') \\&\phantom{} \rtran{n} (\code{if In T b then } \func{inv$^+$}{\code{AP}}~\func{inv}{\code{AOP}}  \code{ else } \func{inv}{\code{AQ}} ~\func{inv}{\code{AQ}} \\& \phantom{=} \code{ end (pa,A)} \mid \square_1''') \\&  \rtran{\circ}^* (\code{if In T b then } \code{AR$''$} ~\func{inv}{\code{AOP}} \code{ else } \func{inv}{\code{AQ}} ~\func{inv}{\code{AQ}} \code{ end (pa,A)} \mid \square_1'') \\& \rtran{m} (\code{if In T b then } \code{AR$'$} ~\func{inv}{\code{AOP}} \code{ else } \func{inv}{\code{AQ}} ~\func{inv}{\code{AQ}} \code{ end (pa,A)} \mid \square_1) \\& \rtran{}_s^* (\code{if In T b then } \code{AR} ~\func{inv}{\code{AOP}} \code{ else } \func{inv}{\code{AQ}} ~\func{inv}{\code{AQ}} \code{ end (pa,A)} \mid \square_1)  \end{aligned} $$ which matches our desired execution at each point. Therefore we take \code{AT$'$} = \code{if In T b then} \code{AR$'$} ~\func{inv}{\code{AOP}} \code{else} \func{inv}{\code{AQ}} \func{inv}{\code{AQ}} \code{end (pa,A)} and \code{AT} = \code{if In T b then} \code{AR} \func{inv}{\code{AOP}} \code{else} \func{inv}{\code{AQ}} \func{inv}{\code{AQ}} \code{end (pa,A)}, meaning this case holds.
\end{case}

\begin{case}{\textbf{(Conditional \sos{I3a})}} \label{sp-i3a} 
Consider an identifier step that comes from the false branch of a conditional statement. This case matches Case~\ref{sp-i2a} but uses \sos{I3a} in place of \sos{I2a}.
\end{case}

\begin{case}{\textbf{(Conditional \sos{I4a})}} \label{sp-i4a} 
Consider the closing of a conditional statement (the final step). Let \code{AS} = \code{if In T b then skip I AP else AQ AQ end (pa,A)} and the initial program state $\square$ = $\allstores$. Assume the following execution exists with the first transition via \sos{I4a} $$ \begin{aligned} & (\code{AS} \mid \square) = (\code{if In T b then skip I AP else AQ AQ end (pa,A)} \mid \square) \\ &\phantom{(\code{AS} \mid \square)} \atran{m} (\code{skip (pa,m:A)} \mid \square') \end{aligned} $$ for program state $\square'$ such that $\square'$ = $\square$[$\delta$[\code{(m,T) $\rightharpoonup$ B}]].  Using \sos{I4a} means all premises were valid at the time, namely \code{m} = \func{next}{}. Note that no skip steps or further execution is possible.

We need to show that there exists an execution $$(\func{inv$^+$}{\code{AS}} \mid \square_1') \rtran{\circ}^* (\code{AT$''$} \mid \square_1'') \rtran{m} (\code{AT$'$} \mid \square_1) \rtran{}_s^* (\code{AT} \mid \square_1) $$ for some statements \code{AT$''$}, \code{AT$'$}, \code{AT}, and program states $\square_1'$, $\square_1''$, $\square_1$ such that $\square_1' \approx \square'$, $\square_1'' \approx \square''$ and $\square_1 \approx \square$. From Figure~\ref{func-def:inv+}, we have that \func{inv$^+$}{\code{AS}} = \code{if In b b then abort} \func{inv}{\code{AP}} \code{else} \func{inv}{\code{AQ}} \func{inv}{\code{AQ}} \code{end (pa,m:A)}, with abort inserted to stop the inverse execution at the beginning of the execution of the inverse true branch. 

From \func{inv$^+$}{\code{AS}}, it is clear that we are immediately in a state such that the corresponding inverse rule \sos{I1rT} can be applied. We therefore take \code{AT$''$} = \func{inv$^+$}{\code{AS}} and $\square_1''$ = $\square_1'$, giving the first part of our desired execution with length 0.

From the program state $\square_1''$ (recall $\square_1''$ = $\square_1'$), we have that \code{m} = \func{previous}{} (since no other identifiers were used), the identifier stack for this statement is equal to \code{m:A} and $\square_1''$($\delta_1''$(\code{B})) = \code{(m,T):B$'$} where \code{B$'$} is the remaining stack. With all premises of the rule \sos{I1rT} shown to be valid, and that no other rule can be applied due to these premises, application of the rule \sos{I1rT} (the corresponding inverse identifier rule as expected) gives us $$ \begin{aligned} &(\code{AT$''$} \mid \square_1'') = (\code{if In b b then abort } \func{inv}{\code{AP}} \code{ else } \func{inv}{\code{AQ}}~\func{inv}{\code{AQ}} \\&\phantom{(\code{AT$''$} \mid \square_1'') = } \code{ end (pa,m:A)} \mid \square_1'') \\&\phantom{(\code{AT$''$} \mid \square_1'')} \rtran{m} (\code{if In T b then abort } \func{inv}{\code{AP}} \code{ else } \func{inv}{\code{AQ}}~\func{inv}{\code{AQ}} \\&\phantom{(\code{AT$''$} \mid \square_1'') = } \code{ end (pa,A)} \mid \square_1) \\&\phantom{(\code{AT$''$} \mid \square_1'')} \rtran{}_s^* (\code{if In T b then AP$'$ } \func{inv}{\code{AP}} \code{ else } \func{inv}{\code{AQ}}~\func{inv}{\code{AQ}} \\&\phantom{(\code{AT$''$} \mid \square_1'') = } \code{ end (pa,A)} \mid \square_1) \end{aligned} $$ for some program \code{AP$'$} and program state $\square_1$ such that $\square_1$ = $\square_1''$[$\delta_1''$[\code{B/B$'$}]], meaning $\square_1 \approx \square$ as required. The abort statement forcibly stops the execution at that point, meaning no skip steps are available in the final transition. Therefore \code{AP$'$} = \code{abort}. We can take \code{AT$'$} = \code{if In T b then abort} \func{inv}{\code{AP}} \code{else} \func{inv}{\code{AQ}}\func{inv}{\code{AQ}} \code{end (pa,A)}, and \code{AT} = \code{AT$'$}. As such, this shows our desired execution and therefore this case to hold.
\end{case}

\begin{case}{\textbf{(Conditional \sos{I5a})}} \label{sp-i5a} 
Consider the closing of a conditional statement that previously executed the false branch. This follows Case~\ref{sp-i4a} and uses the rules \sos{I5a} and \sos{I1rF} instead of \sos{I4a} and \sos{I1rT} respectively.
\end{case}

\vspace{.3cm}
\noindent \textbf{While Loops}

\begin{case}{\textbf{(Loop \sos{W2aP})}} \label{sp-w2a}
Consider the closing of a while loop. This final step of the forward execution will be the first step of the reverse execution. To stop the execution after the matching inverse step, an abort statement is manually inserted to ensure the rule \sos{W2aP} is used. Because of this, the second copies are not required. In normal execution outside of the proof, this abort will not appear and the rule \sos{W2a} is used. Therefore let \code{AS} = \code{while Wn b do abort end (pa,A)} and the initial program state $\square$ = $\allstores$. Assume the following execution exists with the first transition via \sos{W2aP} $$ (\code{AS} \mid \square) = (\code{while Wn b do abort end (pa,A)} \mid \square) \atran{m} (\code{skip (pa,m:A)} \mid \square') $$ for program state $\square'$ such that $\square'$ = $\square[\delta[\code{(m,T,T) $\rightarrow$ W, (m,C) $\rightarrow$ WI}], \beta[\texttt{Wn}]]$. Note that no skip steps are available, and that a true abort flag is also pushed to stack \code{W}. All premises of \sos{W2aP} must have been valid, namely \code{m} = \func{next}{}, $\beta$(\code{Wn}) = \code{while Wn b do AP end (pa,A)} (a mapping exists as the loop must have started), $ (\code{b pa} \mid \square) \hookrightarrow_b^* (\code{F} \mid \square) $ (the condition evaluates to false) and the body of the loop is \code{abort}.

We need to show that there exists an execution $$(\func{inv$^+$}{\code{AS}} \mid \square_1') \rtran{\circ}^* (\code{AT$''$} \mid \square_1'') \rtran{m} (\code{AT$'$} \mid \square_1) \rtran{}_s^* (\code{AT} \mid \square_1) $$ for some statements \code{AT$''$}, \code{AT$'$}, \code{AT}, and program states $\square_1'$, $\square_1''$, $\square_1$ such that $\square_1' \approx \square'$, $\square_1'' \approx \square''$ and $\square_1 \approx \square$. From Figure~\ref{func-def:inv+}, we have that \func{inv$^+$}{\code{AS}} = \code{while Wn b do} \func{inv}{\code{AP}} \code{end (pa,A)}.

The assumed forward execution is the closing of a while loop (the final step), meaning the corresponding inverse step will be the opening of the inverse loop (the first step). Therefore we can take \code{AT$''$} = \func{inv$^+$}{\code{AS}} and $\square_1''$ = $\square_1'$, giving the first part of our desired execution with length 0. 

From program state $\square_1''$, we have that \code{m} = \func{previous}{} (since no other identifiers were saved), the identifier stack for this statement is equal to \code{m:A}, $\square_1''$($\delta_1''$(\code{W}) = \code{(m,T,T):W$'$}, $\delta_1''$(\code{WI}) = \code{(m,C):WI$'$})), $\beta_1''$(\code{Wn}) = \emph{und})). This means that all premises of the corresponding inverse rule \sos{W3rP} are valid, and that no other rule is applicable. Using the rule \sos{W3rP} (the corresponding inverse identifier rule as expected) we get $$ \begin{aligned} &(\code{AT$''$} \mid \square_1'') = (\code{while Wn b do } \func{inv}{\code{AP}} \code{ end (pa,m:A)} \mid \square_1'') \\&\phantom{(\code{AT$''$} \mid \square_1'')} \rtran{m} (\code{while Wn T do abort end (pa,A)} \mid \square_1) \end{aligned} $$ for program state $\square_1$ such that $\square_1$ = $\square_1''$[$\delta_1''$[\code{W/W$'$},\code{WI/WI$'$}],$\beta_1''$[\code{Wn} $\mapsto$ \code{while Wn b do }\func{inv}{\code{AP}}~\code{end (pa,A)}]], meaning $\square_1 \approx \square$ as required. The abort is added by this rule (and crucially not by the normal rule \sos{W3r}) in order to stop the execution at this point. With the program state matching with our desired execution, we can take \code{AT$'$} to be \code{while Wn T do abort end (pa,A)}, which is a valid skip equivalent. Since no skip steps are available (due to abort), we can also take \code{AT} = \code{AT$'$} and complete the desired execution showing the case to hold. 
\end{case}

\begin{case}{\textbf{(Loop \sos{W3a})}} \label{sp-w3a}
Consider the opening of a while loop that performs at least one iteration (the condition must evaluate to true as Case~\ref{sp-w1a} considers loops with zero iterations). Since this cannot be a partial execution, no second copies are required. Let \code{AS} = \code{while Wn b do AP end (pa,A)} and the initial program state $\square$ = $\allstores$. Assume the following execution exists with the first transition via \sos{W3a} $$ \begin{aligned} & (\code{AS} \mid \square) = (\code{while Wn b do AP end (pa,A)} \mid \square) \\&\phantom{(\code{AS} \mid \square)} \atran{m} (\code{while Wn T do AP$'$ end (pa,m:A)} \mid \square'') \\&\phantom{(\code{AS} \mid \square)} \atran{}_s^* (\code{while Wn T do AP$''$ end (pa,m:A)} \mid \square'')  \atran{\circ}^* (\code{skip (pa,A$'$)} \mid \square') \end{aligned} $$ for some program \code{AP$'$} such that \code{AP$'$} = \func{reL}{\code{AP}}, and program states $\square''$, $\square'$ such that $\square''$ = $\square$[$\delta$[\code{(m,F,F) $\Rightarrow$ W]}, $\beta$[\code{Wn $\rightarrow$ while Wn b do AP$'$ end (pa,m:A)}]]. For \sos{W3a} to be applied as above, all premises must have been valid prior to it. This means that \code{m} = \func{next}{}, $\beta$(\code{Wn}) = \emph{und} (the while loop has not yet started), $ (\code{b pa} \mid \square) \hookrightarrow_b^* (\code{T} \mid \square) $ (the condition evaluates to true) and the body of the loop is not equal to \code{abort}. Note \code{A$'$} is a version of the persistent identifier stack for this loop statement ending with the sub-stack \code{(m:A)}. From this point on, we use \code{AS$'$} to denote \code{while Wn T do AP$''$ end (pa,m:A)}. 

We need to show that there exists an execution $$(\func{inv$^+$}{\code{AS}} \mid \square_1') \rtran{\circ}^* (\code{AT$''$} \mid \square_1'') \rtran{m} (\code{AT$'$} \mid \square_1) \rtran{}_s^* (\code{AT} \mid \square_1) $$ for some statements \code{AT$''$}, \code{AT$'$}, \code{AT}, and program states $\square_1'$, $\square_1''$, $\square_1$ such that $\square_1' \approx \square'$, $\square_1'' \approx \square''$ and $\square_1 \approx \square$. From Figure~\ref{func-def:inv+}, we have that \func{inv$^+$}{\code{AS}} = \code{while Wn b do} \func{inv}{\code{AP}} \code{end (pa,A)} (note the use of \func{inv}{} as \code{AP} is guaranteed to be a full program).   

Firstly, the part of the assumed execution $$ \begin{aligned} &(\code{while Wn T do AP$'$ end (pa,m:A)} \mid \square'') \\ &\atran{}_s^* (\code{while Wn T do AP$''$ end (pa,m:A)} \mid \square'') \end{aligned}$$ only exists if the loop body \code{AP$'$} begins with skip steps. If this is the case, the inverted version of this group of skip steps will be performed in the inverse execution after the previous inverse identifier rule. Therefore this will be provided via the induction hypothesis use shown later. 

Next, the execution $$ \begin{aligned} & (\code{AS$'$} \mid \square'') = (\code{while Wn T do AP$''$ end (pa,m:A)} \mid \square'')  \atran{\circ}^* (\code{skip (pa,A$'$)} \mid \square') \end{aligned} $$ must be shorter than our original (since the while loop is first started). From Figure~\ref{func-def:inv+}, we have that \func{inv$^+$}{\code{AS$'$}} = \code{while Wn b do} \func{inv}{\code{AP}} \code{end (pa,A)}. Then application of the induction hypothesis of the Statement Property (Lemma~\ref{lem:sp}) on this shorter execution (which is guaranteed to be uniform and to begin with an identifier step) gives $$ (\func{inv$^+$}{\code{AS$'$}} \mid \square_1') = (\code{while Wn b do}~ \func{inv}{\code{AP}}~\code{end (pa,m:A)} \mid \square_1') \rtran{\circ}^* (\code{AR} \mid \square_1'') $$ by a sequence of rule applications SR, for program \code{AR} such that \code{AR} is a skip equivalent of some program. As this shorter execution begins with an identifier step from within the while loop body, this shorter execution is guaranteed to stop via the modified rules \sos{W2aPI} or \sos{W4aPI} as shown in Case~\ref{sp-w5a}. This means that \code{AR} = \code{while Wn T do abort end (pa,m:A)} (with abort coming from the use of the modified rules). 

Now compare the programs \func{inv$^+$}{\code{AS}} and \func{inv$^+$}{\code{AS$'$}}. Since both begin with the same prefix of execution, then by the same sequence of rules SR, we obtain $$ (\func{inv$^+$}{\code{AS}} \mid \square_1') = (\code{while Wn b do } \func{inv}{\code{AP}} \code{ end (pa,A)} \mid \square_1') \rtran{\circ}^* (\code{AT$''$} \mid \square_1'') $$ for some programs \code{AT$''$}. With the program \code{AR} containing an abort statement to stop the execution at the corresponding point, \code{AT$''$} will be identical but without this abort. However, since this will be replaced with \code{skip}, a single application of the skip rule \sos{W6a} occurs to reset the while loop. Therefore \code{AT$''$} = \code{while Wn b do } \func{inv}{\code{AP}} \code{end (pa,m:A)}.

Consider the current state $\square_1''$. We have \code{m} = \func{previous}{} (with all subsequently used identifiers reversed via the induction hypothesis), the identifier stack for this loop is equal to  \code{m:A} (\code{m} as its head), $\beta_1''(\code{Wn})$ = \emph{def} (there is a mapping for this loop since it will have already started) and that $\delta''_1$(\code{W}) = \code{(m,F,F)} (the stack \code{W} has the triple \code{(m,F,F)} as its head). All premises of the rule \sos{W2r} are valid (and as such mean no other rule is applicable), giving us $$ (\code{while Wn b do } \func{inv$^+$}{\code{AP}} \code{ end (pa,m:A)} \mid \square''_1) \rtran{m} (\code{skip (pa,A)} \mid \square_1) $$ for some program state $\square_1 = \square''_1[\delta''_1[\code{W/W$'$}], \beta''_1[\code{Wn}]]$, such that $\square_1 \approx \square$ as required. The program state matches at each required point with our desired execution, meaning we can take \code{AT$'$} to be \code{skip (pa,A)} and \code{AT} to be \code{AT$'$}. Since \code{AT} has reached skip, this case is shown to be valid. 
\end{case}

\begin{case}{\textbf{(Loop \sos{W4aP})}} \label{sp-w4a}
Consider the loop condition evaluation for any iteration that is not the first or the last (each have separate rules \sos{W1a}/\sos{W3a}/\sos{W2a}). To trigger the modified rule \sos{W4aP}, an abort is inserted manually into the initial program and therefore no second copies are required. This cannot occur during normal execution and the rule \sos{W4a} will be applied instead. Let \code{AS} = \code{while Wn b do abort end (pa,A)} and the initial program state $\square$ = $\allstores$. Assume the following execution exists with the first transition via \sos{W4aP} $$ \begin{aligned} & (\code{AS} \mid \square) = (\code{while Wn b do abort end (pa,A)} \mid \square)  \\&\phantom{(\code{AS} \mid \square)} \atran{m} (\code{while Wn T do AP$'$ end (pa,m:A)} \mid \square'') \\&\phantom{(\code{AS} \mid \square)} \atran{}_s^* (\code{while Wn T do AP$''$ end (pa,m:A)} \mid \square'')  \atran{\circ}^* (\code{skip (pa,A$'$)} \mid \square') \end{aligned} $$ for some programs \code{AP$''$}, \code{AP$'$} such that \code{AP$'$} is a renamed version of \code{AP} and states $\square''$ = $\square$[$\delta$[\code{(m,T,T) $\rightarrow$ W}], $\beta$[\code{Wn $\rightarrow$ while Wn b do AP$'$ end (pa,m:A)}]], and $\square'$. Application of the rule \sos{W2aP} means the premises were valid at that point, namely that \code{m} = \func{next}{}, $\beta$(\code{Wn}) = \code{while Wn b do AP end (pa,m:A)} (there is no mapping for this while loop as it has not yet started), $ (\code{b pa} \mid \square) \hookrightarrow_b^* (\code{T} \mid \square) $ (the condition evaluates to true) and the body of the loop is equal to \code{abort}. Note \code{A$'$} is a version of the persistent identifier stack for this loop statement that ends with the sub-stack \code{(m:A)}. From this point on, we use \code{AS$'$} to denote \code{while Wn T do AP$''$ end (pa,m:A)}. 

We need to show that there exists an execution $$(\func{inv$^+$}{\code{AS}} \mid \square_1') \rtran{\circ}^* (\code{AT$''$} \mid \square_1'') \rtran{m} (\code{AT$'$} \mid \square_1) \rtran{}_s^* (\code{AT} \mid \square_1) $$ for some statements \code{AT$''$}, \code{AT$'$}, \code{AT}, and program states $\square_1'$, $\square_1''$, $\square_1$ such that $\square_1' \approx \square'$, $\square_1'' \approx \square''$ and $\square_1 \approx \square$. From Figure~\ref{func-def:inv+}, we have that \func{inv$^+$}{\code{AS}} = \code{while Wn b do} \func{inv}{\code{AP}} \code{end (pa,A$'$)}. 

We first note that the execution $$ \begin{aligned} &(\code{while Wn T do AP$'$ end (pa,m:A)} \mid \square'') \\ &\atran{}_s^* (\code{while Wn T do AP$''$ end (pa,m:A)} \mid \square'') \end{aligned} $$ will have length in all cases except where the loop body \code{AP$'$} begins with at least one skip step. In this situation, the inverted version of this group of skip steps will be applied immediately following the previous inverse identifier step. This means due to uniformity that these will be given within the induction hypothesis use on the rest of the forward execution. 

Returning to our assumed transition sequence above, the execution $$ \begin{aligned} & (\code{AS$'$} \mid \square'') = (\code{while Wn T do AP$''$ end (pa,m:A)} \mid \square'')  \atran{\circ}^* (\code{skip (pa,A$'$)} \mid \square') \end{aligned} $$ must be shorter than our original (since the iteration of the loop must have first been started). From Figure~\ref{func-def:inv+}, we have that \func{inv$^+$}{\code{AS$'$}} = \code{while Wn b do} \func{inv}{\code{AP}} \code{end (pa,A$'$)}. Then application of the induction hypothesis of the Statement Property (Lemma~\ref{lem:sp}) on this shorter execution (which is guaranteed to be uniform and to begin with an identifier step) gives $$ (\func{inv$^+$}{\code{AS$'$}} \mid \square_1') = (\code{while Wn b do } \func{inv}{\code{AP}} \code{ end (pa,A$'$)} \mid \square_1') \rtran{\circ}^* (\code{AR} \mid \square_1'') $$ by a sequence of rule applications SR, for program \code{AR} such that \code{AR} is a skip equivalent of some program. As in Case~\ref{sp-w3a}, this shorter execution (that must start with an identifier rule) is guaranteed to stop via the modified rules \sos{W2aPI} or \sos{W4aPI} as shown in Case~\ref{sp-w5a}. This means that \code{AR} = \code{while Wn T do abort end (pa,m:A)} (with abort coming from the use of the modified rules). 

Now compare the programs \func{inv$^+$}{\code{AS}} and \func{inv$^+$}{\code{AS$'$}}. As in Case~\ref{sp-w3a}, both start with identical transition sequences meaning the sequence of rules SR gives $$ (\func{inv$^+$}{\code{AS}} \mid \square_1') = (\code{while Wn b do } \func{inv}{\code{AP}} \code{ end (pa,A$'$)} \mid \square_1') \rtran{\circ}^* (\code{AT$''$} \mid \square_1'') $$ for some programs \code{AT$''$}. With the program \code{AR} containing an abort statement to stop the execution at the corresponding point, \code{AT$''$} will be identical but without this abort. However, since this will be replaced with \code{skip}, a single application of the skip rule \sos{W6a} occurs to reset the while loop. Therefore \code{AT$''$} = \code{while Wn b do } \func{inv$^+$}{\code{AP}} \code{end (pa,m:A)}.

We are now in the program state $\square_1''$. We have that \code{m} = \func{previous}{} (with all subsequently used identifiers reversed via the induction hypothesis), the identifier stack for this loop will have \code{m} as its head, $\beta_1''(\code{Wn})$ = \emph{def} (there is a mapping for this loop since it will have already started), that $\delta''_1$(\code{W}) = \code{(m,T,T)} (the stack \code{W} has the triple \code{(m,T,T)} as its head) and $\delta''_1$(\code{WI}) = \code{m,C} (the stack \code{WI} has the pair \code{(m,C)} as its head). This means all premises of the rule \sos{W4rP} to be valid (and that no other rule is applicable). Therefore we have $$ \begin{aligned} & (\code{while Wn b then } \func{inv}{\code{AP}} \code{ end (pa,m:A)} \mid \square''_1) \\ & \rtran{m} (\code{while Wn b then abort end (pa,A)} \mid \square_1) \end{aligned} $$ for some program state $\square_1 = \square''_1[\delta''_1[\code{W/W$'$}], \beta''_1[\code{Wn}]]$, such that $\square_1 \approx \square$ as required. With no skip steps available and all stores matching at the required positions, we can take \code{AT$'$} to be \code{while Wn b then abort end (pa,A)} (a valid skip equivalent) and \code{AT} to be \code{AT$'$}, showing our desired execution to be valid. Therefore this holds. 
\end{case}

\begin{case}{\textbf{(Loop \sos{W5a})}} \label{sp-w5a}
Consider an identifier rule from within a loop body. Let \code{AS} = \code{while Wn T do AP end (pa,A)} and the initial state $\square$ = $\allstores$. Assume the following execution exists with the first step via \sos{W5a} $$ \begin{aligned} & (\code{AS} \mid \square) = (\code{while Wn b do AP end (pa,A)} \mid \square) \\&\phantom{(\code{AS} \mid \square)}  \atran{m} (\code{while Wn T do AP$'$ end (pa,m:A)} \mid \square'') \\&\phantom{(\code{AS} \mid \square)} \atran{}_s^* (\code{while Wn T do AP$''$ end (pa,m:A)} \mid \square'') \\&\phantom{(\code{AS} \mid \square)}  \atran{\circ}^* (\code{skip (pa,A$'$)} \mid \square') \end{aligned} $$ for some programs \code{AP$'$} and \code{AP$''$}, and program states $\square''$ and $\square'$. Note \code{A$'$} is a modified version of the identifier stack that ends with the sub-stack \code{m:A}. From here, let \code{AS$'$} = \code{while Wn T do AP$''$ end (pa,m:A)}. The while loop must have started, meaning $\beta$(\code{Wn}) = \code{while Wn b do AQ end (pa,A)}.

We need to show that there exists an execution $$(\func{inv$^+$}{\code{AS}} \mid \square_1') \rtran{\circ}^* (\code{AT$''$} \mid \square_1'') \rtran{m} (\code{AT$'$} \mid \square_1) \rtran{}_s^* (\code{AT} \mid \square_1) $$ for statements \code{AT$''$}, \code{AT$'$}, \code{AT}, and states $\square_1'$, $\square_1''$, $\square_1$ such that $\square_1' \approx \square'$, $\square_1'' \approx \square''$ and $\square_1 \approx \square$. From Figure~\ref{func-def:inv+}, we have that \func{inv$^+$}{\code{AS}} = \code{while Wn b do} \func{inv}{\code{AP}} \code{end (pa,A)}. 

In order to correctly stop the inverse execution, we must consider two possible cases. The first case is where the first transition from our assumed execution is from within the final iteration of the loop, and the second case is where it is not. Since the function $\func{inv$^+$}{}$ cannot insert the abort as further iterations may need to be inverted first, each of our cases uses one of our modified semantic rules, namely \sos{W2aPI} or \sos{W4aPI} respectively. We rewrite our execution to show the next identifier step made by the loop itself (to next evaluate the condition), where the partial program will be saved. This will then ensure that the inverse partial rules are used, namely \sos{W3rPI} and \sos{W4rPI} to insert the abort at the correct position. Recall the syntax of while loops for these rules to apply require a copy of the partial body \code{AP} to be added.

\vspace{.3cm}
\noindent \textbf{Final iteration} Assuming the transition is from within  the final iteration, our assumed execution can be rewritten as $$ \begin{aligned} & (\code{AS} \mid \square) = (\code{while Wn b do AP AP end (pa,A)} \mid \square) \\&\phantom{(\code{AS} \mid \square)} \atran{m} (\code{while Wn T do AP$'$ AP end (pa,m:A)} \mid \square'') \\&\phantom{(\code{AS} \mid \square)} \atran{}_s^* (\code{while Wn T do AP$''$ AP end (pa,m:A)} \mid \square'') \\&\phantom{(\code{AS} \mid \square)} \atran{\circ}^* (\code{while Wn T do skip I AP end (pa,m:A)} \mid \square''') \\&\phantom{(\code{AS} \mid \square)} \atran{}_s (\code{while Wn T do AQ AP end (pa,m:A)} \mid \square''')  \atran{n} (\code{skip (pa,A$'$)} \mid \square') \end{aligned} $$ for some program state $\square'''$ and program \code{AQ} such that \code{AQ} is the full loop body (retrieved from the while environment), where the penultimate transition is via \sos{W6a} and the final step (using \code{n}) is via \sos{W2aPI}.

Repeated use of \sos{W5a} (from conclusion to premises) gives us $$ \begin{aligned} & (\code{AP} \mid \square)  \atran{m} (\code{AP$'$} \mid \square'')  \atran{}_s^* (\code{AP$''$} \mid \square'')  \atran{\circ}^* (\code{skip I} \mid \square''')  \end{aligned} $$ which must be shorter than our original execution (as the while loop itself has to finish). Applying the induction hypothesis of Part~2 of Lemma~\ref{lem:pp} means we obtain $$ \begin{aligned} & (\func{inv$^+$}{\code{AP}} \mid \square_1''')  \rtran{\circ}^* (\code{AR$''$} \mid \square_1'')  \rtran{m} (\code{AR$'$} \mid \square_1)  \rtran{}_s^* (\code{AR} \mid \square_1)  \end{aligned} $$ for some programs \code{AR$''$}, \code{AR$'$}, \code{AR} such that \code{AR} is a valid skip equivalent of some program. Repeated use of the rule \sos{W5r} gives \begin{equation} \label{eq:w5a-part-1}
\begin{split}  &(\code{while Wn T do } \func{inv$^+$}{\code{AP}} \code{ end (pa,m:A)} \mid \square_1''')  \\&\phantom{} \rtran{\circ}^* (\code{while Wn T do AR$''$ end (pa,m:A)} \mid \square_1'') \\&\phantom{} \rtran{m} (\code{while Wn T do AR$'$ end (pa,A)} \mid \square_1'') \\&\phantom{} \rtran{}_s^* (\code{while Wn T do AR end (pa,A)} \mid \square_1) \end{split}
\end{equation}
from which a single application of \sos{W6r} will then reset the while loop (recall a while loop is a special case of skip equivalent). This gives us the majority of our desired execution, with only the last step of the inverse execution left to consider. From our assumed execution (with the copy of the loop body omitted), $$ (\code{while Wn T do AQ end (pa,m:A)} \mid \square''') \atran{n} (\code{skip (pa,A$'$)} \mid \square') $$ takes a single step, with no skip steps available after. The corresponding inverse rule \sos{W3aPI} gives the execution (starting from the inverted complete program as this must be very first step) \begin{equation} \label{eq:w5a-part-2}
\begin{split} &(\func{inv$^+$}{\code{AS}} \mid \square_1') \rtran{n} (\code{while Wn b do } \func{inv$^+$}{\code{AP}} \code{ end (pa,m:A)} \mid \square_1''') \end{split}
\end{equation}
as expected. From here, composition of the executions \ref{eq:w5a-part-2} and \ref{eq:w5a-part-1} (in that order) gives our desired execution and therefore shows that this part of the case holds.

\vspace{.3cm}
\noindent \textbf{nth iteration} We now consider executions where the first step is from the loop body of any iteration of a loop other than the last. Such executions can be rewritten as $$ \begin{aligned} & (\code{AS} \mid \square) = (\code{while Wn b do AP AP end (pa,A)} \mid \square) \\&\phantom{(\code{AS} \mid \square)} \atran{m} (\code{while Wn T do AP$'$ AP end (pa,m:A)} \mid \square'') \\&\phantom{(\code{AS} \mid \square)} \atran{}_s^* (\code{while Wn T do AP$''$ AP end (pa,m:A)} \mid \square'') \\&\phantom{(\code{AS} \mid \square)} \atran{\circ}^* (\code{while Wn T do skip I AP end (pa,m:A)} \mid \square''') \\&\phantom{(\code{AS} \mid \square)} \atran{}_s (\code{while Wn T do AQ AP end (pa,m:A)} \mid \square''') \\&\phantom{(\code{AS} \mid \square)} \atran{n} (\code{while Wn T do AQ$'$ skip end (pa,n:m:A)} \mid \square'''') \\&\phantom{(\code{AS} \mid \square)} \atran{}_s^* (\code{while Wn T do AQ$''$ skip end (pa,n:m:A)} \mid \square'''') \\&\phantom{(\code{AS} \mid \square)} \atran{\circ}^* (\code{skip (pa,A$'$)} \mid \square') \end{aligned} $$ where the identifier step using \code{n} is via \sos{W4aPI}, for some program state $\square'''$ and programs \code{AQ}, \code{AQ$'$}, \code{AQ$''$} such that \code{AQ} is the full loop body (retrieved from the while environment).

Our proof of such an execution follows closely to that for the first iteration above. The major difference is there is an extra use of the induction hypothesis, specifically on the execution $$ (\code{while Wn T do AQ$''$ end (pa,n:m:A)} \mid \square'''')  \atran{\circ}^* (\code{skip (pa,A$'$)} \mid \square') $$ that does not appear in the first part of this case. All other parts of the proof follow as above. We therefore conclude that this part of the case holds, meaning we have shown this case to hold. 
\end{case}

\vspace{.3cm}
\noindent \textbf{Block statements}

\begin{case}{\textbf{Block \sos{B1a}}} \label{sp-b1a}
This concerns identifier step from within a block. This follows Case~\ref{sp-i2a}, but \sos{B2a} and \sos{B2a} in place of \sos{I2a} and \sos{I4a} respectively. 
\end{case}

Since we assume local variables, arrays and procedures are only declared within block statements, we consider each such case here. 

\begin{case}{\textbf{Variable declaration \sos{L1a}}} \label{sp-l1a}
Declaring a local variable completes in a single step. This follows Case~\ref{sp-d1a} and so it omitted. 
\end{case}

\begin{case}{\textbf{Procedure declaration \sos{L2a}}} \label{sp-l2a}
A single application of \sos{L2a} declares a procedure. A proof of this case is omitted as it follows Case~\ref{sp-d1a}.
\end{case}

\begin{case}{\textbf{Array declaration \sos{L3a}}} \label{sp-l3a}
An array is declared in a single step. This follows the previous declaration statements, with the small difference being that \code{n} memory locations are ``reserved''.
\end{case}

\begin{case}{\textbf{Variable removal \sos{H1a}}} \label{sp-h1a}
As with the declaration of a local variable, removal takes a single step. A proof is not shown as it follows closely to Case~\ref{sp-d1a}. 
\end{case}

\begin{case}{\textbf{Procedure removal \sos{H2a}}} \label{sp-h2a}
The rule \sos{H2a} takes a single step to remove a procedure.  A proof is not listed as it follows Case~\ref{sp-d1a}.
\end{case}

\begin{case}{\textbf{Array removal \sos{H3a}}} \label{sp-h3a}
An array is removed via \sos{H3a} in a single step. This follows previous removal cases, but with the values of multiple memory locations saved before their deletion.
\end{case}

\vspace{.3cm}
\noindent \textbf{Procedure call statements}

\begin{case}{\textbf{Call \sos{G1a}}} \label{sp-g1a}
Consider the opening of a procedure call, which declares a local copy of the procedure body and produces the corresponding \code{runc} construct. This follows with the opening of a conditional statement in Case~\ref{sp-i1at} (though with changes made to the procedure environment $\mu$) using rules \sos{G1a} and \sos{G3r} in place of \sos{I1aT} and \sos{I4r}, and so is omitted.  
\end{case}

\begin{case}{\textbf{Call \sos{G2a}}} \label{sp-g2a}
Consider an execution beginning with an identifier step from within a procedure body. This follows the case of such steps from within a conditional statement, namely Case~\ref{sp-i2a}. A partial execution beginning within a call statement uses the modified rule \sos{G3aPI}, and follows as in Case~\ref{sp-w5a}. 
\end{case}

\begin{case}{\textbf{Call \sos{G3a}}} \label{sp-g3a}
The closing of a procedure call will remove the local copy of the procedure body for this specific call statement, after saving any identifiers assigned to it. This case follows with the closing of a conditional statement, namely Case~\ref{sp-i4a} with rules \sos{G3a} and \sos{G1r} in place of \sos{I4a} and \sos{I1rT}. A partial execution beginning at the closure of a block uses the modified rule \sos{G3aP} and follows Case~\ref{sp-w2a}. 
\end{case}

All inductive cases are therefore valid. Since all base cases also hold, we can conclude that Lemma~\ref{lem:sp} (Statement Property) is valid.
\end{proof}

%% file: pp-proof.tex
In this section we give a proof of Lemma~\ref{lem:pp} (Program Property). By abuse of notation, we use \code{AP} to represent either a complete or partially executed program. In cases considering full program execution (e.g. Case~\ref{pp-part1-s2a}), we ignore the abort statement introduced by \func{inv$^+$}{} as the program stops as desired.  

\lempp*
\setcounter{lemma}{10}

All executions from here are uniform and so we omit the specific notation. For example, a uniform execution $(\code{P} \mid \square) \uatran{m}{} (\code{P$'$} \mid \square')$ is written as $(\code{P} \mid \square) \atran{m} (\code{P$'$} \mid \square')$.

\begin{proof}
This proof is by mutual induction of the Program Property (this lemma) and the Statement Property (Lemma~\ref{lem:sp}), on the length of the executions $ (\code{AP} \mid \square) \atran{\circ}^* (\code{skip I} \mid \square') $ and $ (\code{AS} \mid \square) \atran{\circ}^* (\code{skip I} \mid \square') $ respectively. We now consider the two types of possible executions, namely those that start with at least one skip step (\textbf{Part~1}) and those that start with an identifier step (\textbf{Part~2}).

\subsection{Proof of Part 1} \label{appen-proof-part-1}
Consider all executions that begin with at least one skip step, of the form $(\code{AP} \mid \square) \atran{}_s^* (\code{AP$'$} \mid \square) \atran{\circ}^* (\code{skip I} \mid \square')$, where \code{AP} is a complete program of the following syntax as no partial programs can begin with a skip step.
\begin{align*}
\code{P} &::= \code{S} ~|~ \code{skip; AP} ~|~ \code{P par P} \\
\code{S} &::= \code{skip} ~|~  \code{begin Bn P AP end} 
\end{align*} %

Since no executions of length 0 exist, we consider our base cases to be any execution of length 1. There are three examples of such cases: two sequentially composed skip statements (Case~\ref{pp-part1-s2a}), a parallel statement where both component programs are skip (Case~\ref{pp-part1-p3a}) and a block statement whose body is skip (Case~\ref{pp-part1-b2a}). 

\begin{case}{\textbf{Skips \sos{S2a}}} \label{pp-part1-s2a}
Consider two sequentially composed skip statements. Let \code{AP} be the program \code{skip I; skip I$'$}, and $\square$ be the tuple of initial program state environments $\allstores$. Assume the following uniform execution via the rule~\sos{S2a} $$ (\code{skip I; skip I$'$} \mid \square) \atran{}_s (\code{skip I$'$;} \mid \square) $$ that completes in a single step and does not alter the program state.

We need to show that there exists an execution $$(\func{inv$^+$}{\code{AP}} \mid \square_1') \rtran{\circ}^* (\code{AT$'$} \mid \square_1) \rtran{}_s^* (\code{AT} \mid \square_1) $$ for some statements \code{AT$'$}, \code{AT}, and program states $\square_1'$, $\square_1$ such that $\square_1' \approx \square'$ and $\square_1 \approx \square$. By \ref{inv+:sc} and \ref{inv+:skip} in Figure~\ref{func-def:inv+}, we have that \func{inv$^+$}{\code{AP}} = \code{skip I$'$; skip~I} (ignoring the abort).   

With no further execution, we can take \code{AT$'$} = \func{inv$^+$}{\code{AP}} and $\square_1$ = $\square_1'$. A single use of \sos{S2r} can be applied to remove the first skip statement while not changing the program state, giving the execution $$ (\code{AT$'$} \mid \square_1) \rtran{}_s (\code{skip I} \mid \square_1) $$ Taking \code{AT} = \code{skip I}, we have shown our desired execution and this case to hold.
\end{case}

\begin{case}{\textbf{Empty Parallel \sos{P3a}}} \label{pp-part1-p3a}
Consider a program containing only a parallel statement where both component programs are single skip statements. Let \code{AP} be the program \code{skip I par skip I$'$}, and let $\square$ be the tuple of initial program state environments $\allstores$. Assume the uniform execution via the rule \sos{P3a} $$ (\code{skip I par skip I$'$} \mid \square) \atran{}_s (\code{skip} \mid \square) $$ with length 0 and no effect on the program state.

We need to show that there exists an execution $$(\func{inv$^+$}{\code{AP}} \mid \square_1') \rtran{\circ}^* (\code{AT$'$} \mid \square_1) \rtran{}_s^* (\code{AT} \mid \square_1) $$ for some statements \code{AT$'$}, \code{AT}, and program states $\square_1'$, $\square_1$ such that $\square_1' \approx \square'$ and $\square_1 \approx \square$. By \ref{inv+:pc} and \ref{inv+:skip} in Figure~\ref{func-def:inv+}, we have \func{inv$^+$}{\code{AP}} = \code{skip I par skip I$'$} (ignoring the abort).

With no further execution, we can take \code{AT$'$} = \func{inv$^+$}{\code{AP}} and $\square_1$ = $\square_1'$. A single application of \sos{P3r} can close the inverted parallel statement while not changing the program state, giving the execution $$ (\code{AT$'$} \mid \square_1) \rtran{}_s (\code{skip} \mid \square_1) $$ Taking \code{AT} = \code{skip}, this case has been shown to hold.
\end{case}

\begin{case}{\textbf{Empty Block \sos{B2a}}} \label{pp-part1-b2a}
Consider a program containing a single block statement, whose body is a single skip statement. Let \code{AP} be the program \code{begin Bn skip I skip I end (pa,A)} (with the original block body also being skip as expected), and $\square$ be the tuple of initial program state environments $\allstores$. Assume the uniform execution via \sos{B2a} $$ (\code{begin Bn skip I skip I end (pa,A)} \mid \square) \atran{}_s (\code{skip I$'$} \mid \square) $$ with length 1 and no effect on the program state.

We need to show that there exists an execution $$(\func{inv$^+$}{\code{AP}} \mid \square_1') \rtran{\circ}^* (\code{AT$'$} \mid \square_1) \rtran{}_s^* (\code{AT} \mid \square_1) $$ for some statements \code{AT$'$}, \code{AT}, and states $\square_1'$, $\square_1$ such that $\square_1' \approx \square'$ and $\square_1 \approx \square$. By Figure~\ref{func-def:inv+}, we have \func{inv$^+$}{\code{AP}} = \code{begin Bn skip I skip I end~(pa,A)} (ignoring the abort).

With no further execution, we can take \code{AT$'$} = \func{inv$^+$}{\code{AP}} and $\square_1$ = $\square_1'$. Application of \sos{B2r} closes the inverted block statement via  $$ (\code{AT$'$} \mid \square_1) \rtran{}_s (\code{skip} \mid \square_1) $$ Since \code{AT} = \code{skip}, this case holds.
\end{case}

Each base case has been shown to valid, meaning Part~1 of Lemma~\ref{lem:pp} (Program Property) holds for all executions of length 1. We now consider all inductive cases. Assume the Program Property holds for all programs \code{AR} and program states $\square^*$ such that the execution $ (\code{AR} \mid \square^*) \atran{}_s (\code{AR$'$} \mid \square^*) \atran{\circ}^* (\code{skip I} \mid \square^{*'}) $ has length $k$ (where $k \geq 1$). Now assume that the execution $ (\code{AP} \mid \square) \atran{}_s (\code{AP$'$} \mid \square) \atran{\circ}^* (\code{skip I} \mid \square') $ has length $l$ such that $l > k$.

Each inductive case is now shown. This consists of executions of complete programs beginning with skip steps. These are skip steps from within a block statement (Case~\ref{pp-part1-b1a}), and as a result of sequential composition (Case~\ref{pp-part1-s1a}) and parallel composition (Cases~\ref{pp-part1-p1a}--\ref{pp-part1-p2a}). We note that there are no cases for steps via a conditional branch or loop body as these cannot be the first step of a complete program. Each require at least one previous step to open the conditional or the loop. There is also no case for the skip step \sos{W6a} with the same reasoning.

\begin{case}{\textbf{(Block \sos{B1a})}} \label{pp-part1-b1a} Consider a block statement containing a body that begins with at least one skip step. Let \code{AP} be \code{begin Bn AQ AQ end} and $\square$ be the initial program state. If \code{AQ} contains only skip statements and execution of those mean the block itself will close, then this is a version of the (base) Case~\ref{pp-part1-b2a}. Therefore assume \code{AQ} contains more than just skip steps. Assume the following uniform execution $$ (\code{begin Bn AQ AQ end} \mid \square) \atran{}_s^* (\code{begin Bn AQ$'$ AQ end} \mid \square) \atran{\circ}^* (\code{skip} \mid \square') $$ exists for some program \code{AQ$'$} and state $\square'$. Crucially all available skip steps are performed, meaning the first step of the remaining execution must be via an identifier rule. We note that this execution can be rewritten to highlight the point at which the body finishes, namely as $$ \begin{aligned} &(\code{begin Bn AQ AQ end} \mid \square) \atran{}_s^* (\code{begin Bn AQ$'$ AQ end} \mid \square) \\& \atran{\circ}^* (\code{begin Bn skip I AQ end} \mid \square') \atran{}_s (\code{skip} \mid \square') \end{aligned} $$ where the final transition is via \sos{B2a}. 

We need to show that there exists an execution $$(\func{inv$^+$}{\code{AP}} \mid \square_1') \rtran{\circ}^* (\code{AT$'$} \mid \square_1) \rtran{}_s^* (\code{AT} \mid \square_1) $$ for some statements \code{AT$'$}, \code{AT}, and program states $\square_1'$, $\square_1$ such that $\square_1' \approx \square'$ and $\square_1 \approx \square$. By Figure~\ref{func-def:inv+}, we have \func{inv$^+$}{\code{AP}} = \code{begin Bn} \func{inv}{\code{AQ}} \func{inv}{\code{AQ}} \code{end}.

From our rewritten assumed execution, repeated use of the rule \sos{B1a} (from conclusion to premises), we have the uniform execution $$ (\code{AQ} \mid \square) \atran{}_s^* (\code{AQ$'$} \mid \square) \atran{\circ}^* (\code{skip I} \mid \square')  $$ which is guaranteed to be shorter (as the block must close). By the induction hypothesis of Part~1 of Lemma~\ref{lem:pp} (Program Property), we get $$ (\func{inv$^+$}{\code{AQ}} \mid \square''_1) \rtran{\circ}^* (\code{AR$'$} \mid \square_1) \rtran{}_s^* (\code{skip I} \mid \square_1) $$ for some program \code{AR$'$} and states $\square_1''$ and $\square_1$, which will reach skip as it must be a full program. We note that since \code{AQ} must be a complete program, \func{inv$^+$}{\code{AQ}} = \func{inv}{\code{AQ}}. Then using this substitution and repeated use of the corresponding inverse rule \sos{B1r}, we get $$ \begin{aligned} & (\code{begin Bn } \func{inv}{\code{AQ}} ~\func{inv}{\code{AQ}} \code{ end} \mid \square'_1) \rtran{\circ}^* (\code{begin Bn AR$'$ } \func{inv}{\code{AQ}} \code{ end} \mid \square_1) \\& \rtran{}_s^* (\code{begin Bn skip I } \func{inv}{\code{AQ}} \code{ end} \mid \square_1) \end{aligned} $$

Finally, since \code{AP} must be a complete program, the inverse execution is required to reach skip. A single application of the skip step \sos{B2a} will close this block and allow the whole execution to reach skip. Therefore we have the execution $$ \begin{aligned} & (\code{begin Bn } \func{inv}{\code{AQ}} ~\func{inv}{\code{AQ}} \code{ end} \mid \square'_1) \rtran{\circ}^* (\code{begin Bn AR$'$ } \func{inv}{\code{AQ}} \code{ end} \mid \square_1) \\& \rtran{}_s^* (\code{begin Bn skip I } \func{inv}{\code{AQ}} \code{ end} \mid \square_1) \rtran{}_s (\code{skip} \mid \square_1) \end{aligned} $$ that matches our desired execution and so shows this case to be valid.
\end{case}


\begin{case}{\textbf{(Sequential Composition \sos{S1a})}} \label{pp-part1-s1a}
Consider sequential composition. Let \texttt{AS} be a statement, \texttt{AP} be a program, and $\square$ be the initial program state environments $\allstores$. Assume a program of the form \code{AS;AP}, and the execution via the initial skip rule \sos{S1a} $$ (\code{AS;AP} \mid \square) \atran{}_s^* (\code{AS$'$;AP} \mid \square) \atran{\circ}^* (\code{skip I} \mid \square') $$ for some statement \code{AS$'$}, a pair \code{I} and program state $\square'$. As in previous cases, note that all available skip steps are performed and the first step of the remaining execution must be via a identifier rule. 

We need to show that there exists an execution $$(\func{inv$^+$}{\code{AS;AP}} \mid \square_1') \rtran{\circ}^* (\code{AT$'$} \mid \square_1) \rtran{}_s^* (\code{AT} \mid \square_1) $$ for some programs \code{AT$'$}, \code{AT} such that \code{AT} is skip or a skip equivalent, and program states $\square_1'$, $\square_1$ such that $\square_1' \approx \square'$ and $\square_1 \approx \square$. By \ref{inv+:sc} in Figure~\ref{func-def:inv+}, we have that \func{inv$^+$}{\code{AS;AP}} = \func{inv}{\code{AP}};\func{inv$^+$}{\code{AS}}.   

Our assumed execution can be rewritten to highlight the point at which \code{AS} concludes, namely as $$ \begin{aligned} &(\code{AS;AP} \mid \square) \atran{}_s^* (\code{AS$'$;AP} \mid \square) \atran{\circ}^* (\code{skip I$'$;AP} \mid \square'') \\ &  \atran{}_s (\code{AP} \mid \square'') \atran{\circ}^* (\code{skip I} \mid \square') \end{aligned} $$ for program state $\square''$. From this rewritten version of the execution, the application of the rule \sos{S1a} repeatedly (from conclusion to premises), we can get the execution $$ (\code{AS} \mid \square) \atran{}_s^* (\code{AS$'$} \mid \square) \atran{\circ}^* (\code{skip I$'$} \mid \square'') $$ which concerns only the statement \code{AS}. Therefore this is guaranteed to be shorter than our original (as \code{AP} must require at least one step). With a single statement equivalent to a program containing only that statement, the induction hypothesis of Part~1 of the Program Property (Lemma~\ref{lem:pp}) on this shorter execution gives \begin{equation} \label{eq:pp-p1-sc-1} (\func{inv$^+$}{\code{AS}} \mid \square_1'') \rtran{\circ}^* (\code{AR$'$} \mid \square_1) \rtran{}_s^* (\code{AR} \mid \square_1) \end{equation} for programs \code{AR$'$}, \code{AR} such that \code{AR} is a valid skip equivalent of some program, and program states $\square_1''$, $\square_1$ such that $\square_1'' \approx \square''$ and $\square_1 \approx \square$.

Returning to our rewritten assumed execution above, the uniform execution $(\code{AP} \mid \square'') \atran{\circ}^* (\code{skip I} \mid \square')$ is guaranteed to be both shorter than our original and to be a complete program (since it was sequentially composed with \code{AS}). In most cases this will begin with an identifier step, however it could also begin with at least one skip step. Application of the induction hypothesis of Lemma~\ref{lem:pp} (Part~1 if the execution begins with a skip step, or Part~2 if the execution begins with an identifier step) on this gives $$ (\func{inv$^+$}{\code{AP}} \mid \square_1')\rtran{\circ}^* (\code{skip I$_1$} \mid \square_1'')$$  for program state $\square_1''$ such that $\square_1'' \approx \square''$. Since \code{AP} must be a complete statement, we note that \func{inv$^+$}{\code{AP}} = \func{inv}{\code{AP}}. With this substitution, use of the rule \sos{S1r} repeatedly (from premises to conclusion) gives us \begin{equation} \label{eq:pp-p1-sc-2} (\func{inv}{\code{AP}};\func{inv$^+$}{\code{AS}} \mid \square_1')\rtran{\circ}^* (\code{skip I$_1$};\func{inv$^+$}{\code{AS}} \mid \square_1'') \end{equation} A single application of the rule \sos{S2r} will remove the skip, giving the execution \begin{equation} \label{eq:pp-p1-sc-3} (\code{skip I$_1$};\func{inv$^+$}{\code{AS}} \mid \square_1'') \rtran{}_s (\func{inv$^+$}{\code{AS}} \mid \square_1'') \end{equation} At this point, we have three equations that can easily be composed. We therefore combine \ref{eq:pp-p1-sc-2}, \ref{eq:pp-p1-sc-3} and \ref{eq:pp-p1-sc-1} in that order, giving the execution $$ \begin{aligned} & (\func{inv}{\code{AP}};\func{inv$^+$}{\code{AS}} \mid \square_1')\rtran{\circ}^* (\code{skip I$_1$};\func{inv$^+$}{\code{AS}} \mid \square_1'') \rtran{}_s (\func{inv$^+$}{\code{AS}} \mid \square_1'') \\& \rtran{\circ}^* (\code{AR$'$} \mid \square_1) \rtran{}_s^* (\code{AR} \mid \square_1) \end{aligned} $$ Taking \code{AT$''$} = \code{AR$''$}, \code{AT$'$} = \code{AR$'$} and \code{AT} = \code{AR} (\code{AR} is either skip or a skip equivalent by the induction hypothesis), we have our desired execution, meaning the case holds. 
\end{case}

\begin{case}{\textbf{(Parallel Composition \sos{P1a})}} \label{pp-part1-p1a}
Consider parallel composition. Let \texttt{AS} be a statement (of the reduced syntax for programs beginning with skip steps), \texttt{AP} and \texttt{AQ} be programs, and $\square$ be the initial program state. Assume a program of the form \code{(AS;AP) par AQ}, where \code{AQ} begins with an identifier step (this could be extended to allow skip steps from \code{AQ} which would have to be ordered in the fixed manner of after all skip steps from \code{AS}). The matching case of a program of the form \code{AQ par (AS;AP)} is considered in Case~\ref{pp-part1-p2a}. Further, any number of nested parallel statements could have been used, all of which can be reduced to a program of the form above via the use of \sos{P1a}/\sos{P2a}. Now assume an execution with all of the initial skip steps via \sos{P1a} $$ (\code{((AS;AP) par AQ)} \mid \square) \atran{}_s^* (\code{((AS$'$;AP) par AQ)} \mid \square)  \atran{\circ}^* (\code{skip I} \mid \square') $$ for some statement \code{AS$'$} and program state $\square'$. Consider the initial sequence of skip steps, where all available skip steps are performed. This means that the first step of the renaming execution must be via an identifier step. Each skip step transition has an inference tree associated with it that proves the transition to be valid. Each of these inference trees has a leaf, which in each case is an instance of one of the following three rules.  

\begin{enumerate}
\item \textbf{Rule \sos{B2a}}: the closing of a block meaning the block body is empty. \code{AS} contains \code{begin Bn skip I$_1$ end} (potentially nested within constructs), and \code{AS$'$} contains \code{skip} (with same nesting).
\item \textbf{Rule \sos{S2a}}: the removal of a skip statement. \code{AS} contains \code{skip I$_1$;AR} (with nesting), and \code{AS$'$} contains \code{AR} (with same nesting).
\item \textbf{Rule \sos{P3a}}: an empty parallel statement. \code{AS} contains \code{skip I$_1$ par skip I$_2$; AR} (potentially nested), and \code{AS$'$} = \code{skip; AR} (with the same nesting).
\end{enumerate}

We need to show that there exists an execution $$(\func{inv$^+$}{\code{(AS;AP) par AQ}} \mid \square_1') \rtran{\circ}^* (\code{AT$'$} \mid \square_1) \rtran{}_s^* (\code{AT} \mid \square_1) $$ for some programs \code{AT$'$}, \code{AT} such that \code{AT} is skip or a skip equivalent, and program states $\square_1'$, $\square_1$ such that $\square_1' \approx \square'$ and $\square_1 \approx \square$. By \ref{inv+:pc} and \ref{inv+:sc} in Figure~\ref{func-def:inv+}, we have that \func{inv$^+$}{\code{(AS;AP) par AQ}} = (\func{inv}{\code{AP}};\func{inv$^+$}{\code{AS}} \code{par} \func{inv}{\code{AQ}}.  

From our assumed transition sequence, the execution $ (\code{((AS$'$;AP) par AQ)} \mid \square)  \atran{\circ}^* (\code{skip I} \mid \square') $ is guaranteed to both be shorter than our original and to begin with an identifier step (since the overall execution is uniform meaning all skip steps are performed in the previous transition). Application of the induction hypothesis of Part~2 of Lemma~\ref{lem:pp} on this shorter uniform execution gives us $$ (\func{inv}{\code{AP}};\func{inv$^+$}{\code{AS$'$}} \code{ par } \func{inv}{\code{AQ}} \mid \square_1') \rtran{\circ}^* (\code{AT$''$} \mid \square_1) $$ by a sequence of rule applications SR, for some program \code{AT$''$} and program states $\square_1'$, $\square_1$ such that $\square_1' \approx \square'$ and $\square_1 \approx \square$.

With \code{AP} guaranteed to be a complete program, inversion will reach \code{skip I$_4$} for some \code{I$_4$}. \code{AQ} also must be a complete program meaning inversion will reach \code{skip I$_5$}. \code{AS$'$} will be a partial program (since the beginning skip steps have been performed), meaning the inversion of this will reach \code{AS$_1$} (either skip or a skip equivalent). Therefore \code{AT$''$} = \code{AS$_1$ par skip I$_5$}. 

Now compare the programs \func{inv}{\code{AP}};\func{inv$^+$}{\code{AS}} \code{par} \func{inv}{\code{AQ}} and \func{inv}{\code{AP}};\func{inv$^+$}{\code{AS$'$}} \code{par} \func{inv}{\code{AQ}}. Since both executions begin with the same prefix, by the same sequence of rule applications SR we obtain $$ (\func{inv}{\code{AP}};\func{inv$^+$}{\code{AS}} \code{ par } \func{inv}{\code{AQ}} \mid \square_1') \rtran{\circ}^* (\code{AT$'$} \mid \square_1) $$ for a program \code{AT$'$}. With the same reasoning as above, we have that \code{AT$'$} = \code{AS$_2$ par skip I$_5$}. The format of \code{AS$_2$} is identical to \code{AS$_1$} but with further execution of skip steps only available. All of the constructs that were closed (or removed) during the initial skip steps of the forward execution will appear in inverted form within the inverse program. In order to complete the inverse execution, all of these constructs must also be closed (or removed) here. The skip steps that will be required will depend on the types of construct that began the forward execution. We therefore return to our three cases of skip step and consider each in turn.

\begin{enumerate}
\item \textbf{Rule \sos{B2r}}: The initially empty block will be here in inverted form. Therefore the skip rule \sos{B2r} will close this block to skip.
\item \textbf{Rule \sos{S2r}}: A hard-coded skip statement will appear in the inverted program, with the rule \sos{S2r} used to remove it.
\item \textbf{Rule \sos{P3r}}: A hard-coded empty parallel statement appears identically in the inverted program. The inverse skip step \sos{P3r} is used to close it. 
\end{enumerate}
In all cases, we have that there exists the execution $$ (\code{AT$'$} \mid \square_1) \rtran{}_s^* (\code{AT} \mid \square_1) $$ such that \code{AT} = \code{skip}. 

At this point, recall discussion of non-matching skip steps from a forward and reverse execution. Since all skip steps cannot alter the program state, any mismatch here is not a problem. Therefore we have shown the desired execution and this case to hold.
\end{case}

\begin{case}{\textbf{(Parallel Composition \sos{P2a})}} \label{pp-part1-p2a}
This case follows closely to Case~\ref{pp-part1-p1a}, using the rule \sos{P2a} in place of \sos{P1a} to represent the situation where the execution begins with skip rules from the right hand side.
\end{case}

With all inductive cases shown to be valid, we have therefore completed the proof of Part~1 of Lemma~\ref{lem:pp}.

\subsection{Proof of Part 2}
We now consider all executions beginning with an identifier step, of the form $(\code{AP} \mid \square) \atran{m} (\code{AP$'$} \mid \square) \atran{}_s^* (\code{AP$''$} \mid \square) \atran{\circ}^* (\code{skip I} \mid \square')$. With no executions of length 0, we consider our base cases to be any execution of length 1. The two base cases include a program containing a single assignment (with equivalent cases for each type of  assignment) and a program containing a single loop that performs zero iterations. These proofs are omitted as they follow correspondingly to Case~\ref{sp-d1a} and Case~\ref{sp-w1a} of Lemma~\ref{lem:sp} from page~\pageref{sp-d1a} (with \code{AS} replaced by \code{AP}).

All base cases are therefore valid, meaning we now consider all inductive cases. Assume the Program Property holds for all programs \code{AR} and program states $\square^*$ such that the execution $(\code{AR} \mid \square^*) \atran{m} (\code{AR$'$} \mid \square^*) \atran{}_s^* (\code{AR$''$} \mid \square^*) \atran{\circ}^* (\code{skip I} \mid \square^{*'})$ has length $k$ (where $k \geq 1$). Now assume that the execution $ (\code{AR} \mid \square) \atran{m} (\code{AR$'$} \mid \square) \atran{}_s^* (\code{AR$''$} \mid \square) \atran{\circ}^* (\code{skip I} \mid \square') $ has length $l$ such that $l > k$.

Each inductive case is now considered. We note that the execution of a program containing a single statement is equivalent to an execution of that single statement. This means all such executions have been considered within the proof of Lemma~\ref{lem:sp} and so are omitted here. We therefore consider the remaining cases for sequential composition (Case~\ref{pp-part2-s1a}) and for parallel composition (Case~\ref{pp-part2-p1a}--\ref{pp-part2-p2a}).

\begin{case}{\textbf{(Sequential Composition \sos{S1a})}} \label{pp-part2-s1a}
Consider sequential composition. Let \texttt{AS} be a statement, \texttt{AP} be a program, and $\square$ be the tuple of initial program state environments $\allstores$. Assume a program of the form \code{AS;AP} (sequential composition), and the uniform execution with the first transition via the rule \sos{S1a} $$ (\code{AS;AP} \mid \square) \atran{m} (\code{AS$'$;AP} \mid \square'') \atran{}_s^* (\code{AS$''$;AP} \mid \square'') \atran{\circ}^* (\code{skip I} \mid \square') $$ for some statements \code{AS$'$} and \code{AS$''$}, and program states $\square''$ and $\square'$. 

We need to show that there exists an execution $$(\func{inv$^+$}{\code{AS;AP}} \mid \square_1') \rtran{\circ}^* (\code{AT$''$} \mid \square_1'') \rtran{m} (\code{AT$'$} \mid \square_1) \rtran{}_s^* (\code{AT} \mid \square_1) $$ for some statements \code{AT$''$}, \code{AT$'$}, \code{AT}, and program states $\square_1'$, $\square_1''$, $\square_1$ such that $\square_1' \approx \square'$, $\square_1'' \approx \square''$ and $\square_1 \approx \square$. By \ref{inv+:sc} in Figure~\ref{func-def:inv+}, we have that \func{inv$^+$}{\code{AS;AP}} = \func{inv}{\code{AP}};\func{inv$^+$}{\code{AS}}.   

Our assumed execution can be rewritten to show the point at which \code{AS} has executed completely, namely as $$ \begin{aligned} & (\code{AS;AP} \mid \square) \atran{m} (\code{AS$'$;AP} \mid \square'') \atran{}_s^* (\code{AS$''$;AP} \mid \square'') \atran{\circ}^* (\code{skip I$'$;AP} \mid \square''') \\ &\phantom{} \atran{}_s (\code{AP} \mid \square''') \atran{\circ}^* (\code{skip I} \mid \square') \end{aligned} $$ for some program state $\square'''$. From this rewritten version of the execution, the application of the rule \sos{S1a} repeatedly (from conclusion to premises), we can get the execution $$ (\code{AS} \mid \square) \atran{m} (\code{AS$'$} \mid \square'') \atran{}_s^* (\code{AS$''$} \mid \square'') \atran{\circ}^* (\code{skip I$'$} \mid \square''') $$ Since this concerns only \code{AS}, it is guaranteed to be shorter than the original execution as \code{AP} takes at least one step of execution. This allows application of the induction hypothesis of the Statement Property (Lemma~\ref{lem:sp}) on this shorter execution, giving~us  \begin{equation} \label{eq:pp-p2-sc-1} (\func{inv$^+$}{\code{AS}} \mid \square_1''') \rtran{\circ}^* (\code{AR$''$} \mid \square_1'') \rtran{m} (\code{AR$'$} \mid \square_1) \rtran{}^*_s (\code{AR} \mid \square_1) \end{equation} for some programs \code{AR$''$}, \code{AR$'$} and \code{AR} such that \code{AR} is either \code{skip} or a skip equivalent of some program, and some program states $\square''_1$ such that $\square_1'' \approx \square''$, and $\square_1$ such that $\square_1 \approx \square$. 

Returning to our rewritten assumed execution above, the uniform execution $(\code{AP} \mid \square''') \atran{\circ}^* (\code{skip I} \mid \square')$ is shorter than our original. This program will begin with an identifier step in the majority of cases, with a chance for it to begin with skip steps (hard coded skips, empty blocks etc.). This means that application of the induction hypothesis of Lemma~\ref{lem:pp} (Part~1 if the execution begins with a skip step, or Part~2 if the execution begins with an identifier step) on this shorter execution gives us $$ (\func{inv$^+$}{\code{AP}} \mid \square_1')\rtran{\circ}^* (\code{skip I$_1$} \mid \square_1''')$$  for some program state $\square_1'''$ such that $\square_1''' \approx \square'''$. Since \code{AP} must be a complete program (as it is sequentially composed with \code{AS}), we note that \func{inv$^+$}{\code{AP}} = \func{inv}{\code{AP}}. Using this substitution, application of the rule \sos{S1r} repeatedly (from premises to conclusion) gives \begin{equation} \label{eq:pp-p2-sc-2} (\func{inv}{\code{AP}};\func{inv$^+$}{\code{AS}} \mid \square_1') \rtran{\circ}^* (\code{skip I$_1$};\func{inv$^+$}{\code{AS}} \mid \square_1''') \end{equation} A single applicable rule is available at this point, namely \sos{S2r} to remove the skip. Since this is a skip rule and so does not affect the program state, we have the execution \begin{equation} \label{eq:pp-p2-sc-3} (\code{skip I$_1$};\func{inv$^+$}{\code{AS}} \mid \square_1''')\rtran{}_s (\func{inv$^+$}{\code{AS}} \mid \square_1''') \end{equation} At this point, the three equations \ref{eq:pp-p2-sc-2}, \ref{eq:pp-p2-sc-3} and \ref{eq:pp-p2-sc-1} can be composed in that order. This gives $$ \begin{aligned} &(\func{inv}{\code{AP}};\func{inv$^+$}{\code{AS}} \mid \square_1') \rtran{\circ}^* (\code{skip I$_1$};\func{inv$^+$}{\code{AS}} \mid \square_1''') \rtran{}_s (\func{inv$^+$}{\code{AS}} \mid \square_1''') \\& \rtran{\circ}^* (\code{AR$''$} \mid \square_1'') \rtran{m} (\code{AR$'$} \mid \square_1) \rtran{}^*_s (\code{AR} \mid \square_1)  \end{aligned} $$ Taking \code{AT$''$} = \code{AR$''$}, \code{AT$'$} = \code{AR$'$} and \code{AT} = \code{AR} (such that \code{AR} is either skip or a skip equivalent by induction hypothesis), we have shown this case to hold. 
\end{case}

\begin{case}{\textbf{(Parallel Composition \sos{P1a})}} \label{pp-part2-p1a}
Consider parallel composition. Let \texttt{AS} be a statement, \texttt{AP} and \texttt{AQ} be programs, and $\square$ be the tuple of initial program state environments $\allstores$. Assume a program of the form \code{(AS;AP) par AQ}. We note that the corresponding form \code{AQ par (AS;AP)} follows accordingly (see Case~\ref{pp-part2-p2a}) and that any number of nested parallel statements could be used, with matching uses of \sos{P1a}/\sos{P2a} reducing this into a program of the above form. Assume the following uniform execution with the first transition via \sos{P1a} $$ \begin{aligned} & (\code{((AS;AP) par AQ)} \mid \square) \atran{m} (\code{((AS$''$;AP) par AQ)} \mid \square'') \\ & \phantom{} \atran{}_s^* (\code{((AS$'$;AP) par AQ)} \mid \square'')  \atran{\circ}^* (\code{skip I} \mid \square') \end{aligned} $$ for statements \code{AS$''$} and \code{AS$'$}, and program states $\square''$ and $\square'$.

We need to show that there exists an execution $$(\func{inv$^+$}{\code{AS;AP par AQ}} \mid \square_1') \rtran{\circ}^* (\code{AT$''$} \mid \square_1'') \rtran{m} (\code{AT$'$} \mid \square_1) \rtran{}_s^* (\code{AT} \mid \square_1) $$ for some statements \code{AT$''$}, \code{AT$'$}, \code{AT}, and program states $\square_1'$, $\square_1''$, $\square_1$ such that $\square_1' \approx \square'$, $\square_1'' \approx \square''$ and $\square_1 \approx \square$. From \ref{inv+:sc} and \ref{inv+:pc} in Figure~\ref{func-def:inv+}, we have that \func{inv$^+$}{\code{AS;AP par AQ}} = \func{inv}{\code{AP}};\func{inv$^+$}{\code{AS}} \code{par} \func{inv$^+$}{\code{AQ}}.   

From our assumed transition sequence, the execution $ (\code{((AS$'$;AP) par AQ)} \mid \square'')  \atran{\circ}^* (\code{skip I} \mid \square') $ must be shorter than our original and must begin with an identifier step (since the overall execution is uniform meaning all available skip steps must be performed in the previous transition). This means application of the induction hypothesis of Part~2 of Lemma~\ref{lem:pp} on this shorter execution gives us $$ (\func{inv}{\code{AP}};\func{inv$^+$}{\code{AS$'$}} \code{ par } \func{inv$^+$}{\code{AQ}} \mid \square_1') \rtran{\circ}^* (\code{AT$'''$} \mid \square_1'')  $$ by a sequence of rules SR, for some program \code{AT$'''$} and program states $\square_1'$, $\square_1$ such that $\square_1' \approx \square'$ and $\square_1 \approx \square$. From the definition of \sos{P1a}, the first transition is via an underlying identifier rule \sos{R}. Each case of \sos{R} is considered with the proof of the Statement Property (Lemma~\ref{lem:sp}), with the format of \code{AS}, \code{AS$'$} and \code{AS$''$} (which we must consider to be within the necessary context) shown along with $\square''$. 

Since \code{AT$'''$} is returned via the induction hypothesis, we know its format. With \code{AP} guaranteed to be a complete program, inversion will reach \code{skip I$_4$} for some \code{I$_4$}. \code{AQ} is either a complete or a partial program meaning inversion will either reach \code{skip I$_5$} or a skip equivalent program \code{AQ$_1$}. Therefore \code{AT$'''$} = \func{inv$^+$}{\code{AS$'$}} \code{ par skip I$_5$} or \code{AT$'''$} = \func{inv$^+$}{\code{AS$'$}} \code{ par AQ$_1$} where \code{AQ$_1$} is a skip equivalent.

Now compare the programs \func{inv}{\code{AP}};\func{inv$^+$}{\code{AS}} \code{par} \func{inv}{\code{AQ}} and \func{inv}{\code{AP}};\func{inv$^+$}{\code{AS$'$}} \code{par} \func{inv}{\code{AQ}}. Since both executions begin with the same prefix, by the same sequence of rule applications SR we obtain $$ (\func{inv}{\code{AP}};\func{inv$^+$}{\code{AS$'$}} \code{ par } \func{inv}{\code{AQ}} \mid \square_1') \rtran{\circ}^* (\code{AT$''$} \mid \square_1) $$ for a program \code{AT$''$}. Drawing conclusions as above, we have that \code{AT$''$} = \code{AS$_2$ par skip I$_5$} or \code{AT$''$} = \code{AS$_2$ par AQ$_1$} for a statement \code{AS$_2$}, such that \code{AS$_2$} is equal to \code{AS$_1$} but with further execution allowed. The format of \code{AS$_2$} depends on the underlying rule \sos{R}, with each case having been considered in the proof of the Statement Property (Lemma~\ref{lem:sp}). From this, each case shows the format of \code{AS$_2$} (from which we know the execution must continue), and details that the only available next step of execution will be via the corresponding inverse identifier rule. From each case, we can see that the execution $$ (\code{AT$''$} \mid \square_1'') \rtran{m} (\code{AT$'$} \mid \square_1) \rtran{}_s^* (\code{AT} \mid \square_1) $$ exists for some programs \code{AT$'$} and \code{AT} (without the context of the parallel statement), and program state $\square_1$ such that $\square_1 \approx \square$, as required. Therefore this case holds. 
\end{case}

\begin{case}{\textbf{(Parallel Composition \sos{P2a})}} \label{pp-part2-p2a}
Consider an identifier step from the right hand side of a parallel statement. This follows correspondingly to Case~\ref{pp-part2-p1a} and so is omitted.
\end{case}

We have now shown all inductive cases to be valid for Part~2 of Lemma~\ref{lem:pp}. Together with the proof of Part~1, we have therefore proved Lemma~\ref{lem:pp} to be valid, as required.
\end{proof}

%% file: main.bbl
\begin{thebibliography}{10}

\bibitem{AubertMedic2021}
Cl{\'{e}}ment Aubert and Doriana Medic.
\newblock Explicit identifiers and contexts in reversible concurrent calculus.
\newblock In Shigeru Yamashita and Tetsuo Yokoyama, editors, {\em Reversible
  Computation}, volume 12805 of {\em LNCS}, pages 144--162. Springer, 2021.

\bibitem{CB1973}
C.~H. {Bennett}.
\newblock Logical reversibility of computation.
\newblock {\em IBM Journal of Research and Development}, 17(6):525--532, 1973.

\bibitem{backtracking}
Jan~A. Bergstra, Alban Ponse, and Jos van Wamel.
\newblock Process algebra with backtracking.
\newblock In J.~W. de~Bakker, Willem~P. de~Roever, and Grzegorz Rozenberg,
  editors, {\em A Decade of Concurrency, Reflections and Perspectives, {REX}
  School/Symposium}, volume 803 of {\em LNCS}, pages 46--91. Springer, 1993.

\bibitem{CC1999}
Christopher~D. Carothers, Kalyan~S. Perumalla, and Richard Fujimoto.
\newblock {E}fficient optimistic parallel simulations using reverse
  computation.
\newblock {\em {ACM} Transactions on Modeling and Computer Simulation},
  9(3):224--253, 1999.

\bibitem{DC2016}
Davide Cingolani, Mauro Ianni, Alessandro Pellegrini, and Francesco Quaglia.
\newblock Mixing hardware and software reversibility for speculative parallel
  discrete event simulation.
\newblock In Simon~J. Devitt and Ivan Lanese, editors, {\em Reversible
  Computation}, volume 9720 of {\em LNCS}, pages 137--152. Springer, 2016.

\bibitem{DC2017}
Davide Cingolani, Alessandro Pellegrini, and Francesco Quaglia.
\newblock Transparently mixing undo logs and software reversibility for state
  recovery in optimistic {PDES}.
\newblock {\em ACM Transactions on Modeling and Computer Simulation},
  27(2):11:1--11:26, 2017.

\bibitem{DC2017MS}
Davide Cingolani, Alessandro Pellegrini, Markus Schordan, Francesco Quaglia,
  and David~R. Jefferson.
\newblock Dealing with reversibility of shared libraries in {PDES}.
\newblock In Wentong Cai, Yong~Meng Teo, Philip Wilsey, and Kevin Jin, editors,
  {\em {SIGSIM} Conference on Principles of Advanced Discrete Simulation},
  pages 41--52. {ACM}, 2017.

\bibitem{TH2018}
Martin~Holm Cservenka, Robert Gl{\"{u}}ck, Tue Haulund, and Torben~{\AE}gidius
  Mogensen.
\newblock {D}ata structures and dynamic memory management in reversible
  languages.
\newblock In Jarkko Kari and Irek Ulidowski, editors, {\em Reversible
  Computation}, volume 11106 of {\em LNCS}, pages 269--285. Springer, 2018.

\bibitem{VD2004}
Vincent Danos and Jean Krivine.
\newblock Reversible communicating systems.
\newblock In Philippa Gardner and Nobuko Yoshida, editors, {\em Concurrency
  Theory}, volume 3170 of {\em LNCS}, pages 292--307. Springer, 2004.

\bibitem{JE2012}
Jakob Engblom.
\newblock A review of reverse debugging.
\newblock In {\em 2012 System, Software, SoC and Silicon Debug Conference},
  2012.

\bibitem{GF2021}
Giovanni Fabbretti, Ivan Lanese, and Jean-Bernard Stefani.
\newblock Causal-consistent debugging of distributed {Erlang} programs.
\newblock In Shigeru Yamashita and Tetsuo Yokoyama, editors, {\em Reversible
  Computation}, volume 12805 of {\em LNCS}, pages 79--95. Springer, 2021.

\bibitem{MF2018}
Michael~P. Frank.
\newblock Physical foundations of {L}andauer's principle.
\newblock In Jarkko Kari and Irek Ulidowski, editors, {\em Reversible
  Computation}, volume 11106 of {\em LNCS}, pages 3--33. Springer, 2018.

\bibitem{MPF1999}
Michael~Patrick Frank and Thomas~F Knight~Jr.
\newblock {\em Reversibility for efficient computing}.
\newblock PhD thesis, Massachusetts Institute of Technology, {USA}, 1999.

\bibitem{RF1999}
Richard~M. Fujimoto.
\newblock {\em Parallel and Distribution Simulation Systems}.
\newblock Wiley, 1st edition, 1999.

\bibitem{caredeb}
Elena Giachino, Ivan Lanese, and Claudio~Antares Mezzina.
\newblock {CaReDeb} website.
\newblock URL: \url{http://www.cs.unibo.it/caredeb}.

\bibitem{RG2016}
Robert Gl{\"{u}}ck and Tetsuo Yokoyama.
\newblock A linear-time self-interpreter of a reversible imperative language.
\newblock {\em Computer Software}, 33:108--128, 2016.

\bibitem{RG2017}
Robert Gl{\"{u}}ck and Tetsuo Yokoyama.
\newblock A minimalist's reversible while language.
\newblock {\em IEICE Transactions on Information and Systems},
  100-D(5):1026--1034, 2017.

\bibitem{NH1972}
A.~Nico Habermann.
\newblock Parallel neighbor-sort (or the glory of the induction principle),
  1972.

\bibitem{TH2017A}
Tue Haulund.
\newblock {D}esign and implementation of a reversible object-oriented
  programming language.
\newblock {\em Computing Research Repository}, abs/1707.07845, 2017.

\bibitem{LHS2021}
Lasse Hay-Schmidt, Robert Gl{\"u}ck, Martin~Holm Cservenka, and Tue Haulund.
\newblock Towards a unified language architecture for reversible
  object-oriented programming.
\newblock In Shigeru Yamashita and Tetsuo Yokoyama, editors, {\em Reversible
  Computation}, volume 12805 of {\em LNCS}, pages 96--106. Springer, 2021.

\bibitem{JHphd}
James Hoey.
\newblock {\em Reversing an Imperative Concurrent Programming Language}.
\newblock PhD thesis, University of Leicester, 2020.

\bibitem{JH2020}
James Hoey, Ivan Lanese, Naoki Nishida, Irek Ulidowski, and Germ{\'{a}}n Vidal.
\newblock A case study for reversible computing: Reversible debugging of
  concurrent programs.
\newblock In {\em Reversible Computation: Extending Horizons of Computing -
  Selected Results of the {COST} Action {IC1405}}, volume 12070 of {\em LNCS},
  pages 108--127. Springer, 2020.

\bibitem{JH2019}
James Hoey and Irek Ulidowski.
\newblock Reversible imperative parallel programs and debugging.
\newblock In Michael~Kirkedal Thomsen and Mathias Soeken, editors, {\em
  Reversible Computation}, volume 11497 of {\em LNCS}, pages 108--127.
  Springer, 2019.

\bibitem{JHWiP}
James Hoey and Irek Ulidowski.
\newblock Towards causal-consistent reversibility of imperative concurrent
  programs.
\newblock In {\em Reversible Computation}, LNCS. Springer, 2022.

\bibitem{HH2010}
Hans H{\"{u}}ttel.
\newblock {\em Transitions and Trees - An Introduction to Structural
  Operational Semantics}.
\newblock Cambridge University Press, 2010.

\bibitem{PAHJ2018}
Petur Andrias~H{\o}jgaard Jacobsen, Robin Kaarsgaard, and Michael~Kirkedal
  Thomsen.
\newblock Corefun: A typed functional reversible core language.
\newblock In Jarkko Kari and Irek Ulidowski, editors, {\em Reversible
  Computation}, volume 11106 of {\em LNCS}, pages 304--321. Springer, 2018.

\bibitem{RL1961}
Rolf Landauer.
\newblock Irreversibility and heat generation in the computing process.
\newblock {\em IBM Journal of Research and Development}, 5(3):183--191, 1961.

\bibitem{DM2021}
Ivan Lanese, Doriana Medic, and Claudio~Antares Mezzina.
\newblock Static versus dynamic reversibility in {CCS}.
\newblock {\em Acta Informatica}, 58(1-2):1--34, 2021.

\bibitem{cauder}
Ivan Lanese, Naoki Nishida, Adri\'an Palacios, and Germ\'an Vidal.
\newblock {CauDEr} website.
\newblock URL: \url{https://github.com/mistupv/cauder}.

\bibitem{IL2018}
Ivan Lanese, Naoki Nishida, Adri{\'{a}}n Palacios, and Germ{\'{a}}n Vidal.
\newblock {CauDEr}: {A} causal-consistent reversible debugger for {E}rlang.
\newblock In John~P. Gallagher and Martin Sulzmann, editors, {\em Functional
  and Logic Programming}, volume 10818 of {\em LNCS}, pages 247--263. Springer,
  2018.

\bibitem{IL2018B}
Ivan Lanese, Naoki Nishida, Adri{\'{a}}n Palacios, and Germ{\'{a}}n Vidal.
\newblock A theory of reversibility for {Erlang}.
\newblock {\em Journal of Logical and Algebraic Methods in Programming},
  100:71--97, 2018.

\bibitem{IL2019}
Ivan Lanese, Adri{\'{a}}n Palacios, and Germ{\'{a}}n Vidal.
\newblock Causal-consistent replay debugging for message passing programs.
\newblock In Jorge~A. P{\'{e}}rez and Nobuko Yoshida, editors, {\em Formal
  Techniques for Distributed Objects, Components, and Systems}, volume 11535 of
  {\em LNCS}, pages 167--184. Springer, 2019.

\bibitem{IL2020}
Ivan Lanese, Iain C.~C. Phillips, and Irek Ulidowski.
\newblock An axiomatic approach to reversible computation.
\newblock In Jean Goubault{-}Larrecq and Barbara K{\"{o}}nig, editors, {\em
  Foundations of Software Science and Computation Structures}, volume 12077 of
  {\em LNCS}, pages 442--461. Springer, 2020.

\bibitem{IL2021}
Ivan Lanese, Ulrik~P. Schultz, and Irek Ulidowski.
\newblock Reversible execution for robustness in embodied {AI} and industrial
  robots.
\newblock {\em {IT} Professional}, 23(3):12--17, 2021.

\bibitem{CL1986}
Christopher Lutz.
\newblock Janus: a time-reversible language.
\newblock {\it Letter to R.\ Landauer}., 1986.

\bibitem{DM2020}
Doriana Medić, Claudio~Antares Mezzina, Iain Phillips, and Nobuko Yoshida.
\newblock A parametric framework for reversible \emph{{\(\pi\)}}-calculi.
\newblock {\em Information and Computation}, 275:104644, 2020.

\bibitem{TM2015}
Torben~{\AE}gidius Mogensen.
\newblock {RSSA:} {A} reversible {SSA} form.
\newblock In {\em Perspectives of System Informatics}, volume 9609 of {\em
  LNCS}, pages 203--217. Springer, 2015.

\bibitem{TM2021}
Torben~{\AE}gidius Mogensen.
\newblock Reversible functional array programming.
\newblock In Shigeru Yamashita and Tetsuo Yokoyama, editors, {\em Reversible
  Computation}, volume 12805 of {\em LNCS}, pages 45--63. Springer, 2021.

\bibitem{KP2014}
Kalyan Perumalla.
\newblock {\em Introduction to Reversible Computing}.
\newblock CRC Press, 2014.

\bibitem{IP2007}
Iain~C.C. Phillips and Irek Ulidowski.
\newblock Reversing algebraic process calculi.
\newblock {\em Journal of Logic and Algebraic Programming}, 73(1-2):70--96,
  2007.

\bibitem{IP2012}
Iain~C.C. Phillips, Irek Ulidowski, and Shoji Yuen.
\newblock A reversible process calculus and the modelling of the {ERK}
  signalling pathway.
\newblock In {\em Reversible Computation}, volume 7581 of {\em LNCS}, pages
  218--232. Springer, 2012.

\bibitem{GP2004}
Gordon~D. Plotkin.
\newblock A structural approach to operational semantics.
\newblock {\em Journal of Logic and Algebraic Programming}, 60-61:17--139,
  2004.

\bibitem{MS2015}
Markus Schordan, David~R. Jefferson, Peter D.~Barnes Jr., Tomas Oppelstrup, and
  Daniel~J. Quinlan.
\newblock Reverse code generation for parallel discrete event simulation.
\newblock In Jean Krivine and Jean{-}Bernard Stefani, editors, {\em Reversible
  Computation}, volume 9138 of {\em LNCS}, pages 95--110. Springer, 2015.

\bibitem{MS2018}
Markus Schordan, Tomas Oppelstrup, David~R. Jefferson, and Peter D.~Barnes Jr.
\newblock Generation of reversible {C++} code for optimistic parallel discrete
  event simulation.
\newblock {\em New Generation Computing}, 36(3):257--280, 2018.

\bibitem{UPS2018}
Ulrik~Pagh Schultz.
\newblock Reversible object-oriented programming with region-based memory
  management.
\newblock In Jarkko Kari and Irek Ulidowski, editors, {\em Reversible
  Computation}, volume 11106 of {\em LNCS}, pages 322--328. Springer, 2018.

\bibitem{UPS2016}
Ulrik~Pagh Schultz and Holger~Bock Axelsen.
\newblock Elements of a reversible object-oriented language.
\newblock In Simon~J. Devitt and Ivan Lanese, editors, {\em Reversible
  Computation}, volume 9720 of {\em LNCS}, pages 153--159. Springer, 2016.

\bibitem{YT2006}
Yarong Tang, Kalyan~S. Perumalla, Richard~M. Fujimoto, Homa Karimabadi,
  Jonathan Driscoll, and Yuri Omelchenko.
\newblock Optimistic simulations of physical systems using reverse computation.
\newblock {\em {SIMULATION}}, 82(1):61--73, 2006.

\bibitem{MKT2015}
Michael~Kirkedal Thomsen and Holger~Bock Axelsen.
\newblock Interpretation and programming of the reversible functional language
  {RFUN}.
\newblock In Ralf L{\"{a}}mmel, editor, {\em Implementation and Application of
  Functional Programming Languages}, pages 8:1--8:13. {ACM}, 2015.

\bibitem{costbook}
Irek Ulidowski, Ivan Lanese, Ulrik~Pagh Schultz, and Carla Ferreira, editors.
\newblock {\em Reversible Computation: Extending Horizons of Computing -
  Selected Results of the {COST} Action {IC1405}}, volume 12070 of {\em LNCS}.
\newblock Springer, 1 edition, 2020.

\bibitem{Undo}
{Undo Software}.
\newblock Undo{DB}.
\newblock Commercially available reversible debugger.
  \url{http://undo-software.com/}.

\bibitem{UndoWhitePaper}
{Undo Software}.
\newblock Increasing software development productivity with reversible
  debugging, 2014.

\bibitem{ADV2010}
Alexis~De Vos.
\newblock {\em Reversible Computing - Fundamentals, Quantum Computing, and
  Applications}.
\newblock Wiley, 2010.

\bibitem{GV2011}
George Vulov, Cong Hou, Richard~W. Vuduc, Richard Fujimoto, Daniel~J. Quinlan,
  and David~R. Jefferson.
\newblock The backstroke framework for source level reverse computation applied
  to parallel discrete event simulation.
\newblock In Sanjay Jain, Roy Creasey, Jan Himmelspach, Preston White, and
  Michael Fu, editors, {\em Winter Simulation Conference}, pages 2965--2979.
  {IEEE}, 2011.

\bibitem{TY2011}
Tetsuo Yokoyama, Holger~Bock Axelsen, and Robert Gl{\"{u}}ck.
\newblock Towards a reversible functional language.
\newblock In Alexis~De Vos and Robert Wille, editors, {\em Reversible
  Computation}, volume 7165 of {\em LNCS}, pages 14--29. Springer, 2011.

\bibitem{TY2007}
Tetsuo Yokoyama and Robert Gl{\"{u}}ck.
\newblock A reversible programmingg language and its invertible
  self-interpreter.
\newblock In G.~Ramalingam and Eelco Visser, editors, {\em {ACM} {SIGPLAN}
  Workshop on Partial Evaluation and Semantics-based Program Manipulation},
  pages 144--153. {ACM}, 2007.

\bibitem{AZ2009}
A.~Zeller.
\newblock {\em Why Programs Fail: A Guide to Systematic Debugging, 2nd
  Edition}.
\newblock Academic Press, 2009.

\end{thebibliography}
